\documentclass[draftclsnofoot,onecolumn]{IEEEtran}

\usepackage[utf8]{inputenc}

\usepackage{graphicx}
\usepackage{amsmath, bm}
\usepackage{amsfonts}
\usepackage{amssymb}

\usepackage{cite}

\usepackage{url}
\usepackage{balance}
\usepackage{xcolor}
\usepackage[capitalise,nameinlink,noabbrev]{cleveref}
\usepackage[font=small]{caption}
\usepackage{subcaption}

\date{}

\newcommand{\comment}[1]{}
\newtheorem{theorem}{Theorem}
\newtheorem{lemma}[theorem]{Lemma}
\newtheorem{proposition}[theorem]{Proposition}
\newtheorem{corollary}[theorem]{Corollary}
\newtheorem{definition}[theorem]{Definition}
\newtheorem{remark}[theorem]{Remark}

\def\bp{{R}} 
\def\bd{{\Xi}} 
\def\br{{R}} 
\def\mc{{\mathsf{T}}}  
\def\tc{{\mathcal{T}}} 
\def\tp{{P}} 

\let\emptyset\varnothing

\begin{document}
	\IEEEoverridecommandlockouts
\title{On the Second-Order Asymptotics of\\ the Hoeffding Test and Other Divergence Tests}
\author{
	\IEEEauthorblockN{K.~V.~Harsha, Jithin~Ravi,~\IEEEmembership{Member,~IEEE}, and Tobias~Koch,~\IEEEmembership{Senior Member,~IEEE}}
  \thanks{K.~V.~Harsha is with the Department of Mathematics, 
		Gandhi Institute of Technology and Management, Hyderabad, India (e-mail: hkalluma@gitam.edu).}
  \thanks{J.~Ravi is with the Department of Electronics and Electrical Communication Engineering, 
		Indian Institute of Technology Kharagpur, India (e-mail: jithin@ece.iitkgp.ac.in).}
  \thanks{T.~Koch is with the Signal Theory and Communications Department, Universidad Carlos III de Madrid, 28911, Leganés, Spain, and also with the Gregorio Marañón Health Research Institute, 28007 Madrid, Spain (e-mail: koch@tsc.uc3m.es).}
	\thanks{Part of this work was done while K.~V.~Harsha and J.~Ravi were with the Signal Theory and Communications Department, Universidad Carlos III de Madrid, Legan\'es, Spain. K.~V.~Harsha, J.~Ravi, and T.~Koch have received funding from the European Research Council (ERC) under the European Union's Horizon 2020 research and innovation programme (Grant No.~714161). J.~Ravi has further received funding from the Science and Engineering Research Board (SERB) under Start-up Research Grant (Grant No: SRG/2022/001393) and also from IIT Kharagpur under Faculty Start-up Research Grant (FSRG). T.~Koch has further received funding from the Spanish Ministerio de Ciencia e Innovaci\'on under Grant PID2020-116683GB-C21~/~AEI~/~10.13039/501100011033. The material in this paper was presented in part at the 2022 IEEE Information Theory Workshop (ITW), Mumbai, India, November 2022, and in part at the 2024 International Zurich Seminar on Information and Communication (IZS), Zurich, Switzerland, March 2024.}
}
	\maketitle
	\begin{abstract}
		Consider a binary statistical hypothesis testing problem, where $n$ independent and identically distributed random variables $Z^n$ are either distributed according to the null hypothesis $P$ or the alternative hypothesis $Q$, and only $P$ is known. A well-known test that is suitable for this case is the so-called Hoeffding test, which accepts $P$ if the Kullback-Leibler (KL) divergence between the empirical distribution of $Z^n$  and $P$ is below some threshold. This work characterizes the first and second-order terms of the type-II error probability for a fixed type-I error probability for the Hoeffding test as well as for divergence tests, where the KL divergence is replaced by a general divergence. It is demonstrated that, irrespective of the divergence, divergence tests achieve the first-order term of the Neyman-Pearson test, which is the optimal test when both $P$ and $Q$ are known. In contrast, the second-order term of divergence tests is strictly worse than that of the Neyman-Pearson test. It is further demonstrated that divergence tests with an invariant divergence achieve the same second-order term as the Hoeffding test, but divergence tests with a non-invariant divergence may outperform the Hoeffding test for some alternative hypotheses $Q$. Potentially, this behavior could be exploited by a composite hypothesis test with partial knowledge of the alternative hypothesis $Q$ by tailoring the divergence of the divergence test to the set of possible alternative hypotheses.
\end{abstract}


\section{Introduction}
\label{sec:intro}
Consider a  binary hypothesis testing problem that decides whether a sequence of independent and identically distributed (i.i.d.)  random variables $Z^n$ is either generated from distribution $P$ or from distribution $Q$. Assume that both distributions are discrete and the hypothesis test has access to $P$ but not to $Q$. A suitable test for this case is  the well-known \emph{Hoeffding test} \cite{H65}, denoted by $\mc^{D_{\textnormal{KL}}}_n(r)$, which accepts  $P$ if $D_{\textnormal{KL}}(\tp_{Z^{n}} \|P) < r$,  for some $r>0$, and otherwise accepts $Q$.
Here, $\tp_{Z^{n}}$ is the type (the empirical distribution) of $Z^n$
 	and $D_{\textnormal{KL}}(P \| Q)$  is the Kullback-Leibler (KL) divergence between $P$ and $Q$~\cite{ATCB}.
 In this paper, we analyze the second-order performance of the Hoeffding test as well as of Hoeffding-like tests,  referred to as \emph{divergence tests}, where the KL divergence is replaced by other divergences (see Section \ref{sec:model} for a rigorous definition).
 
 We focus on the asymptotic behaviour of the type-II error $\beta_n$ (the probability of declaring hypothesis $P$ under hypothesis $Q$) for a fixed type-I error $\alpha_n$ (the probability of declaring hypothesis $Q$ under hypothesis $P$).
When both $P$ and $Q$ are known, the  optimal test is the \emph{likelihood ratio test}, also known as the \emph{Neyman-Pearson test}, denoted by $\mc^{\textnormal{NP}}_n$. For this test, the smallest type-II error $\beta_n$ for which $\alpha_n \leq \epsilon$
satisfies \cite[Prop.~2.3]{T214} 
\begin{equation}
-\ln \beta_{n} =nD_{\textnormal{KL}}(P \| Q) -  \sqrt{nV_{\textnormal{KL}}(P \| Q) }  \mathsf{Q}_{\mathcal{N}}^{-1}(\epsilon) + O(\ln n) \label{eq:NP}
\end{equation}
as $n \rightarrow \infty$, where the KL divergence $D_{\textnormal{KL}}(P \| Q) $  is defined as
\begin{equation}
	D_{\textnormal{KL}}(P \| Q)   \triangleq \sum_{i=1}^{k} P_{i} \ln \frac{P_{i}}{Q_{i}}; \label{eq:kl}
\end{equation}
$\mathsf{Q}_{\mathcal{N}}^{-1}(\cdot)$ denotes the inverse of the tail probability of the standard Normal distribution;
\begin{equation}
\label{eq:divergence_var}
V_{\textnormal{KL}}(P \| Q) \triangleq \sum_{i=1}^{k} P_{i} \left[ \left( \ln  \frac{P_{i}}{Q_{i}} -D_{\textnormal{KL}}(P \|Q) \right)^{2}  \right]
\end{equation}
denotes the KL divergence variance; and $O(\ln n)$ denotes terms that grow at most as fast as $\ln n$ as $n \to \infty$. In \eqref{eq:kl} and \eqref{eq:divergence_var}, $P_i$ and $Q_i$ denote the $i$-th components of $P$ and $Q$, and $k$ denotes their dimension.
 By inspecting the expansion of $-\ln \beta_n$ in \eqref{eq:NP}, one can define the first-order term 
 \begin{equation}
    \label{eq_beta'_g}
     \beta'\triangleq \lim_{n\rightarrow \infty} \frac{-\ln \beta_n}{n}
\end{equation}
and the second-order term
\begin{equation}
    \label{eq_beta"_g}
    \beta''\triangleq \lim_{n\rightarrow \infty} \frac{-\ln \beta_n -n \beta'} {\sqrt{n}}
\end{equation}
if the limits exist.
For the Neyman-Pearson test, we then have
\begin{equation}
    \beta'_{\textnormal{NP}} = D_{\textnormal{KL}}(P\|Q)
\end{equation}
and
\begin{equation}
    \label{eq:beta''_NP}
    \beta_{\textnormal{NP}}'' = -\sqrt{V_{\textnormal{KL}}(P \| Q) }  \mathsf{Q}_{\mathcal{N}}^{-1}(\epsilon).
\end{equation}
It was shown in \cite{H65} that the first-order term $\beta_{D_{\textnormal{KL}}}'$ of the Hoeffding test is also $D_{\textnormal{KL}}(P\|Q)$. Thus, the Hoeffding test is first-order optimal. 
 
 In this paper, we analyze the second-order performance of the Hoeffding test as well as of divergence tests. Specifically, a divergence test $\mc^{D}_n(r)$ accepts $P$ if $D(\tp_{Z^{n}} \|P) < r$,  for some $r>0$ and an arbitrary divergence $D$, and otherwise accepts $Q$. Consequently, the divergence test $\mc^{D}_n(r)$ includes the Hoeffding test as a special case when $D=D_{\textnormal{KL}}$.

\comment{
Recently, we have demonstrated \cite{HJT_ITW22} that the second-order term of the Hoeffding test is $\beta''=-\sqrt{V(P \| Q) \mathsf{Q}^{-1}_{\chi^{2}_{k-1}}(\epsilon) }$, where $\mathsf{Q}^{-1}_{\chi^{2}_{k-1}}(\cdot)$ denotes the inverse of the tail probability of the chi-square distribution with $k-1$ degrees of freedom. Since $ \sqrt{ \mathsf{Q}^{-1}_{\chi^{2}_{k-1}}(\epsilon) }>   \mathsf{Q}^{-1}(\epsilon)$, it follows that the second-order performance of the Hoeffding test is worse than that of the Neyman-Pearson test.

In this paper, we analyze the second-order performance of the divergence test  $\mc^{D}$, which accepts $P$ if $D(T_{Z^{n}} \|P) < c$,  for some $c>0$, and otherwise accepts $Q$. The divergence $D$ of the divergence test $\mc^{D}$ is arbitrary, so $\mc^{D}$ includes the Hoeffding test as a special case when $D=D_{\textnormal{KL}}$. We demonstrate that the divergence test $\mc^{D}$ achieves the same first-order term $\beta'$ as the Neyman-Pearson test, irrespective of the divergence $D$. Hence, $\mc^{D}$ is first-order optimal for every divergence $D$. We further demonstrate that, for the class of \emph{invariant divergences} \cite{CC82}, which includes the R\'enyi divergence and the f-divergence (and, hence, also the KL divergence), the divergence test $\mc^D$ achieves the same second-order term $\beta''$ as the Hoeffding test. In contrast, we show that a divergence test $\mc^{D}$ with a non-invariant divergence may achieve a second-order term $\beta''$ that is strictly better than that of the Hoeffding test for some $Q$ and $\epsilon$. 

}

\subsection{Related Work}
\label{sub:related}
The considered hypothesis testing problem falls under the category of \emph{composite hypothesis testing} \cite{FM02}. Indeed, in composite hypothesis testing, the test has no access to the null hypothesis $P$ and the alternative hypothesis $Q$, but it has the knowledge that $P$ and $Q$ belong to the sets of distributions $\mathcal{P}$ and $\mathcal{Q}$, respectively. Our setting corresponds to the case where $\mathcal{P}=\{P\}$ and $\mathcal{Q}=\mathcal{P}^c$ (where we use the notation $\mathcal{A}^c$ to denote the complement of a set $\mathcal{A}$).

The Hoeffding test is a particular instance of the \emph{generalized likelihood-ratio test (GLRT)} \cite{VT68}, which is arguably the most common test used in composite hypothesis testing. A useful benchmark for the Hoeffding test is the Neyman-Pearson test, which is the optimal test when both $P$ and $Q$ are known. As indicated before, the Hoeffding test achieves the same first-order term $\beta'$ as the Neyman-Pearson test, both in \emph{Stein's regime}, where the type-I error satisfies $\alpha_n \leq \epsilon$, as well as in the \emph{Chernoff regime}, where $\alpha_n \leq e^{-n\gamma}$, $\gamma >0$; see, e.g., \cite{H65,Zeitouni92,Unnikrishnan11,BF22-1,BF22-2}. Thus, the first-order term of the Neyman-Pearson test can be achieved without having access to the distribution $Q$ of the alternative hypothesis. However, not having access to $Q$ negatively affects higher-order terms. For example, for a given threshold $r$, the type-I error of the Hoeffding test satisfies \cite[Eq.~(10)]{BF22-2}
\begin{equation}
\alpha_n = n^{\frac{k-3}{2}}e^{-nr}(c'+o(1))
\end{equation}
whereas for the corresponding Neyman-Pearson test \cite[Eq.~(9)]{BF22-2}
\begin{equation}
\alpha_n = n^{-\frac{1}{2}}e^{-nr}(c+o(1)).
\end{equation}
Here, $c$ and $c'$ are constants that only depend on $P$, $Q$, and $r$. Moreover, it was demonstrated by Unnikrishnan \emph{et al.} \cite{Unnikrishnan11} that the variance of the normalized Hoeffding test statistic $n D_{\textnormal{KL}}(\tp_{Z^n}\|P)$ converges to $\frac{1}{2}(k-1)$ as $n\to\infty$. Both results suggest that, for moderate $n$, the Hoeffding test scales unfavorably with the cardinality of $P$ and $Q$. This was the motivation of \cite{Unnikrishnan11} to propose the \emph{test via mismatched divergence}. As we shall see, our characterization of the second-order term of the divergence test (and, hence, also the Hoeffding test) also suggests an unfavorable scaling with the cardinality of $P$ and $Q$.

Our setting where $\mathcal{P}=\{P\}$ and $\mathcal{Q}=\mathcal{P}^c$ was also studied by Watanabe \cite{Watanabe18}, who focused on the second-order regime of statistical tests that do not depend on $Q \in \mathcal{Q}$. The test proposed in \cite{Watanabe18} achieves a second-order term which was shown to be \emph{Pareto-optimal}. The related case where only training sequences are available for both $P$ and $Q$ was considered by Gutman \cite{Gut89}. The test proposed in \cite{Gut89} was later shown to be second-order optimal by Li and Tan \cite{YTan20}. 

\subsection{Contributions and Outline}

The contributions of this paper can be summarized as follows:
\begin{enumerate}
\item We characterize the first-order and second-order terms of the divergence test $\mc^{D}_n(r)$ for an arbitrary divergence $D$ (Theorem~\ref{eq:NP}). \emph{Inter alia}, the proof of this result requires the generalization of a well-known result by Wilks~\cite{W938} regarding the convergence of the KL divergence between the type $\tp_{Z^n}$ and the null hypothesis $P$ to general divergences (Lemma~\ref{convgdiv}). Our characterization demonstrates that the divergence test $\mc^{D}_n(r)$ achieves the same first-order term $\beta'$ as the Neyman-Pearson test, irrespective of the divergence $D$. Hence, $\mc^{D}_n(r)$ is first-order optimal for every divergence $D$. This is consistent with the result that the Hoeffding test is first-order optimal \cite{H65}.
\item Specializing Theorem~\ref{eq:NP} to the case where the divergence $D$ is the KL divergence $D_{\textnormal{KL}}$, we obtain that the second-order term of the Hoeffding test is
\begin{equation}
\label{eq:beta''_KL}
\beta_{D_{\textnormal{KL}}}'' = -\sqrt{V_{\textnormal{KL}}(P \| Q) \mathsf{Q}^{-1}_{\chi^{2}_{k-1}}(\epsilon) } 
\end{equation}
where $ \mathsf{Q}^{-1}_{\chi^{2}_{k-1}}(\cdot)$ is the inverse of the tail probability of the chi-square distribution with $k-1$ degrees of freedom.\footnote{The first-order and second-order terms of the Hoeffding test where also presented by Watanabe in his presentation at the \emph{Bombay Information Theory Seminar (BITS) 2018
} \cite{WatanabeBITS18}.} We show that $\mathsf{Q}^{-1}_{\chi^{2}_{k-1}}(\epsilon)$ is monotonically increasing in $k$, and its square root is strictly larger than $\mathsf{Q}_{\mathcal{N}}^{-1}(\epsilon)$. By comparing \eqref{eq:beta''_KL} to the second-order term \eqref{eq:beta''_NP} of the Neyman-Pearson test, we thus again observe that the Hoeffding test scales unfavorably with the cardinality of $P$ and $Q$, and its second-order performance is worse than that of the Neyman-Pearson test.
\item We demonstrate that the divergence test $\mc^{D}_n(r)$ achieves the same second-order term as the Hoeffding test whenever the divergence $D$ is \emph{invariant} (see Definition~\ref{invariancediv}). The class of invariant divergences is large and includes the R\'enyi divergence and the f-divergences, which in turn include the KL divergence. In contrast, when the divergence $D$ is \emph{non-invariant}, the second-order term $\beta''_D$ of the divergence test differs from the second-order term $\beta_{D_{\textnormal{KL}}}''$  of the Hoeffding test.
\item In Section~\ref{sec:computations}, we numerically compare the second-order terms of the Hoeffding test $\mc^{D_{\textnormal{KL}}}_n(r)$ and of a divergence test $\mc^{D}_n(r)$ with a non-invariant divergence $D$. Our results show that the divergence test $\mc^{D}_n(r)$ may achieve a second-order term $\beta''_{D}$ that is strictly larger than the second-order term $\beta_{D_{\textnormal{KL}}}''$ of the Hoeffding test $\mc^{D_{\textnormal{KL}}}_n(r)$ for some alternative hypothesis $Q$. The set of distributions $Q$ for which the divergence test outperforms the Hoeffding test typically depends on the null hypothesis $P$ and on $\epsilon$. Potentially, this behavior could be exploited by a composite hypothesis test with partial knowledge of the alternative hypothesis $Q$ by tailoring the divergence $D$ of the divergence test $\mc^D_n(r)$ to the set of possible alternative hypotheses $\mathcal{Q}$.

\end{enumerate}

The rest of the paper is organized as follows. The problem setting and other definitions are introduced in Section~\ref{sec:model}. The main results are presented in Section~\ref{sec:results}. The second-order performances of various divergence tests are compared by means of numerical results in Section~\ref{sec:computations}. The proofs of our main results can be found in Section~\ref{Sec_Pf_Main}, with some of the details deferred to the appendices. The paper concludes with a summary of our results in Section~\ref{sec:conclusion}.

 \subsection{Notation}
We denote random variables by uppercase letters, such as $Z$, and their realizations by lowercase letters, such as $z$. Sequences of random variables and their realizations are denoted as $Z^n=(Z_1,\ldots,Z_n)$ and $z^n=(z_1,\ldots,z_n)$, respectively. Sets are denoted by calligraphic letters, such as $\mathcal{Z}$, and their complements are denoted with $(\cdot)^c$, such as $\mathcal{Z}^c$.

We next introduce  the order notation we shall use. Let $f(x)$ and $g(x)$ be  two real-valued functions. For \mbox{$a \in \mathbb{R} \cup \{\infty\}$}, we write  $f(x)=O(g(x))$ as $x \rightarrow a$ if $\limsup_{x \to a} \frac{|f(x)|}{|g(x)|} < \infty$.  Similarly, we  write  $f(x)=o(g(x))$ as $x \rightarrow a$ if $\lim_{x \to a} \frac{|f(x)|}{|g(x)|} = 0$. Finally, we write $f(x) = \Theta(g(x))$ as $x \to a$	 if $f(x)=O(g(x))$ and $\liminf_{x \to a} \frac{|f(x)|}{|g(x)|} > 0$.

Let $\{a_{n}\}$ be a sequence of positive real numbers. We say that a sequence of random variables $\{X_n\}$ is  \mbox{$X_{n}=o_{P}(a_{n})$} if, for every $\epsilon >0$, we have 
\begin{equation}
	\lim_{n \rightarrow \infty} P\left(   \left|  \frac{X_{n}}{a_{n}}\right| >\epsilon \right) =0.
\end{equation}
Similarly, we say that $X_{n}=O_{P}(a_{n})$ if, for every $\epsilon >0$, there exist constants $M_{\epsilon}>0$ and $N_{\epsilon} \in \mathbb{N}$ (possibly depending on $\epsilon$) such that 
\begin{equation}
	P\left(   \left|  \frac{X_{n}}{a_{n}}\right|  >M_{\epsilon} \right) \leq \epsilon,\quad  n \geq N_{\epsilon}.
\end{equation}

Convergence in distribution is denoted by $\xrightarrow{d}$. We use the notation $\mathcal{N}(\bm{\mu}, \bm{\Sigma})$ to either denote a Gaussian random vector with mean $\bm{\mu}$ and covariance matrix $\bm{\Sigma}$ or to denote its distribution.

\section{ Divergence and Divergence Test} 
\label{sec:model}
\subsection{Problem Setting} 
Let us consider a  random variable $Z$ that takes value in a discrete set  $\mathcal{Z}=\{ a_{1}, \ldots, a_{k}\}$, where $ k \geq 2$ denotes the cardinality of $\mathcal{Z}$. We denote the probability distribution of $Z$ by a $k$-length vector $\bp=(\bp_{1},\ldots,\bp_{k})^{\mathsf{T}}$, where
\begin{equation*}
	\bp_{i} \triangleq \text{Pr}\{Z=a_{i}\}, \quad i=1,\ldots,k.
\end{equation*}
Let 
\begin{equation}
	\overline{\mathcal{P}}(\mathcal{Z}) \triangleq \left\lbrace \bp=(\bp_{1}, \ldots, \bp_{k})^{\mathsf{T}} \colon \bp_{i} \geq 0, \; i=1,\ldots,k  \mbox{ and }  \sum_{i=1}^{k} \bp_{i}=1 \right\rbrace 
\end{equation}
be the set of all such probability distributions.  Similarly, let
\begin{equation}
	\mathcal{P}(\mathcal{Z}) \triangleq \left\lbrace \bp=(\bp_{1}, \ldots, \bp_{k})^{\mathsf{T}} \colon\bp_{i} > 0, \; i=1,\ldots,k  \mbox{ and }  \sum_{i=1}^{k} \bp_{i}=1 \right\rbrace
\end{equation}
be the set of all positive probability distributions on $\mathcal{Z}$. We denote the set of probability distributions on the boundary of $\mathcal{P}(\mathcal{Z})$ by $ \partial  \mathcal{P}(\mathcal{Z})$, i.e., $ \partial  \mathcal{P}(\mathcal{Z}) = \overline{\mathcal{P}}(\mathcal{Z}) \setminus  \mathcal{P}(\mathcal{Z})$.

We consider a binary hypothesis testing problem with null hypothesis $H_0$ and alternative hypothesis $H_1$.
We assume that, under hypothesis $H_0$,  the sequence of observations $Z^{n}$  is i.i.d.\ according to $P \in  \mathcal{P}(\mathcal{Z})$; under hypothesis $H_1$, the sequence of observations $Z^n$ is i.i.d. according to $Q$, where $Q \in  \mathcal{P}(\mathcal{Z})\setminus \{ P\} $.  Any statistical test, denoted by  $\mc_{n}$, for testing $H_{0}$ against $H_{1}$ is defined by choosing a subset $	\mathcal{A}_{n} \subseteq \mathcal{Z}^{n}$, called the \textit{acceptance region}, such that if $z^{n} \in 	\mathcal{A}_{n}$, then the null hypothesis $P$ is accepted, else the alternative hypothesis $Q$ is accepted. For the test $\mc_{n}$, the type-I error,  denoted by $\alpha_{n}(\mc_{n})$, and the type-II error,  denoted by $\beta_{n}(\mc_{n})$, are defined as follows: 
\begin{IEEEeqnarray}{lCl}
	\alpha_{n}(\mc_{n}) & \triangleq & P^{n} \left(   \mathcal{A}_{n}^{c} \right) \label{alpha_mc}\\
	\beta_{n}(\mc_{n}) &  \triangleq & Q^{n} ( \mathcal{A}_{n}) \label{eq:beta_mc}
\end{IEEEeqnarray}
where $P^{n}(z^{n}) = \prod_{i=1}^{n} P(z_i) $ and $Q^{n}(z^{n}) = \prod_{i=1}^{n} Q(z_i) $. 

For any given test $\mc_{n}$, our goal is to analyze the asymptotic behavior of the type-II error when the type-I error satisfies   $\alpha_{n}(\mc_{n}) \leq \epsilon, \; 0<\epsilon<1$. Following the discussion  in Section~\ref{sec:intro},  we define the first-order term  $\beta'$ and the second-order term $\beta''$ for $\mc_{n}$ by the respective \eqref{eq_beta'_g} and \eqref{eq_beta"_g}, namely,
\begin{IEEEeqnarray}{lCl}
	\beta'&\triangleq & \lim\limits_{n \rightarrow \infty} -\frac{1}{n}  \ln \beta_{n}  \\
	\beta''&\triangleq &  \lim\limits_{n \rightarrow \infty} \frac{-\ln \beta_n -n \beta'} {\sqrt{n}}   \
\end{IEEEeqnarray}
if the limits exist. Thus,  the type-II error satisfies the following asymptotic expansion:
\begin{equation}
	-\ln \beta_{n} (\mc_{n})  = n\beta'  + \sqrt{n} \beta''  + o(\sqrt{n}), \quad \text{as} \quad n \rightarrow \infty. \label{eq:so_gent}
\end{equation}
It follows from \eqref{eq:NP} that, for the Neyman-Pearson test $\mc^{\textnormal{NP}}_n$ satisfying $\alpha_{n}(\mc^{\textnormal{NP}}_n) \leq \epsilon$, 
the first-order term  $\beta'_{\textnormal{NP}}$ and the second-order term $\beta''_{\textnormal{NP}}$  are given by
\begin{IEEEeqnarray}{lCl}
\beta'_{\textnormal{NP}}& =&   D_{\textnormal{KL}}(P\|Q)\\
	\beta''_{\textnormal{NP}}&= &  - \sqrt{V_{\textnormal{KL}}(P \| Q) }  \mathsf{Q}_{\mathcal{N}}^{-1}(\epsilon). \label{eq:secNP}
\end{IEEEeqnarray}

In this paper, we consider tests $\mc_{n}$ that do not depend on the alternative hypothesis $Q$. An important special case of such tests is the Hoeffding test \cite{H65}, which depends only on the empirical distribution $\tp_{Z^n}$ (type distribution) of the observations $Z
^n$ and the null hypothesis $P$. Specifically, the Hoeffding test, denoted by $\mc^{D_{\textnormal{KL}}}_n(r)$, accepts  $P$ if $D_{\textnormal{KL}}(\tp_{Z^{n}} \|P) < r$,  for some $r>0$, and otherwise accepts $Q$. Here, $D_{\textnormal{KL}} $ is the KL-divergence and \mbox{$\tp_{z^{n}} =(\tp_{z^{n}}(a_{1}), \ldots,  \tp_{z^{n}}(a_{k}))^{\mathsf{T}}$} is the type distribution defined as
\begin{equation}
	\tp_{z^{n}}(\zeta) \triangleq \frac{1}{n} \sum_{\ell=1}^{n} \mathbf{1}\{ z_{\ell}=\zeta\}, \quad \zeta\in\mathcal{Z}
\end{equation}
where $\mathbf{1}\{\cdot\}$ denotes the indicator function. For the Hoeffding test, it was shown that the first-order term $\beta'_{D_{\textnormal{KL}}}$ is given by $D_{\textnormal{KL}}(P\|Q)$ \cite{H65}. Thus, the Hoeffding test achieves the same first-order term as the Neyman-Pearson test without requiring knowledge of the alternative distribution $Q$. This motivates us to consider a generalization of the Hoeffding test, referred to as \textit{divergence test}, which is obtained by replacing the KL-divergence in the Hoeffding test by an arbitrary divergence.

\subsection{Divergence}
Given any two probability distributions $T, R \in \mathcal{P}(\mathcal{Z})$, one can define a  non-negative  function $D(T \| R) $, called a {\em divergence}, which represents  a measure of discrepancy between them. A divergence  is not necessarily  symmetric in its arguments and also need not satisfy the triangle inequality.

Mathematically, a divergence is defined on smooth manifolds, referred to as statistical manifolds in the information geometry literature; see \cite{AB216} for more details. In particular, the set $ \mathcal{P}(\mathcal{Z})$ is a smooth manifold which possesses a global coordinate system, say  $\bm{\theta}=(\theta_{1},\ldots,\theta_{k-1})^{\mathsf{T}} \in \bd $, such that every probability vector $\bp \in \mathcal{P}(\mathcal{Z})$ is mapped to $\bm{\theta}_{\bp} = (\theta_{1,R},\ldots,\theta_{k-1,R})^{\mathsf{T}}\in  \bd $, where $ \bd$ is an open subset of $\mathbb{R}^{k-1}$.  In this paper, we choose the coordinate of a probability distribution $\bp\in\mathcal{P}(\mathcal{Z})$ as its first $(k-1)$ components.  That is,  the coordinates of $\bp$ are given by   $\bm{\theta}_{\bp} =\mathbf{\bp} \triangleq (\bp_{1}, \ldots, \bp_{k-1})^{\mathsf{T}} $ and take value in the open subset
\begin{equation}
	\bd  \triangleq \left\lbrace \mathbf{\bp}=(\bp_{1}, \ldots, \bp_{k-1})^{\mathsf{T}} \colon \bp_{i} >0, \; i=1,\ldots,k-1  \mbox{ and }  \sum_{i=1}^{k-1} \bp_{i} <1 \right\rbrace. \label{eq:domaindef}
\end{equation}
Any function of $\bp$ can then be thought of as a function of the coordinate $\mathbf{\bp}$.

A formal definition of a divergence on $\mathcal{P}(\mathcal{Z})$ is as follows:
\begin{definition} \label{divdef} 
	Consider two distributions $T$ and $R$ in $ \mathcal{P}(\mathcal{Z})$ with coordinates $\mathbf{T}$ and $\mathbf{R}$. A {\em divergence}  \mbox{$D: \mathcal{P}(\mathcal{Z}) \times  \mathcal{P}(\mathcal{Z}) \rightarrow [0, \infty) $} between $T$ and $R$, denoted by  $ D(T \| R) $ or $D(\mathbf{T}\|\mathbf{R})$, is a smooth function of $\mathbf{T}$ and $\mathbf{R}$ satisfying the following conditions: 
	\begin{enumerate}
		\item $ D(T \| R) \geq 0$ for every $T, R \in \mathcal{P}(\mathcal{Z})$. 
		\item \label{div_cond2} $D(T \| R) =0$ if, and only if, $T=R$.
		\item   When $\mathbf{T}=\mathbf{R}+\bm{\varepsilon}$, the Taylor expansion of $D$ satisfies
		\begin{equation}
			D(\mathbf{R} +\bm{ \varepsilon} \| \mathbf{R}   ) =\frac{1}{2} \sum_{i,j =1}^{k-1} g_{ij}(\mathbf{R} )  \varepsilon_{i} \varepsilon_{j} +O(\| \bm{\varepsilon} \|_{2}^{3}), \quad \text{as } \| \bm{\varepsilon} \|_{2} \rightarrow 0 \label{eq:divdef1}
		\end{equation}
		for some  $(k-1) \times (k-1)$-dimensional positive-definite matrix $G(\mathbf{R})=[g_{ij}(\mathbf{R} )]$ that depends on $\mathbf{R}$ and $\bm{\varepsilon}=(\varepsilon_{1}, \ldots,\varepsilon_{k-1})^{\mathsf{T}}$. In \eqref{eq:divdef1}, $\| \bm{\varepsilon} \|_{2} \triangleq \sqrt{\sum_{i=1}^{k-1} (\varepsilon_{i})^{2}}$ is the Euclidean norm of $\bm{\varepsilon}$.
		
\item \label{div_cond4} Let $R \in \mathcal{P}(\mathcal{Z})$ and let $\{T_n\} $  be a sequence of distributions in $\mathcal{P}(\mathcal{Z})$ that converges to a distribution \mbox{$T'\in \partial  \mathcal{P}(\mathcal{Z})$}. Then, $D$ satisfies
		\begin{equation}
			\liminf_{n \to \infty}	D(T_n \| R  ) >0. \label{eq:divdefexcon}
		\end{equation}
	\end{enumerate}	
\end{definition}

\begin{remark}\label{rem:divdef} We follow the definition of divergence from the  information geometry literature. In particular, according to \cite[Def.~1.1]{AB216}, a divergence must satisfy the first three conditions in Definition~\ref{divdef}. Often, the behavior of a divergence on the boundary of $\mathcal{P}(\mathcal{Z})$ is not specified. In Definition~\ref{divdef}, we add the fourth condition to treat the case of sequences of distributions $\{T_n\}$ that lie in $\mathcal{P}(\mathcal{Z})$ but converge to a distribution in the boundary set $\partial \mathcal{P}(\mathcal{Z})$. Note that the fourth condition is consistent with the first two conditions.
\end{remark}

For a divergence  $D$ and $\br \in \mathcal{P}(\mathcal{Z})$, consider the function $T \mapsto D(T \| \br)$, and let  $\bm{A}_{D, \mathbf{\br}} $ have components 
\begin{equation}
	\bm{A}_{D,\mathbf{\br}}(i,j)=  \frac{1}{2} 	\left.	\frac{\partial^{2} D(T\| \br) }{\partial T_{i} \partial  T_{j}} \right|_{T=\br}, \quad  i,j=1,\ldots, k-1.  \label{eq:matrixa}
\end{equation}
We shall refer to $\bm{A}_{D, \mathbf{\br}} $ as the matrix associated with the divergence $D$ at $\mathbf{\br}$. It follows from Definition~\ref{divdef} that $\bm{A}_{D, \mathbf{\br}}$ is a symmetric and  positive-definite matrix. It further follows from  \eqref{eq:divdef1} that, for $T \in \mathcal{P}(\mathcal{Z})$, we have
\begin{equation}
	D(T \| \br) 
	= (\mathbf{T}-\mathbf{\br})^{T} \bm{A}_{D,\mathbf{\br}} ( \mathbf{T}-\mathbf{\br})+O(\|\mathbf{T}-\mathbf{\br}\|_{2}^{3}),\quad \text{as} \quad \| \mathbf{T}-\mathbf{\br} \|_{2} \rightarrow 0. \label{eq:tayldivA}
\end{equation}
Based on $\bm{A}_{D, \mathbf{\br}}$, we can distinguish between \emph{invariant} and \emph{non-invariant divergences}.

\begin{definition}\label{invariancediv}
let $D$ be a divergence and $\br  \in \mathcal{P}(\mathcal{Z})$. Then, $D$ is said to be an \emph{invariant divergence} on $\mathcal{P}(\mathcal{Z})$ if the matrix associated with the divergence $D$ at $\mathbf{\br}$  is of the form $\bm{A}_{D, \mathbf{\br}}=\eta \bm{\Sigma}_{\mathbf{\br}}$ for a constant $\eta >0$ \comment{(possibly depending on $\mathbf{\br}$)} and a matrix   $\bm{\Sigma}_{\mathbf{\br}}$ with components
\begin{equation}
	\bm{\Sigma}_{\mathbf{\br}}(i,j)
	 =
	\begin{cases}
	\frac{1}{\br_{i}} +	\frac{1}{\br_{k}},  \quad &  \text{for $i=j$} \\
		\frac{1}{\br_{k}}, \quad & \text{for $i\neq j$.}
	\end{cases}
	  \label{eq:covsigmain1}
\end{equation}
A divergence is said to be a \emph{non-invariant divergence} if it is not invariant.
\end{definition}
\begin{remark} The notion of an invariant divergence is adapted from the notion of invariance of geometric structures in information geometry; see \cite{AB216}, \cite{CL86} for more details. The matrix $\bm{\Sigma}_{\mathbf{\br}}$ represents the unique invariant Riemannian metric, known as the Fisher information metric, in $\mathcal{P}(\mathcal{Z})$ with respect to the coordinate system $\bd$ defined in \eqref{eq:domaindef}; see \cite[Eq. (47)]{AC210}. More precisely, any divergence on a finite-dimensional manifold induces a Riemannian metric. If the induced Riemannian metric is the Fisher information metric, then such divergence is referred to as an invariant divergence in the information geometry literature. In contrast, if the induced Riemannian metric is not a constant multiple of the Fisher information metric, then the divergence is referred to as a non-invariant divergence; see \cite{CC82} for more details. \comment{Note that, in the information geometry literature, the constant $\eta$ must be independent of $\mathbf{R}$.} 
	\end{remark}

For an invariant divergence, \eqref{eq:tayldivA} becomes
\begin{equation}
	D(T \| \br)
	= \eta (\mathbf{T}-\mathbf{\br})^{T} \bm{\Sigma}_{\mathbf{\br}}( \mathbf{T}-\mathbf{\br}) + O(\| \mathbf{T}-\mathbf{\br} \|_{2}^{3}),\quad \text{as} \quad \| \mathbf{T}-\mathbf{\br}\|_{2} \rightarrow 0. \label{eq:tylpfd}
\end{equation}

\subsection{Examples of Divergences}
There are several classes of divergences that are widely used in various fields of science and engineering; see, e.g., \cite[Ch.~2]{CZPA209} for more details.  We next discuss some well-known invariant and non-invariant divergences. 
\subsubsection{Invariant divergences} An important class of invariant divergences are the \emph{$f$-divergences} defined as follows.  
Let $f:(0,\infty) \rightarrow \mathbb{R}$ be a convex function with $f(1)=0$ and $f''(1)>0$. Then, the $f$-divergence between two distributions $T$ and $R$ in $\mathcal{P}(\mathcal{Z})$ is given by\footnote{In the information theory literature, the condition that $f''(1)>0$ is often omitted; see, e.g., \cite[Ch.~4]{CS204}. However, $f''(1)>0$ is required to satisfy the third condition of Definition~\ref{divdef} which, as mentioned in Remark~\ref{rem:divdef}, follows the definition of divergence from the information geometry literature. As a consequence, there are examples of $f$-divergences in the information theory literature, such as the total variation distance, that are not divergences according to Definition~\ref{divdef}.}
\begin{eqnarray}
	D_{f}(T \| R) \triangleq \sum_{i=1}^{k} R_{i} f\left( \frac{T_{i}}{R_{i}}\right).
	\label{eq:fdiv}
\end{eqnarray}
For  $f(u)=u \log u$, the $f$-divergence $D_{f}$ is the Kullback-Leibler divergence \eqref{eq:kl}. For $f(u)=(u-1)^{2}$, the \mbox{$f$-divergence} $D_{f}$ is the $\chi^{2}$-divergence
\begin{equation}
	d_{\chi^{2}} (T, R) \triangleq \sum_{i=1}^{k}  \frac{(T_{i} -R_{i})^{2}}{R_{i}}= (\mathbf{T}-\mathbf{\br})^{T} \bm{\Sigma}_{\mathbf{\br}}( \mathbf{T}-\mathbf{\br}). \label{eq:chisddef}
\end{equation}
For $f(u)= \frac{4}{1-\alpha^{2}} (u-u^{(1-\alpha)/2})$ and some $\alpha \neq \pm 1 $, the $f$-divergence $D_f$ is the \textit{$\alpha$-divergence} \cite[Eq.~(100)]{AC210}. The $f$-divergence $D_{f}$ satisfies \eqref{eq:tylpfd} with $\eta=\frac{f''(1)}{2}$ \cite[Th.~4.1]{CS204}.

Another class of invariant divergences, which is in general not an $f$-divergence,  is the \textit{R\'enyi divergence  of order $\alpha>0$}, given by
\begin{eqnarray}
	I_{\alpha}(T \| R) \triangleq  \frac{1}{\alpha-1} \left[  \ln \left(  \sum _{i=1}^{k} T_{i}^{\alpha} R_{i}^{1-\alpha} \right)   \right], \quad \alpha \neq 1. \label{eq:renyidiv}
\end{eqnarray}
By taking the Taylor-series expansion of $	I_{\alpha}(T \| \br) $ around $T=\br$, it can be shown that $I_{\alpha}$ satisfies \eqref{eq:tylpfd} with $\eta=\frac{\alpha}{2}$.

For other invariant divergences, we refer to \cite[Ch.~2]{CZPA209}.

\subsubsection{Non-invariant divergences} Certain Bregman divergences are non-invariant. Indeed, let $\phi$ be a smooth, real-valued, strictly-convex function on $\Xi$.\footnote{We shall say that a function is smooth if it has partial derivatives of all orders.} The Bregman divergence associated with $\phi$ is defined as
\begin{equation}
	D_{\phi}(T\| R)=\phi(\mathbf{T})-\phi(\mathbf{R})- (\mathbf{T}-\mathbf{R})^{\mathsf{T}}\nabla \phi(\mathbf{R}) \label{eq:Bregdef}
\end{equation}
where $\nabla \phi(\mathbf{R})$ denotes the gradient of $\phi$ evaluated at $\mathbf{R}$. When $\phi(\mathbf{x})=\mathbf{x}^{\mathsf{T}} \bm{W}_{\mathbf{\br}} \mathbf{x}$, $\mathbf{x}\in \bd$ for some $(k-1)\times (k-1)$-dimensional positive-definite matrix $\bm{W}_{\mathbf{\br}}$ which possibly depends on $\mathbf{\br}$, the Bregman divergence associated with $\phi$ specializes to the \emph{squared Mahalanobis distance}
\begin{equation}
	D_{\textnormal{SM}}(T\| \br)\triangleq (\mathbf{T}-\mathbf{\br})^{\mathsf{T}} \bm{W}_{\mathbf{\br}} ( \mathbf{T}-\mathbf{\br}). \label{eq:mahadiv}
\end{equation}
Note that the matrix  $\bm{A}_{D_{\textnormal{SM}}, \mathbf{\br}}$ associated with the squared Mahalanobis distance $D_{\textnormal{SM}}$ at $\mathbf{\br}$ is $\bm{W}_{\mathbf{\br}} $. It follows that the squared Mahalanobis distance is non-invariant if $\bm{W}_{\mathbf{\br}} $ is not a constant multiple of $\bm{\Sigma}_{\mathbf{\br}}$.

For a detailed list of divergences and their properties, we again refer to \cite[Ch.~2]{CZPA209}.

\subsection{Divergence Test}
For a divergence $D$ and a threshold $r$, a \emph{divergence test} $\mc_{n}^{D}(r)$ for testing $H_{0}$ against the alternative $H_{1}$ is defined as follows: 
\begin{center}
\begin{tabular}{ll} Observe  $Z^{n}$: &
if $D(\tp_{Z^n} \| P) < r$, then $H_0$ is accepted;\\
 & else $H_1$ is accepted.\end{tabular} \end{center} 
When the divergence $D$ is the Kullback-Leibler divergence $D_{\textnormal{KL}}$, the divergence test becomes the Hoeffding test. 

For a given $r >0$, the acceptance region for $P$ of the divergence test is defined as
	\begin{equation}
		\mathcal{A}^{D}_{n}(r) \triangleq	\left\lbrace  z^{n} \colon D(\tp_{z^{n}} \| P)  <  r \right\rbrace.  
	\end{equation}
The type-I and type-II errors of the divergence test follow then by evaluating \eqref{alpha_mc} and \eqref{eq:beta_mc} for $\mathcal{A}^{D}_{n}(r)$, i.e.,
	\begin{IEEEeqnarray}{lCl}
		\alpha_{n}(\mc_{n}^{D}(r)) & \triangleq & P^{n} \left(   \mathcal{A}^{D}_{n}(r) ^{c} \right) \\
		\beta_{n}(\mc_{n}^{D}(r)) &  \triangleq & Q^{n} ( \mathcal{A}^{D}_{n}(r) ).
	\end{IEEEeqnarray}
In this paper, we derive the first-order term  $\beta_{D}'$ and the second-order term $\beta_{D}''$ of the divergence test $\mc_{n}^{D}(r)$, thereby characterizing its second-order performance.

\section{Second-order Asymptotics of Divergence Tests}\label{sec:results}
This section presents the main results of this paper. In Subsection~\ref{sub:mainresults}, we characterize the first-order term $\beta'_D$ and the second-order term $\beta_D''$ of the divergence test $\mc_n^D(r)$. In Subsection~\ref{sub:compare}, we compare $\beta_D''$ with the second-order term $\beta_{\textnormal{NP}}''$ of the Neyman-Pearson test. \emph{Inter alia}, we observe that $\beta_D''$ scales unfavorably with the cardinality of $P$ and $Q$.

\subsection{The First-Order and Second-Order Terms}
\label{sub:mainresults}
The asymptotic behavior of divergence tests depends critically on the asymptotic behavior of the random variable $nD( \tp_{Z^{n}}\| P) $. 
For certain divergences, the limiting distribution of $nD( \tp_{Z^{n}}\| P)$ (as $n \rightarrow \infty$) has been analyzed in the literature. For example, when $D=D_{\textnormal{KL}}$, a well-known result by Wilks~\cite{W938} yields that   $ 2 n D_{\textnormal{KL}} (\tp_{Z^{n}} \| P)$ converges in distribution to a chi-square random variable with $k-1$ degrees of freedom. This result generalizes to $\alpha$-divergences \cite[Th.~3.1]{TR84}, \cite[Th.~3]{Y72}. 
In Lemma~\ref{convgdiv}, we show that, for a general divergence $D$,  $n D( \tp_{Z^{n}}\| P) $ converges in distribution to a {\em generalized chi-square random variable}, defined as follows:
\begin{definition} \label{genchisrv}  The generalized chi-square distribution is the distribution of the random variable 
	\begin{eqnarray}
		\xi = \sum_{i=1}^{m} w_{i} \Upsilon_{i} \label{eq:genchirv}
	\end{eqnarray}	
	where $w_{i}, i=1,\ldots,m$ are deterministic weight parameters and  $\Upsilon_{i}$, $i=1,\ldots,m$ are independent chi-square random variables with degrees of freedom $k_{i}$, $i=1,\ldots,m$.  We shall denote the generalized chi-square distribution by $ \chi^{2}_{\bm{w},\mathbf{k}}$
	with vector parameters  $\mathbf{w}=(w_{1},\ldots,w_{m})^{\mathsf{T}}$ and $\mathbf{k}=(k_{1},\ldots,k_{m})^{\mathsf{T}}$. As a special case, if $\Upsilon_{i}, i=1,\ldots,m$ are chi-square random variables with degrees of freedom $1$, then we denote the generalized chi-square random variable as $\chi^{2}_{\mathbf{w},m}$ with  weight vector $\mathbf{w}=(w_{1},\ldots,w_{m})^{\mathsf{T}}$ and degrees of freedom $m$. 
\end{definition}

\begin{lemma} \label{convgdiv} Let $Z^{n}$ be a sequence of i.i.d. random variables  distributed according to the null hypothesis $P$, and  let $D$ be a divergence. Further let  $\bm{\lambda}=(\lambda_{1}, \ldots, \lambda_{k-1})^{\mathsf{T}}$ be a vector that contains the eigenvalues of the matrix $ \bm{\Sigma}_{\mathbf{P}}^{-1/2}\bm{A}_{D,\mathbf{P}} \bm{\Sigma}_{\mathbf{P}}^{-1/2} $, where  $\bm{A}_{D,\mathbf{P}}$ is the matrix associated with the divergence $D$ at $\mathbf{P}$ defined in \eqref{eq:matrixa}, and the matrix  $\bm{\Sigma}_{\mathbf{P}}$ is defined in \eqref{eq:covsigmain1}. Then, the tail probability of the random variable $ nD (\tp_{Z^n} \| P)$ can be approximated as 
	\begin{eqnarray}
		P^{n}( 	 n	D (\tp_{Z^{n}} \| P) \geq c)=  \mathsf{Q}_{\chi^{2}_{\bm{\lambda},k-1}}(c) + O(\delta_{n}) \label{eq:betdiv}
	\end{eqnarray}
	for all $c > 0$ and some positive  sequence  $\{\delta_{n}\}$ that is independent of $c$ and  satisfies  $\lim_{n \rightarrow \infty} \delta_{n}=0$. Here, $\mathsf{Q}_{\chi^{2}_{\bm{\lambda},k-1}}(\cdot)$  denotes the tail probability of the generalized chi-square random variable $\chi^{2}_{\bm{\lambda},k-1}$.
\end{lemma}
\begin{proof} See Appendix~\ref{proofconvgdiv}.
\end{proof}

We are now ready to present the main result of this paper, the second-order asymptotic behavior of the type-II error when the type-I error is bounded by $\epsilon$:
\begin{theorem} \label{divergence}
	Let $D$ be a divergence as defined in Definition~\ref{divdef}, and  let $0< \epsilon < 1$.  Consider a divergence test $\mc_{n}^{D}(r)$  for testing $H_{0}\colon Z^{n} \sim P^{n}$ against the alternative $H_{1}\colon Z^{n} \sim Q^{n}$, where $P, Q \in \mathcal{P}(\mathcal{Z})$ and  $P \neq Q$. Recall that the cardinality of $\mathcal{Z}$ is $k\geq 2$. Then, we have the following results:  
	\begin{enumerate}
		\item There exists a threshold value $r_{n}$ (which depends on $n$) satisfying
		\begin{equation}
			\alpha_{n}(\mc_{n}^{D}(r_{n})) \leq \epsilon \label{eq:typeepsilon}
		\end{equation}
		such that, as $n\rightarrow \infty$,
		\begin{equation}
			-\ln \beta_{n}(\mc_{n}^{D}(r_{n}))   \geq   nD_{\textnormal{KL}}(P \| Q) - \sqrt{n} \sqrt{\mathbf{c}^{\mathsf{T}} \bm{A}_{D,\mathbf{P}} ^{-1} \mathbf{c} }\sqrt{\mathsf{Q}^{-1}_{\chi^{2}_{\bm{\lambda},k-1}}(\epsilon) } +O( \max\{ \delta_{n} \sqrt{n}, \ln n \}). \hspace{3mm} \label{eq:achiev}
		\end{equation}
		\item For all $r_{n} >0$ satisfying \eqref{eq:typeepsilon}, we have as $n\rightarrow \infty$
		\begin{equation}
			-\ln \beta_{n}(\mc_{n}^{D}(r_{n}))   \leq  nD_{\textnormal{KL}}(P \| Q) -   \sqrt{n} \sqrt{\mathbf{c}^{\mathsf{T}} \bm{A}_{D,\mathbf{P}} ^{-1} \mathbf{c} }\sqrt{\mathsf{Q}^{-1}_{\chi^{2}_{\bm{\lambda},k-1}}(\epsilon) } +O( \max\{ \delta_{n} \sqrt{n}, \ln n \}). \hspace{3mm} \label{eq:converse}
		\end{equation}
	\end{enumerate}
Here, $\bm{A}_{D,\mathbf{P}} $ is the matrix associated with the divergence $D$ at $\mathbf{P}$ defined in \eqref{eq:matrixa}; the sequence $\{\delta_{n}\}$ was defined in \eqref{eq:betdiv}; $ \mathbf{c}=(c_{1},\ldots,c_{k-1})^{\mathsf{T}}$ is a vector with components
	\begin{equation}
		c_{i} \triangleq\ln    \left( \frac{P_{i}}{ Q_{i}} \right)-\ln    \left( \frac{P_{k}}{ Q_{k}} \right),\quad  i=1,\ldots, k-1; \label{eq:cidef}
	\end{equation}
and  $\mathsf{Q}^{-1}_{\chi^{2}_{\bm{\lambda},k-1}}$ is the inverse of the tail probability of the generalized chi-square distribution $\chi^{2}_{\bm{\lambda},k-1}$ with vector parameter
$\bm{\lambda}=(\lambda_{1}, \ldots, \lambda_{k-1})^{\mathsf{T}}$ containing the eigenvalues of the matrix $ \bm{\Sigma}_{\mathbf{P}}^{-1/2}\bm{A}_{D,\mathbf{P}} \bm{\Sigma}_{\mathbf{P}}^{-1/2}$, and with degrees of freedom $k-1$.
\end{theorem}
\begin{proof}
	See Section~\ref{Sec_Pf_Main}.
\end{proof}
\begin{remark}Since the sequence $\{\delta_n\}$ in 
 \eqref{eq:achiev}--\eqref{eq:converse} tends to zero as $n \to \infty$, we have that $O( \max\{ \delta_{n} \sqrt{n}, \ln n \}) = o(\sqrt{n})$. Consequently, \eqref{eq:achiev} and~\eqref{eq:converse} imply that
	\begin{equation}\label{eq:third_order}
	\sup_{r_n\colon  \alpha_{n}(\mc_{n}^{D}(r_{n}))\leq \epsilon} -\ln \beta_{n}\bigl(\mc_{n}^{D}(r_{n})\bigr) =  nD_{\textnormal{KL}}(P \| Q) - \sqrt{n} \sqrt{\mathbf{c}^{\mathsf{T}} \bm{A}_{D,\mathbf{P}} ^{-1} \mathbf{c} }\sqrt{\mathsf{Q}^{-1}_{\chi^{2}_{\bm{\lambda},k-1}}(\epsilon) } + o(\sqrt{n}). 
	\end{equation}
Thus, Theorem~\ref{divergence} characterizes the first and second-order terms of the divergence test $\mc^D_n(r)$ for every divergence $D$:
\begin{IEEEeqnarray}{lCl}
	\beta'_{D} &=&  D_{\textnormal{KL}}(P \| Q)  \\
	\beta''_{D} &=&-\sqrt{\mathbf{c}^{\mathsf{T}} \bm{A}_{D,\mathbf{P}} ^{-1} \mathbf{c} }\sqrt{\mathsf{Q}^{-1}_{\chi^{2}_{\bm{\lambda},k-1}}(\epsilon) }. \label{eq:soterm}
\end{IEEEeqnarray}
\end{remark}

For the class of invariant divergences, Theorem~\ref{divergence} specializes to the following corollary. 
\begin{corollary} \label{inv_divergence}
Let  $D$ be an invariant divergence, and let $0< \epsilon < 1$. Consider the divergence test $\mc_{n}^{D}(r)$  for testing $H_{0}: Z^{n} \sim P^{n}$ against the alternative $H_{1}: Z^{n} \sim Q^{n}$, where $P, Q \in \mathcal{P}(\mathcal{Z})$ and  $P \neq Q$.  Then we have the following results:
	\begin{enumerate}
\item There exists a threshold value $r_{n}$ (which depends on $n$) satisfying \eqref{eq:typeepsilon} such that, as $n\rightarrow \infty$,
	\begin{equation}\label{eq:achieva}
		-\ln \beta_{n}(\mc_{n}^{D}(r_{n}))  \geq   nD_{\textnormal{KL}}(P \| Q)  -  \sqrt{nV_{\textnormal{KL}}(P \| Q) \mathsf{Q}^{-1}_{\chi^{2}_{k-1}}(\epsilon) } +O( \max\{ \delta_{n} \sqrt{n}, \ln n \}). 
	\end{equation}
\item For all $r_{n} >0$ satisfying \eqref{eq:typeepsilon}, we have as $n\rightarrow \infty$
	\begin{equation}\label{eq:conversea}
		-\ln \beta_{n}(\mc_{n}^{D}(r_{n}))  \leq  nD_{\textnormal{KL}}(P \| Q)  -  \sqrt{nV_{\textnormal{KL}}(P \| Q) \mathsf{Q}^{-1}_{\chi^{2}_{k-1}}(\epsilon) } +O( \max\{ \delta_{n} \sqrt{n}, \ln n \}). 
	\end{equation}
	\end{enumerate}
\end{corollary}

\begin{proof} When the divergence $D$ is invariant, we have that $\bm{A}_{D,\mathbf{P}}= \eta \bm{\Sigma}_{\mathbf{P}}$ for some $\eta>0$. After some algebraic manipulations, it can be shown that
\begin{equation}
    \mathbf{c}^{\mathsf{T}} \bm{A}_{D,\mathbf{P}} ^{-1} \mathbf{c} =\frac{1}{\eta} \mathbf{c}^{\mathsf{T}} \bm{\Sigma}_{\mathbf{P}}^{-1} \mathbf{c}=\frac{1}{\eta} V_{\textnormal{KL}}(P \|Q).
\end{equation}
Furthermore, in this case, $\bm{\lambda}=(\eta,\ldots,\eta)$, and the tail probability of the generalized chi-square random variable  satisfies
\begin{equation}
    \mathsf{Q}_{\chi^{2}_{\bm{\lambda},k-1}} (c)= \mathsf{Q}_{\chi^{2}_{k-1}} (c/\eta), \quad c > 0.
\end{equation}
Consequently,  $\mathsf{Q}^{-1}_{\chi^{2}_{\bm{\lambda},k-1}} (\cdot)=\eta \mathsf{Q}^{-1}_{\chi^{2}_{k-1}}(\cdot)$, from which we obtain that
	\begin{equation}
		\sqrt{\mathbf{c}^{\mathsf{T}} \bm{A}_{D,\mathbf{P}} ^{-1} \mathbf{c} }\sqrt{\mathsf{Q}^{-1}_{\chi^{2}_{\bm{\lambda},k-1}}(\epsilon) } =\sqrt{V_{\textnormal{KL}}(P \| Q) \mathsf{Q}^{-1}_{\chi^{2}_{k-1}}(\epsilon) }. \label{eq:gentoinv}
	\end{equation}
The corollary follows then directly from Theorem~\ref{divergence}.
\end{proof}
Note that Corollary~\ref{inv_divergence} continues to hold if the constant $\eta$ in the definition of an invariant divergence (Definition~\ref{invariancediv}) depends on $P$, i.e., if $\bm{A}_{D,\mathbf{P}}= \eta_P \bm{\Sigma}_{\mathbf{P}}$ for some constant $\eta_P$ which depends on $P$.
\begin{remark} From \eqref{eq:soterm} and \eqref{eq:gentoinv}, we obtain that the second-order term of the divergence test $\mc^D_n(r)$ for an invariant divergence $D$ is given by
\begin{equation}
	\beta''_{D} = -\sqrt{V_{\textnormal{KL}}(P \| Q)}\sqrt{\mathsf{Q}^{-1}_{\chi^{2}_{k-1}}(\epsilon) }. \label{eq:sotermin}
\end{equation}
\end{remark}

For certain divergences,  we can obtain more precise asymptotics by characterizing the sequence $\{\delta_{n}\}$ in \eqref{eq:betdiv} more precisely.  For example, the following result implies that, for both the $\alpha$-divergence and the KL divergence, \eqref{eq:betdiv} holds with  $\delta_{n}=\frac{1}{\sqrt{n}}$.  Indeed, it follows from~\cite[Th.~3.1]{TR84} and \cite[Th.~3]{Y72} that the power divergence statistic 	
\begin{IEEEeqnarray}{lCl}
	T_{\bar{\lambda}} (X)  &\triangleq & \frac{2}{\bar{\lambda} (\bar{\lambda}+1)}\sum_{i=1}^{k} X_{i} \left[  \left(  \frac{X_{i}}{nP_{i}}\right) ^{\bar{\lambda}}-1 \right], \quad  \bar{\lambda} \in \mathbb{R} \setminus \{0,-1\} \\
	T_{0} (X)  &\triangleq &	\lim_{\bar{\lambda} \rightarrow 0}T_{\bar{\lambda}} (X) \\
	T_{-1} (X)  &\triangleq &	\lim_{\bar{\lambda} \rightarrow -1}T_{\bar{\lambda}} (X) 
\end{IEEEeqnarray}
satisfies
\begin{eqnarray}
	P^{n} (T_{\bar{\lambda}}(X) \geq c)=  \mathsf{Q}_{\chi^{2}_{k-1}}(c) + O(n^{-1/2}),\quad c>0, \; \bar{\lambda} \in \mathbb{R}. \label{eq:ratealphadiv1}
\end{eqnarray}
See  \cite{UZ09} for more details. Letting  $\bar{\lambda}=\frac{-(1+\alpha)}{2}, \; \alpha \neq \pm 1$, we can write 	$T_{\bar{\lambda}} (X) $ as
\begin{eqnarray*}
	T_{\bar{\lambda}} (X)  =
	2n D_{\alpha}(\tp_{Z^{n}} \| P), \quad \alpha \neq \pm 1.
\end{eqnarray*}
Similarly, for $\bar{\lambda}=0$, $T_{0} (X) $ can be expressed as $ T_{0} (X)  =
2n D_{\textnormal{KL}}(\tp_{Z^{n}} \| P)$. It then follows from \eqref{eq:ratealphadiv1} that $D_{\alpha}$, $\alpha \neq \pm 1$ and $D_{\textnormal{KL}}$  satisfy \eqref{eq:betdiv} with $\delta_{n}=\frac{1}{\sqrt{n}}$.  Together with \eqref{eq:achieva} and \eqref{eq:conversea}, this implies that the type-II error of the divergence test $\mc_{n}^{D}(r)$ with $D=D_{\alpha}$, $\alpha \neq \pm 1$ or $D=D_{\textnormal{KL}}$ satisfies 
		\begin{equation}\label{eq:achievainb}
			\sup_{r_n\colon  \alpha_{n}(\mc_{n}^{D}(r_{n}))\leq \epsilon} -\ln \beta_{n}\bigl(\mc_{n}^{D}(r_{n})\bigr) =  nD_{\textnormal{KL}}(P \| Q)  -  \sqrt{nV_{\textnormal{KL}}(P \| Q) \mathsf{Q}^{-1}_{\chi^{2}_{k-1}}(\epsilon) } +O(\ln n).
		\end{equation}
Recall that the divergence test $\mc_{n}^{D}(r)$ with $D=D_{\textnormal{KL}}$ is the Hoeffding test. Hence, \eqref{eq:achievainb} characterizes the second-order asymptotic behavior of the Hoeffding test as a special case. 

Next, consider the squared Mahalanobis distance $D_{\textnormal{SM}}$ defined in \eqref{eq:mahadiv}. We have that 
\begin{eqnarray}
	nD_{\textnormal{SM}}(\tp_{Z^{n}}\| P) &=&
n	\left( \mathbf{\tp}_{Z^{n}}-\mathbf{P}\right) ^{\mathsf{T}} \bm{W}_{\mathbf{P}} \left( \mathbf{\tp}_{Z^{n}}-\mathbf{P} \right) \notag \\
	& =& (\mathbf{V}_{n}^{k-1})^{\mathsf{T}}   \bm{W}_{\mathbf{P}} \mathbf{V}_{n}^{k-1} \label{eq:bregstat}
\end{eqnarray}
where   $\mathbf{\tp}_{Z^{n}}=(\tp_{Z^{n}}(a_{1}), \ldots, \tp_{Z^{n}}(a_{k-1}) )^{\mathsf{T}}$ and $	\mathbf{V}^{k-1}_{n} =\sqrt{n}\left(  \mathbf{\tp}_{Z^{n}}-\mathbf{P} \right) $. Since  $\mathbf{x}^{\mathsf{T}} \bm{W}_{\mathbf{P}} \mathbf{x}$ is a strictly convex function of  $\mathbf{x}$, the set $\mathcal{C}= \left\lbrace \mathbf{x}\in\mathbb{R}^{k-1}\colon \mathbf{x}^{\mathsf{T}} \bm{W}_{P} \mathbf{x} \leq c  \right\rbrace $ is convex.  Let  $\mathbf{V}^{k-1} \sim \mathcal{N}(\bm{0}, \bm{\Sigma}^{-1}_{\mathbf{P}})$. Using that \mbox{$\text{Pr}(\mathbf{V}^{k-1} \in \mathcal{C})=1-\mathsf{Q}_{\chi^{2}_{\bm{\lambda},k-1}}(c)$}, we then obtain from \eqref{eq:bregstat} and the Vector Berry-Esseen Theorem \cite[Cor.~1.9]{T214} that
\begin{equation}
	\left| P^{n} \left( n D_{\textnormal{SM}}(\tp_{Z^{n}} \| P) \geq c \right) - \mathsf{Q}_{\chi^{2}_{\bm{\lambda},k-1}}(c) \right|  \leq \frac{ M_{0}}{\sqrt{n}} \label{eq:bebreg}
\end{equation} 
for some constant $M_{0}$ depending on $P$ and the dimension $k$. It follows that, for the squared Mahalanobis distance, \eqref{eq:betdiv} holds again with $\delta_{n}=\frac{1}{\sqrt{n}}$, which together with \eqref{eq:achiev} and \eqref{eq:converse} implies that
	\begin{equation}
			\sup_{r_n\colon  \alpha_{n}(\mc_{n}^{D_{\textnormal{SM}}}(r_{n}))\leq \epsilon} -\ln \beta_{n}\bigl(\mc_{n}^{D_{\textnormal{SM}}}(r_{n})\bigr) = nD_{\textnormal{KL}}(P \| Q) - \sqrt{n} \sqrt{\mathbf{c}^{\mathsf{T}} \bm{A}_{D,\mathbf{P}} ^{-1} \mathbf{c} }\sqrt{\mathsf{Q}^{-1}_{\chi^{2}_{\bm{\lambda},k-1}}(\epsilon) } +O(\ln n). \label{eq:achievsmd}
		\end{equation}

 \subsection{Comparison with the Neyman-Pearson Test}
 \label{sub:compare}
Clearly, the second-order term $\beta_{D}''$  of the divergence test cannot be larger than the second-order term  $ \beta_{\textnormal{NP}}''$ of the Neyman-Pearson test, since the Neyman-Pearson test achieves the minimum type-II error among all hypothesis tests that have access to both $P$ and $Q$. The following proposition shows that $\beta_{D}'' $ is strictly smaller than $ \beta_{\textnormal{NP}}''$. 
\begin{proposition}\label{lem:chi-normal}
	For every $0 <\epsilon<1$ and $\mathbf{c}$ defined in \eqref{eq:cidef}, we have 
	\begin{equation}
		\sqrt{\mathbf{c}^{\mathsf{T}} \bm{A}_{D,\mathbf{P}} ^{-1} \mathbf{c} }\sqrt{\mathsf{Q}^{-1}_{\chi^{2}_{\bm{\lambda},k-1}}(\epsilon) } >\sqrt{V_{\textnormal{KL}}(P \| Q)} \mathsf{Q}_{\mathcal{N}}^{-1}(\epsilon).\label{eq:qsotgnp}
	\end{equation}
\end{proposition}
\begin{IEEEproof}  
	Consider any eigenvalue $\lambda_{j}$  of the positive-definite matrix $\bm{B}= \bm{\Sigma}_{\mathbf{P}}^{-1/2}\bm{A}_{D,\mathbf{P}} \bm{\Sigma}_{\mathbf{P}}^{-1/2}$. Let $Y_{1},\ldots,Y_{k-1}$ be i.i.d. standard Normal random variables. Then, for \mbox{$y \geq 0$},
	\begin{equation}
		\mathsf{Q}_{\mathcal{N}}\left( \frac{y}{\sqrt{\lambda_{j}}}\right) 	<  \text{Pr} \left( Y_{j}^{2} \geq  \frac{y^{2}}{\lambda_{j}}\right)  \leq \text{Pr} \left( \sum_{i=1}^{k-1}  \lambda_{i}Y_{i}^{2} \geq  y^{2}\right) = \mathsf{Q}_{\chi^{2}_{\bm{\lambda},k-1}}(y^{2}) \label{eq:qcomp1}
	\end{equation}
	where the first inequality follows because
 \begin{equation}
    \text{Pr}( Y_{j}^{2} \geq  d^{2})=\text{Pr}(|Y_j| \geq d) = 2 \mathsf{Q}_{\mathcal{N}}(d),\quad d \geq 0
    \end{equation}
    and the second inequality follows because $\lambda_{j}Y_j^2 \leq \sum_{i=1}^{k-1} \lambda_{i} Y_i^2$ with probability one. Let $y$ be such that \mbox{$ \mathsf{Q}_{\chi^{2}_{\bm{\lambda},k-1}}(y^{2})=\epsilon$}. Since $y \mapsto\mathsf{Q}_{\mathcal{N}}(y)$ is strictly decreasing, we obtain from \eqref{eq:qcomp1} that there exists a \mbox{$y' < \frac{y}{\sqrt{\lambda_{j}}}$} such that \mbox{$\mathsf{Q}_{\mathcal{N}}(y')=\epsilon$}. It then follows that $\mathsf{Q}_{\mathcal{N}}^{-1}(\epsilon)  =y'$, which implies that
	\begin{equation}
		\mathsf{Q}_{\mathcal{N}}^{-1}(\epsilon)  < \frac{1}{\sqrt{\lambda_{j}}}\sqrt{\mathsf{Q}^{-1}_{\chi^{2}_{\bm{\lambda},k-1}}(\epsilon)}, \quad j = 1,\ldots,k-1. \label{eq:qcomp2}
	\end{equation}
	Next, we note that
	\begin{equation}
		\frac{\mathbf{c}^{\mathsf{T}} \bm{A}_{D,\mathbf{P}} ^{-1} \mathbf{c}}{V_{\textnormal{KL}}(P \| Q)} =\frac{\mathbf{c}^{\mathsf{T}} \bm{A}_{D,\mathbf{P}} ^{-1} \mathbf{c}}{	\mathbf{c}^{\mathsf{T}} \bm{\Sigma}^{-1}_{\mathbf{P}} \mathbf{c}}.
	\end{equation}
	Using the transformation $\tilde{\mathbf{c}}=\bm{\Sigma}_{\mathbf{P}}^{-1/2} \mathbf{c}$, this can be written as
	\begin{IEEEeqnarray}{lCl}
		\frac{\mathbf{c}^{\mathsf{T}} \bm{A}_{D,\mathbf{P}} ^{-1} \mathbf{c}}{V_{\textnormal{KL}}(P \| Q)} &= &\frac{\tilde{\mathbf{c}}^{\mathsf{T}}   \bm{\Sigma}_{\mathbf{P}}^{1/2}\bm{A}_{D,\mathbf{P}} ^{-1}   \bm{\Sigma}_{\mathbf{P}}^{1/2}\tilde{\mathbf{c}}}{	\tilde{\mathbf{c}}^{\mathsf{T}} \tilde{\mathbf{c}}} \notag\\
		&=& \frac{\tilde{\mathbf{c}}^{\mathsf{T}}  \bm{B}^{-1}\tilde{\mathbf{c}}}{\tilde{\mathbf{c}}^{\mathsf{T}} \tilde{\mathbf{c}}}.\label{eq:quadbound1} 
	\end{IEEEeqnarray}
	It then follows from the Rayleigh-Ritz theorem \cite[Th.~4.2.2]{RJ90} that
	\begin{eqnarray}
		\frac{\tilde{\mathbf{c}}^{\mathsf{T}}  \bm{B}^{-1}\tilde{\mathbf{c}}}{\tilde{\mathbf{c}}^{\mathsf{T}} \tilde{\mathbf{c}}}\geq \frac{1}{\lambda_{\textnormal{max}}}  \label{eq:quadbound3} 
	\end{eqnarray}
	where $\lambda_{\textnormal{max}}>0$ denotes the maximum eigenvalue of the positive-definite matrix $\bm{B}$. Together with \eqref{eq:quadbound1}, this yields
	\begin{eqnarray}
		 \frac{ \sqrt{\mathbf{c}^{\mathsf{T}} \bm{A}_{D,\mathbf{P}}^{-1} \mathbf{c}}}{\sqrt{V_{\textnormal{KL}}(P \| Q)}} \geq \frac{1}{\sqrt{\lambda_{\textnormal{max}}}}. \label{eq:quadbound4} 
	\end{eqnarray}
	Combining  \eqref{eq:qcomp2} and \eqref{eq:quadbound4}, we obtain the desired inequality
	\begin{IEEEeqnarray}{lCl}
	    \sqrt{\mathbf{c}^{\mathsf{T}} \bm{A}_{D,\mathbf{P}} ^{-1} \mathbf{c} }\sqrt{\mathsf{Q}^{-1}_{\chi^{2}_{\bm{\lambda},k-1}}(\epsilon) }&\geq &  \frac{\sqrt{V_{\textnormal{KL}}(P \| Q)} \sqrt{\mathsf{Q}^{-1}_{\chi^{2}_{\bm{\lambda},k-1}}(\epsilon)}}{\sqrt{\lambda_{\textnormal{max}}}} \nonumber \\
		& > & 	\mathsf{Q}_{\mathcal{N}}^{-1}(\epsilon) \sqrt{V_{\textnormal{KL}}(P \| Q)}. \label{eq:quadbound5} 
	\end{IEEEeqnarray} 
\end{IEEEproof}

For an invariant divergence $D$, it follows from Proposition~\ref{lem:chi-normal} and \eqref{eq:gentoinv} that 
\begin{equation}
\sqrt{\mathsf{Q}^{-1}_{\chi^{2}_{k-1}}(\epsilon) } > \mathsf{Q}_{\mathcal{N}}^{-1}(\epsilon), \quad k \geq 2,\, 0<\epsilon<1.\label{eq:q_fun_g}
\end{equation}
Furthermore, we have the following proposition. 
\begin{proposition} \label{chis_coro} For every $0 <\epsilon<1$ and  $k=2, 3, \ldots$, we have $ \mathsf{Q}^{-1}_{\chi^{2}_{k}}(\epsilon) >\mathsf{Q}^{-1}_{\chi^{2}_{k-1}}(\epsilon)$.
\end{proposition}
\begin{IEEEproof} Let $Y_{1},\ldots,Y_{k-1}$ be i.i.d. standard normal random variables. Then, for \mbox{$y >0$}, consider the events
	\begin{equation} 
		A= \left\lbrace  \sum_{i=1}^{k}  Y_{i}^{2} \geq  y\right\rbrace , \quad B= \left\lbrace  \sum_{i=1}^{k-1}  Y_{i}^{2} \geq  y\right\rbrace, \quad C= \left\lbrace  Y_{k}^{2} \geq  y \right\rbrace.
	\end{equation}
	Since $B \cup C \subseteq A$, we have
	\begin{IEEEeqnarray}{lCl}
		\mathrm{Pr}(A) & \geq & \mathrm{Pr}(B \cup C) \nonumber \\
		& =&  \mathrm{Pr}(B) +  \mathrm{Pr}( C \cap B^c)\nonumber\\
		&=&  \mathrm{Pr}(B) +  \mathrm{Pr}( C)  \mathrm{Pr}( B^{c}) \label{eq:ind2}
	\end{IEEEeqnarray}
	because the events $B^c$ and $C$ are independent. 
	Since, for any  $y>0$, we have that $\mathrm{Pr}( C) >0$ and $  \mathrm{Pr}( B^{c})>0$, this implies that
	$\mathrm{Pr}(A)>  \mathrm{Pr}(B)$, i.e., 
	\begin{equation} 
		\mathsf{Q}_{\chi^{2}_{k}}(y) >  \mathsf{Q}_{\chi^{2}_{k-1}}(y), \quad y >0 \label{eq:ind3}.
	\end{equation}
	Now, for $0 < \epsilon <1$, let $y_0$ be such that $\mathsf{Q}_{\chi^{2}_{k-1}}(y_{0})=\epsilon$.  Then, \eqref{eq:ind3} yields that $\mathsf{Q}_{\chi^{2}_{k}}(y_{0}) >  \epsilon$. Since $y \mapsto\mathsf{Q}_{\chi^{2}_{k}}(y)$ is strictly decreasing, there exists a \mbox{$y' >y_{0}$} such that $\mathsf{Q}_{\chi^{2}_{k}}(y')=\epsilon$, so $ \mathsf{Q}^{-1}_{\chi^{2}_{k}}(\epsilon) >\mathsf{Q}^{-1}_{\chi^{2}_{k-1}}(\epsilon)$.
	
\end{IEEEproof}

\begin{figure} 
	\centering
	\includegraphics[width=0.5\textwidth]{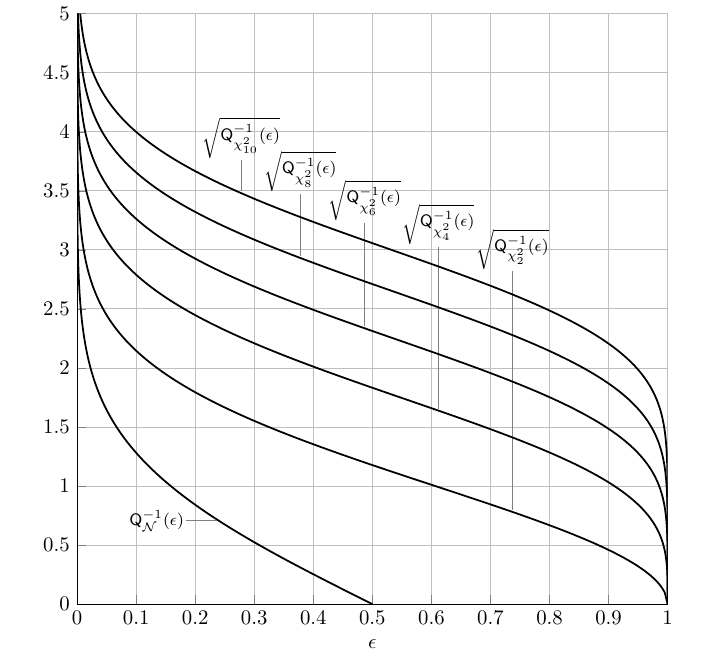}
	\caption{Inverse tail probability functions of the Normal and chi-square distributions as a function of $\epsilon$.}
	\label{quantile}
\end{figure}

Fig.~\ref{quantile} plots the inverse tail probabilities discussed in Propositions~\ref{lem:chi-normal} and \ref{chis_coro}, illustrating the presented inequalities. Theorem~\ref{divergence} together with Proposition~\ref{lem:chi-normal} demonstrate that the divergence test is first-order optimal, but not  second-order optimal. Furthermore, Proposition~\ref{chis_coro} suggests that the divergence test $\mc^D_n(r)$ with an invariant divergence $D$ scales unfavorably with the cardinality $k$ of $P$ and $Q$. The same observation was made by Unnikrishnan \emph{et al.} \cite{Unnikrishnan11} by studying the variance of the normalized Hoeffding test statistic $n D_{\textnormal{KL}}(\tp_{Z^n}\|P)$ and by Boroumand and Guillén i Fàbregas \cite{BF22-2} by studying the subexponential terms of the type-I error; see also Section~\ref{sub:related}.
\comment{ In this context, it is worth to mention the work by , where  they showed that the mismatched divergence has a significant advantage over the Hoeffding test  for the case of composite hypothesis testing problem in a low-dimensional parametric families.}

\section{Numerical Results}
\label{sec:computations}
In order to contrast the performances of different divergence tests, we numerically evaluate the second-order performances of the Neyman-Pearson test $\mc^{\textnormal{NP}}_n$, the Hoeffding test $\mc^{D_{\textnormal{KL}}}_n(r)$, and the divergence test $\mc^{D_{\textnormal{SM}}}_n(r)$ with the squared Mahalanobis distance. Recall that the KL divergence is an invariant divergence. For the squared Mahalanobis distance, we shall consider \eqref {eq:mahadiv}  with  $\bm{W}_{\mathbf{P}}$ having components
\begin{equation}
	\bm{W}_{\mathbf{P}}(i,j)
	=
	\begin{cases}
		\frac{1}{2P_{i}^{2}} +	\frac{1}{2P_{k}^{2}}, \quad & i=j \\
		\frac{1}{2P_{k}^{2}}, \quad & i \neq j. \\
	\end{cases}
\end{equation}
To better visualize the second-order performances, we focus on distributions with dimension $k=3$ and represent them by the two-dimensional vectors $\mathbf{P}=(P_{1},P_{2})^{\mathsf{T}} $ and  $\mathbf{Q}=(Q_{1}, Q_{2})^{\mathsf{T}} $ in the coordinate space $\Xi$ (defined in \eqref{eq:domaindef}).

We shall approximate the second-order performances of the aforementioned hypothesis tests by
\begin{IEEEeqnarray}{lCl}
	-\frac{1}{n}\ln \beta_{n}(\mc_{n}^{\textnormal{NP}})   &\approx &  D_{\textnormal{KL}}(P \| Q) - \sqrt{\frac{V_{\textnormal{KL}}(P \| Q)}{n}}\mathsf{Q}_{\mathcal{N}}^{-1}(\epsilon)   \label{eq:so1}\\
	-\frac{1}{n}\ln \beta_{n}(\mc_{n}^{D_{\textnormal{KL}}})   &\approx &  D_{\textnormal{KL}}(P \| Q) - \sqrt{\frac{V_{\textnormal{KL}}(P \| Q)}{n}}\sqrt{\mathsf{Q}^{-1}_{\chi^{2}_{2}}(\epsilon) }   \label{eq:so2}\\
	-\frac{1}{n}\ln \beta_{n}(\mc_{n}^{D_{\textnormal{SM}}})  & \approx &  D_{\textnormal{KL}}(P \| Q) -  \sqrt{\frac{\mathbf{c}^{\mathsf{T}} \bm{W}_{\mathbf{P}}^{-1} \mathbf{c} }{n}}\sqrt{\mathsf{Q}^{-1}_{\chi^{2}_{\bm{\lambda},2}(\epsilon)}}.   \label{eq:so3}
\end{IEEEeqnarray}
Since the first-order term $\beta'$ is not affected by the choice of the hypothesis test, we shall further compare them by means of their second-order terms $\beta''$. In particular, we shall compare the Hoeffding test $\mc^{D_{\textnormal{KL}}}_n(r)$ and the divergence test $\mc^{D_{\textnormal{SM}}}_n(r)$ by considering the ratio of their second-order terms $\beta''$ as a function of $P$, $Q$, and $\epsilon$:
\begin{equation}
\rho(P,Q,\epsilon)= \frac{  \sqrt{\mathbf{c}^{\mathsf{T}} (\bm{W}_{\mathbf{P}})^{-1} \mathbf{c} } \sqrt{\mathsf{Q}^{-1}_{\chi^{2}_{\bm{\lambda},2}}(\epsilon)}}{\sqrt{V_{\textnormal{KL}}(P \| Q)}\sqrt{ \mathsf{Q}^{-1}_{\chi^{2}_{2}}(\epsilon)}}. \label{eq:ratiodh}
\end{equation}
If $\rho(P,Q,\epsilon)>1$, then the second-order term of the divergence test $\mc^{D_{\textnormal{SM}}}_n(r)$ is strictly smaller than the second-order term of the Hoeffding test $\mc^{D_{\textnormal{KL}}}_n(r)$, hence the Hoeffding test has a better second-order performance. In contrast, if $\rho(P,Q,\epsilon)<1$, then the divergence test has a better second-order performance.

\subsection{Comparison Based on Ratio of Second-Order Terms}

\begin{figure}[t]
	\centering
	\begin{subfigure}[b]{0.49\textwidth}
		\centering
		\includegraphics[width=\textwidth]{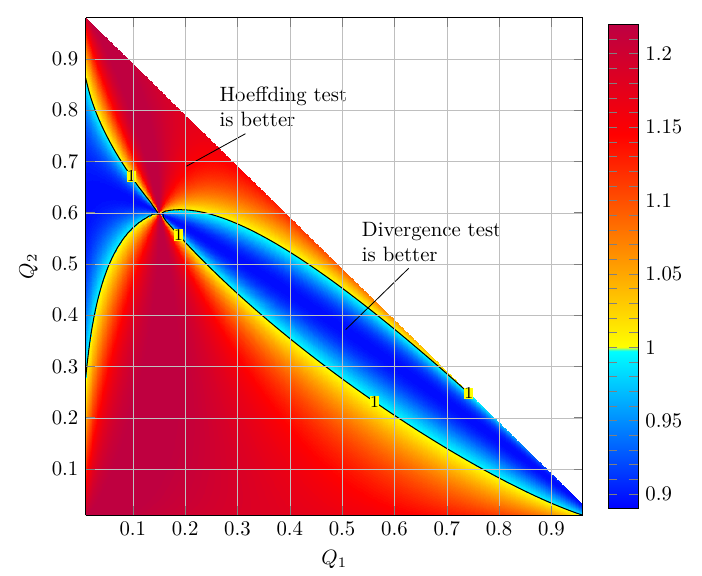}
		\caption{$P=(0.15, 0.6, 0.25)$}
		\label{fig:is_p1}
	\end{subfigure}
	\hfill
	\begin{subfigure}[b]{0.49\textwidth}
		\centering
		\includegraphics[width=\textwidth]{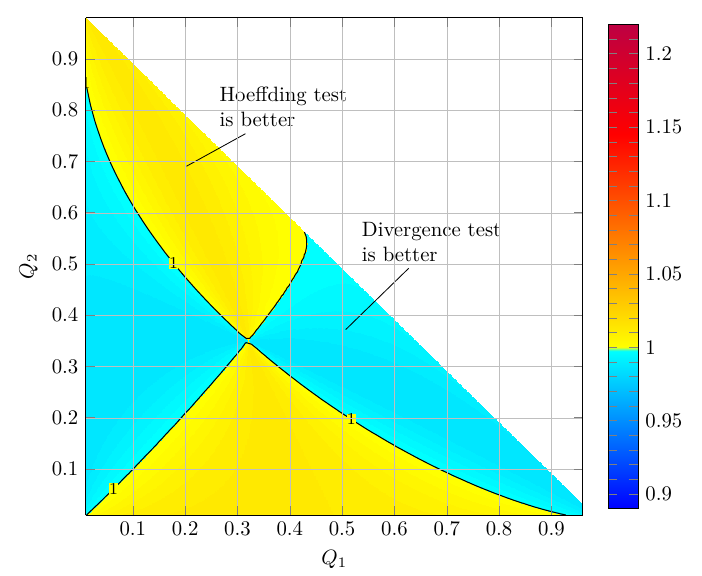}
		\caption{$P=(0.32,  0.35, 0.33)$}
		\label{fig:is_p2}
	\end{subfigure}
	\caption{Ratio $\rho(P,Q,\epsilon)$ of the second-order terms of the Hoeffding test $\mc^{D_{\textnormal{KL}}}_n(r)$ and the divergence test $\mc^{D_{\textnormal{SM}}}_n(r)$ as a function of $\mathbf{Q}$ for $\epsilon=0.02$.}
	\label{fig:is_dif_p}
\end{figure}
\begin{figure}[t]
	\centering
	\begin{subfigure}{0.49\textwidth}
		\centering
		\includegraphics[width=\textwidth]{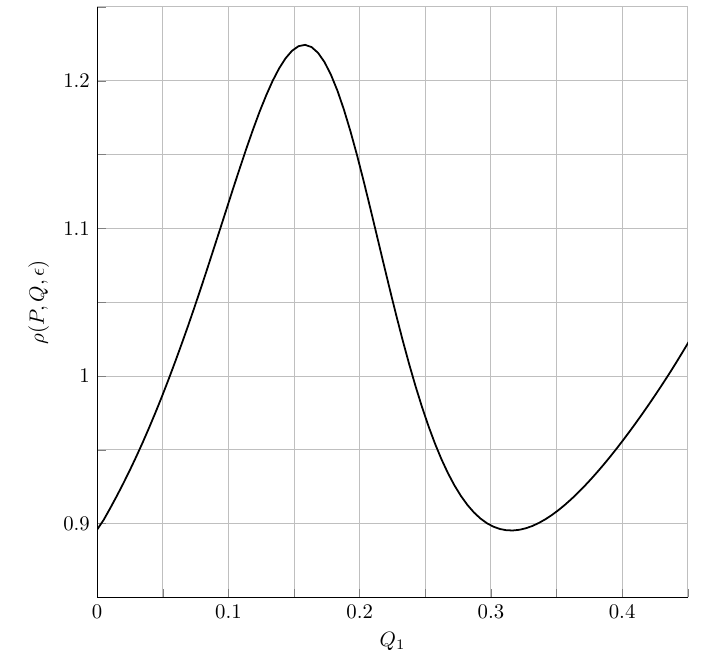}
		\caption{$P=(0.15, 0.6, 0.25)$}
	\end{subfigure}
	\hfill
	\begin{subfigure}{0.49\textwidth}
		\centering
		\includegraphics[width=\textwidth]{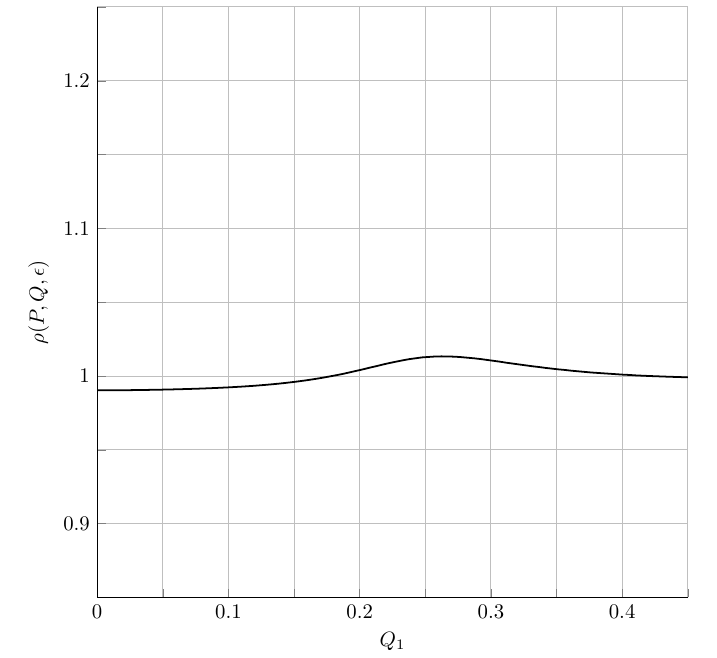}
		\caption{$P=(0.32, 0.35, 0.33)$}
	\end{subfigure}
	\caption{Ratio $\rho(P,Q,\epsilon)$ of the second-order terms of the Hoeffding test $\mc^{D_{\textnormal{KL}}}_n(r)$ and the divergence test $\mc^{D_{\textnormal{SM}}}_n(r)$ as a function of $Q_1$ for $Q_2=0.5$ and $\epsilon=0.02$.}
	\label{fig:ratio_dif_p}
\end{figure}

In Fig.~\ref{fig:is_dif_p}, we plot the contour lines of $\rho(P,Q,\epsilon)$ as a function of $\mathbf{Q}\in\Xi$ for $\epsilon=0.02$ and two different null hypotheses $P=(0.15,0.6,0.25)$ and $P=(0.32,0.35,0.33)$. In the figure, the coordinate space $\Xi$ is divided into two regions: one region is labeled as ``Hoeffding test is better" and includes the points $\mathbf{Q}  \in \Xi$ for which $\rho(P,Q,\epsilon)>1$; the other region is labeled as ``Divergence test is better" and includes the points $\mathbf{Q}  \in \Xi$ for which $\rho(P,Q,\epsilon)<1$. The solid contour lines drawn in the subfigures show all the points $\mathbf{Q}  \in \Xi$ for which the Hoeffding test and the divergence test have the same second-order performance. For each subfigure, the color-bar on the right indicates the values of the ratio $\rho(P,Q,\epsilon)$. Note that both subfigures use the same color-coding, i.e., they use the same colors for the same values of $\rho(P,Q,\epsilon)$. In Fig.~\ref{fig:ratio_dif_p}, we further plot $\rho(P,Q,\epsilon)$ as a function of $Q_1$, the first component of $Q$, when the second component of $Q$ is set to $Q_2 = 0.5$, and $\epsilon$ and $P$ are as in Fig.~\ref{fig:is_dif_p}.

Observe that there are alternative distributions $Q$ for which the Hoeffding test $\mc^{D_{\textnormal{KL}}}_n(r)$ has a better second-order performance than the divergence test $\mc^{D_{\textnormal{SM}}}_n(r)$, and there are distributions $Q$ for which the opposite is true. The set of distributions $Q$ for which one test outperforms the other, as well as the fluctuations of the ratio $\rho(P,Q,\epsilon)$ with $Q$, depend on the distribution $P$.

\begin{figure}[t]
	\centering
	\begin{subfigure}[b]{0.49\textwidth}
		\centering
		\includegraphics[width=\textwidth]{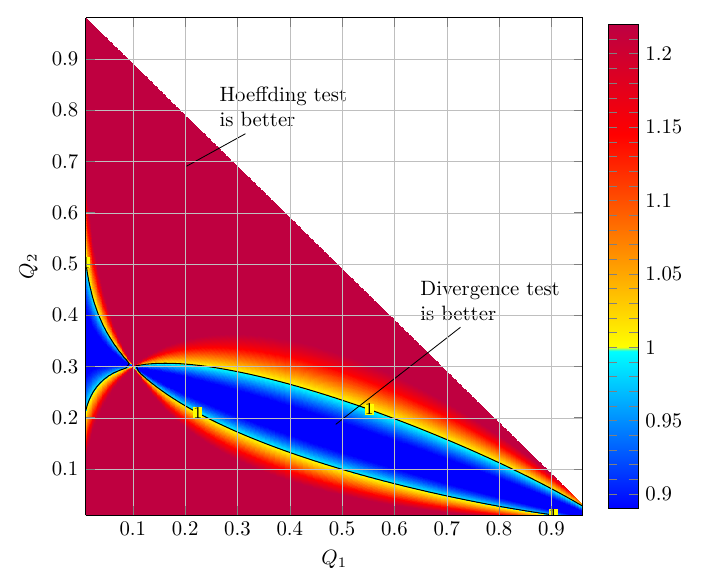}
		\caption{$\epsilon=0.02$}
	\end{subfigure}
	\hfill
	\begin{subfigure}[b]{0.49\textwidth}
		\centering
		\includegraphics[width=\textwidth]{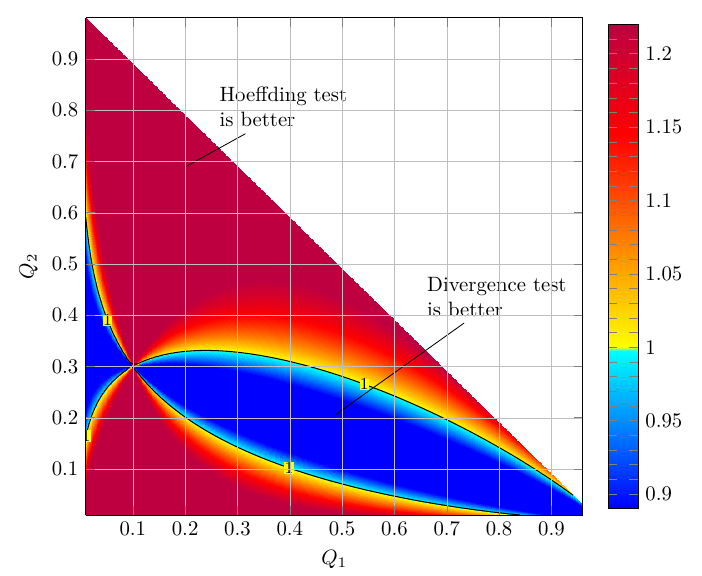}
		\caption{$\epsilon=0.5$}
	\end{subfigure}
	\caption{Ratio $\rho(P,Q,\epsilon)$ of the second-order terms of the Hoeffding test $\mc^{D_{\textnormal{KL}}}_n(r)$ and the divergence test $\mc^{D_{\textnormal{SM}}}_n(r)$ as a function of $\mathbf{Q}$ for $P=(0.1,0.3,0.6)$.}
	\label{fig:is_dif_eps}
\end{figure}

Fig.~\ref{fig:is_dif_eps} depicts the contour lines of $\rho(P,Q,\epsilon)$ as a function of $\mathbf{Q}\in\Xi$ for $P=(0.1, 0.3, 0.6)$ and two different values of $\epsilon$, namely, $\epsilon=0.02$ and $\epsilon=0.5$. Again, the coordinate space $\Xi$ is divided into the two regions ``Hoeffding test is better" (where $\rho(P,Q,\epsilon)>1$) and ``Divergence test is better" (where $\rho(P,Q,\epsilon)<1$), and the solid contour lines in the subfigures show all the points of $\mathbf{Q}\in\Xi$ for which the Hoeffding test and the divergence test have the same second-order performance. Observe again that there are alternative distributions $Q$ for which the Hoeffding test $\mc^{D_{\textnormal{KL}}}_n(r)$ has a better second-order performance than the divergence test $\mc^{D_{\textnormal{SM}}}_n(r)$, and there are distributions $Q$ for which the opposite is true. The set of distributions $Q$ for which one test outperforms the other, as well as the fluctuations of the ratio $\rho(P,Q,\epsilon)$ with $Q$, depend on $\epsilon$.

 \subsection{Comparison Depending on Sample Size}
 \begin{figure}[t]
	\centering
	\begin{subfigure}[b]{0.49\textwidth}
		\centering
		\includegraphics[width=\textwidth]{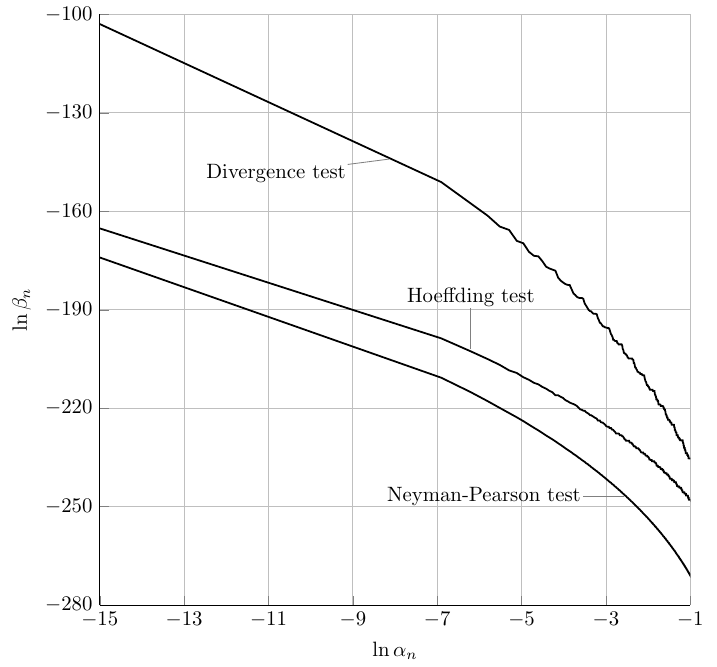}
		\caption{$Q=(0.45, 0.15, 0.4)$}
	\end{subfigure}
	\hfill
	\begin{subfigure}[b]{0.49\textwidth}
		\centering
		\includegraphics[width=\textwidth]{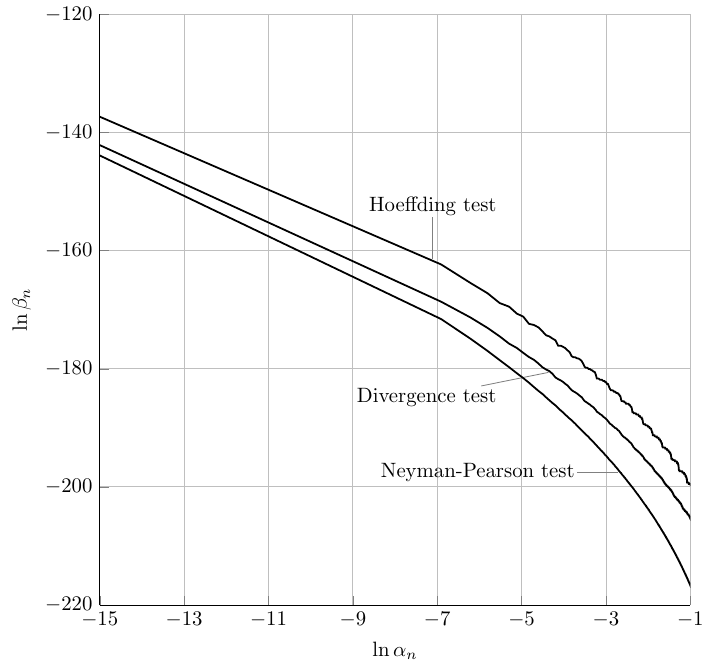}
		\caption{$Q=(0.6, 0.3,0.1)$}
	\end{subfigure}
	\caption{Type-I error vs.\ type-II error for $P=(0.15, 0.6, 0.25)$ and $n=500$.}
	\label{fig:actual_type12}
\end{figure}

In Fig.~\ref{fig:actual_type12}, we plot the logarithm of the type-I error $\ln \alpha_n$ versus the logarithm of the type-II error $\ln\beta_n$ for the Neyman-Pearson test $\mc^{\textnormal{NP}}_n$, the Hoeffding test $\mc^{D_{\textnormal{KL}}}_n(r)$, and the divergence test $\mc^{D_{\textnormal{SM}}}_n(r)$ for the null hypothesis $P=(0.15, 0.6,0.25)$, sample size $n=500$, and two different alternative hypotheses $Q=(0.45, 0.15, 0.4)$ and $Q=(0.6, 0.3,0.1)$. Note that $\ln\alpha_n$ and $\ln\beta_n$ are the actual (nonasymptotic) values computed numerically for the considered hypothesis tests, they are not obtained from the second-order approximations \eqref{eq:so1}--\eqref{eq:so3}. Observe that, for the alternative distribution $Q=(0.45, 0.15, 0.4)$, the Hoeffding test achieves a smaller type-II error than the divergence test, whereas for the alternative distribution $Q=(0.6, 0.3,0.1)$, the divergence test achieves a smaller type-II error. As expected, in both cases the smallest type-II error is achieved by the Neyman-Pearson test.

\begin{figure}[t]
	\centering
	\begin{subfigure}[b]{0.49\textwidth}
		\centering
		\includegraphics[width=\textwidth]{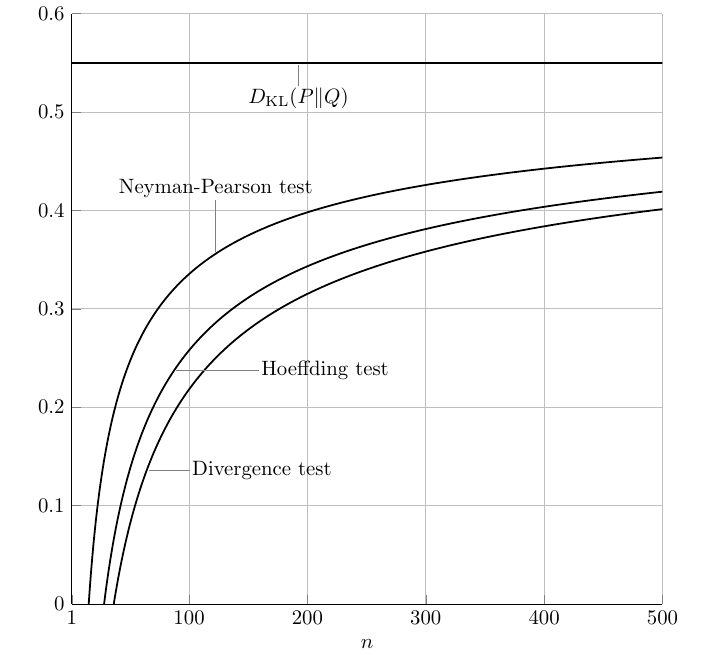}
		\caption{$Q=(0.45, 0.15, 0.4)$}
	\end{subfigure}
	\hfill
	\begin{subfigure}[b]{0.49\textwidth}
		\centering
		\includegraphics[width=\textwidth]{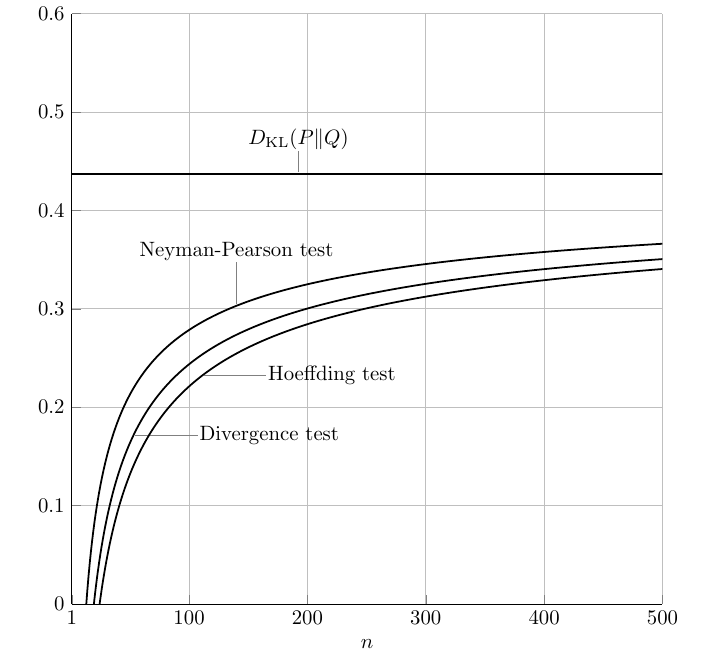}
		\caption{$Q=(0.6, 0.3,0.1)$}
	\end{subfigure}
 
	\caption{Second-order approximations \eqref{eq:so1}--\eqref{eq:so3} as a function of the sample size $n$ for $P=(0.15, 0.6, 0.25)$ and $\epsilon=0.02$.}
	\label{fig:so_n}
\end{figure}

To better illustrate the behavior of the type-II error $\beta_n$ as a function of the sample size $n$, we plot in Fig.~\ref{fig:so_n} the second-order approximations \eqref{eq:so1}--\eqref{eq:so3} versus $n$ for $P=(0.15, 0.6, 0.25)$, $\epsilon=0.02$, and two different alternative hypotheses $Q=(0.45, 0.15, 0.4)$ and $Q=(0.6, 0.3,0.1)$. We further indicate the value of the first-order term $\beta'=D_{\textnormal{KL}}(P\|Q)$. As expected, $-\frac{1}{n}\ln \beta_{n}(\mc_{n}^{\textnormal{NP}})$ exhibits the fastest convergence to $D_{\textnormal{KL}}(P \| Q)$. However, whether the Hoeffding test $\mc_n^{D_{\textnormal{KL}}}(r)$ or the divergence test $\mc_n^{D_{\textnormal{SM}}}(r)$ converges faster to $D_{\textnormal{KL}}(P \| Q)$ depends on $Q$. In this example, the Hoeffding test $\mc_n^{D_{\textnormal{KL}}}(r)$ exhibits a faster convergence for $Q=(0.45, 0.15, 0.4)$, whereas the divergence test $\mc_n^{D_{\textnormal{SM}}}(r)$ exhibits a faster convergence for $Q=(0.6,0.3,0.1)$.

\subsection{Comparison Based on Second-Order Term}
\begin{figure}
	\centering
	\begin{subfigure}[b]{0.49\textwidth}
		\centering
		\includegraphics[width=\textwidth]{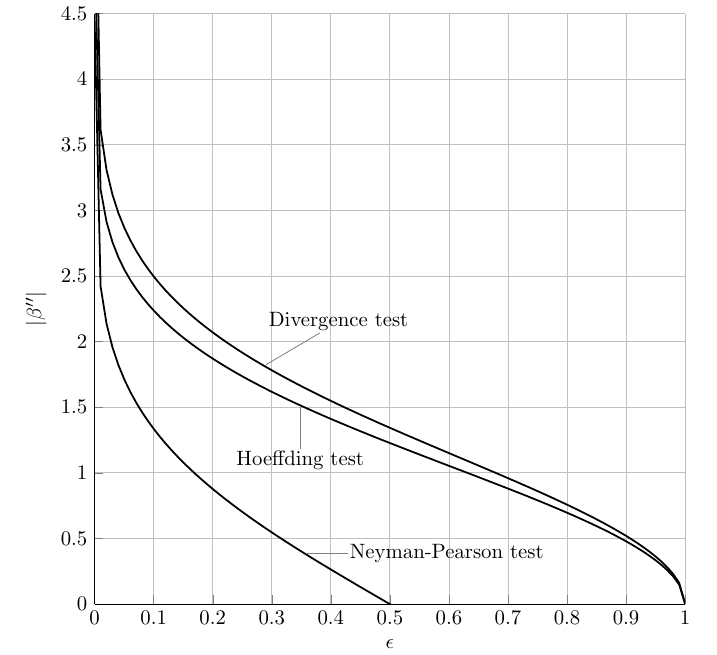}
		\caption{$Q=(0.45, 0.15, 0.4)$}
	\end{subfigure}
	\hfill
	\begin{subfigure}[b]{0.49\textwidth}
		\centering
		\includegraphics[width=\textwidth]{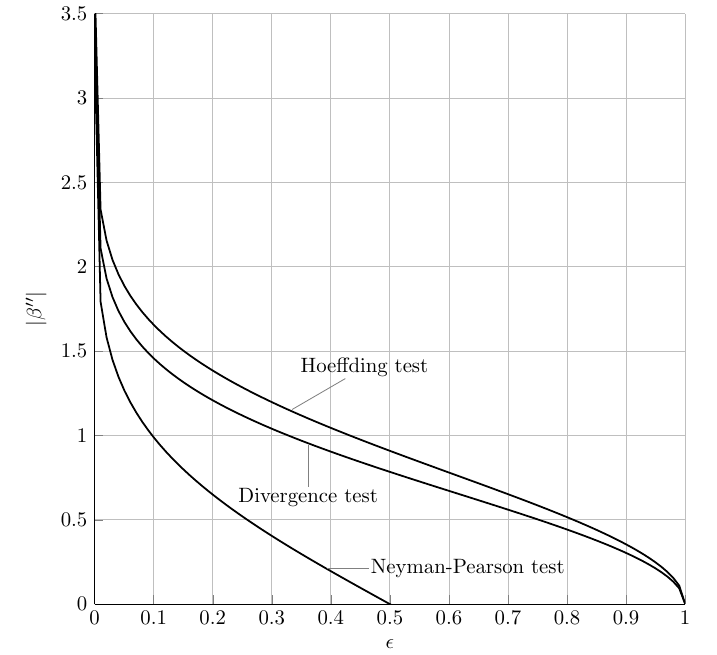}
		\caption{$Q=(0.6, 0.3,0.1)$}
	\end{subfigure}
	\caption{Absolute values of the second-order terms $\beta''_{\textnormal{NP}}$, $\beta''_{D_{\textnormal{KL}}}$, and $\beta''_{D_{\textnormal{SM}}}$ as a function of $\epsilon$ for $P=(0.15, 0.6, 0.25)$ and $k=3$.}
	\label{fig:eps_so_a3}
\end{figure}

We next compare the second-order terms $\beta''_{\textnormal{NP}}$, $\beta''_{D_{\textnormal{KL}}}$, and $\beta''_{D_{\textnormal{SM}}}$ as a function of $\epsilon$. Specifically, Fig.~\ref{fig:eps_so_a3} shows the absolute values of these second-order terms versus $\epsilon$ for $P=(0.15, 0.6, 0.25)$ and two different alternative hypotheses $Q=(0.45, 0.15, 0.4)$ and $Q=(0.6, 0.3,0.1)$. Note that, since the second-order term $\beta''$ is negative, a small absolute value implies a large second-order term and, hence, a better second-order performance. Again, we observe that the Neyman-Pearson test has the largest second-order term. However, whether the second-order term of the Hoeffding test $\mc_n^{D_{\textnormal{KL}}}(r)$ is larger or smaller than the second-order term of the divergence test $\mc_n^{D_{\textnormal{SM}}}(r)$ depends on the alternative hypothesis $Q$. 

\begin{figure} 
	\centering
	\includegraphics[width=0.5\textwidth]{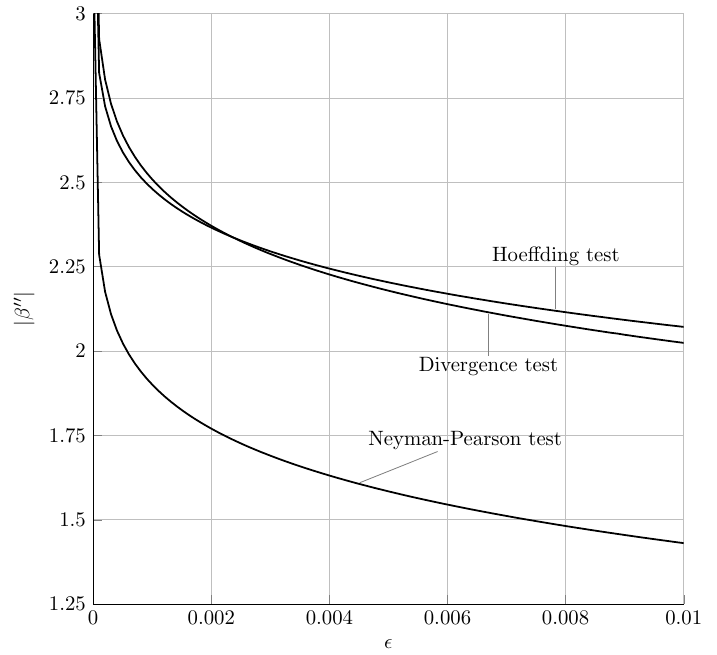}
	\caption{Absolute values of the second-order terms $\beta''_{\textnormal{NP}}$, $\beta''_{D_{\textnormal{KL}}}$, and $\beta''_{D_{\textnormal{SM}}}$ as a function of $\epsilon$ for $P=(0.1,0.3,0.4,0.2)$ and $Q=(0.36, 0.16, 0.22, 0.26)$.}
	\label{fig:tik_fig1}
\end{figure}

Finally, Fig.~\ref{fig:tik_fig1} shows the absolute values of the second-order terms $\beta''_{\textnormal{NP}}$, $\beta''_{D_{\textnormal{KL}}}$, and $\beta''_{D_{\textnormal{SM}}}$ as a function of $\epsilon$ for the $4$-dimensional distributions $P=(0.1,0.3,0.4,0.2)$ and $Q=(0.36, 0.16, 0.22, 0.26)$. Interestingly, in this example the second-order terms of the Hoeffding test $\mc_n^{D_{\textnormal{KL}}}(r)$ and the divergence test $\mc_n^{D_{\textnormal{KL}}}(r)$ cross: when $\epsilon$ is below a threshold, we have $|\beta''_{D_{\textnormal{KL}}}| < |\beta''_{D_{\textnormal{SM}}}|$; when $\epsilon$ is above the threshold, we have $|\beta''_{D_{\textnormal{KL}}}| > |\beta''_{D_{\textnormal{SM}}}|$.

\section{Proof of Theorem~\ref{divergence}}\label{Sec_Pf_Main}
\subsection{Proof of  Part~1 of  Theorem~\ref{divergence}}
From  Lemma~\ref{convgdiv}, it follows that  there exist $M_{0}>0$ and $N_{0} \in \mathbb{N}$ such that 
\begin{equation}
	\left| P^{n} \left( n D(\tp_{Z^{n}} \| P) \geq c \right) - \mathsf{Q}_{\chi^{2}_{\bm{\lambda},k-1}}(c) \right|  \leq M_{0}\delta_{n}, \quad n \geq N_{0} \label{eq:ratekl1}
\end{equation} 
for some positive $\delta_n$ that vanishes as $n \to \infty$.
For $0<\epsilon<1$, let 
\begin{equation}
	r_{n}= \frac{1}{n} \mathsf{Q}^{-1}_{\chi^{2}_{\bm{\lambda},k-1}} ( \epsilon- M_{0} \delta_{n}).  \label{eq:rvalue}
\end{equation}
It follows then from \eqref{eq:ratekl1} that the type-I error is bounded by $\epsilon$ since
\begin{IEEEeqnarray}{lCl}
	\alpha_{n}(\mc_{n}^{D}(r_{n})) & =&  P^{n}( D(\tp_{Z^{n}} \| P)  \geq  r_{n}) \nonumber\\
	& =&  P^{n} \left(  n D(\tp_{Z^{n}} \| P) \geq  \mathsf{Q}^{-1}_{\chi^{2}_{\bm{\lambda},k-1}}( \epsilon- M_{0} \delta_{n} ) \right) \nonumber\\
	& \leq &  \mathsf{Q}_{\chi^{2}_{\bm{\lambda},k-1}} \left( \mathsf{Q}^{-1}_{\chi^{2}_{\bm{\lambda},k-1}} \left( \epsilon- M_{0} \delta_{n}  \right) \right) +  M_{0} \delta_{n}  \nonumber\\
	& =& \epsilon, \quad n \geq N_0.
\end{IEEEeqnarray}
Thus, the threshold $r_n$ satisfies \eqref{eq:typeepsilon}.

We next consider the type-II error for this threshold. To this end, we define the  divergence ball of radius $r'>0$ around the distribution $P$ as
\begin{equation}
	\mathcal{B}_{D,P}(r') \triangleq \left\lbrace T \in \mathcal{P}(\mathcal{Z}) \colon D (T \|  P) <r' \right\rbrace, \quad r'>0. \label{eq:klball}  
\end{equation} 
We  denote  the set of types with denominator $n$ by $\mathcal{P}_{n}$, and $\tc(\tilde{P})$ is the type class  of $\tilde{P}$,  defined as
\begin{equation}
	\tc(\tilde{P}) \triangleq \left\lbrace z^{n} \in \mathcal{Z}^{n} \colon P_{z^{n}} = \tilde{P} \right\rbrace. \label{eq:typeclass}
\end{equation}
It follows that
\begin{IEEEeqnarray}{lCl}
	\beta_{n}(\mc_{n}^{D}(r_{n}))  &=&  Q^{n}(D(\tp_{Z^{n}} \| P)   < r_{n}) \nonumber  \\
	&= &  \sum_{z^{n}: D(\tp_{z^{n}} \| P)   < r_{n} } Q^{n}(z^{n}) \nonumber  \\
	&=&  \sum_{\tilde{P} \in \mathcal{P}_{n} \cap \mathcal{B}_{D,P}(r_{n}) }  Q^{n}(\tc(\tilde{P})) \nonumber  \\
	& \leq&   \sum_{\tilde{P} \in \mathcal{P}_{n} \cap \mathcal{B}_{D,P}(r_{n})}  \exp \{-n D_{\text{KL}}(\tilde{P} \| Q) \},\quad n\geq N_{0}  \label{eq:tyerr}
\end{IEEEeqnarray}
where the last step follows from  \cite[Th.~11.1.4]{ATCB}. 

We next derive a lower bound on $D_{\text{KL}}(\tilde{P} \| Q)$ for $\tilde{P} \in \mathcal{P}_{n} \cap \mathcal{B}_{D,P}(r_{n})$. To this end,  we shall use  the following auxiliary results.
\begin{lemma} \label{pinskercon} For any divergence $D$ and distribution  $\br \in \mathcal{P}(\mathcal{Z})$,  consider  the function 	$f_{\br}(\mathbf{T})\triangleq D(T \| \br)$. Suppose  that, for any sequence $\{\mathbf{T}_{n}\}$ of coordinates, we  have $f_{\br}(\mathbf{T}_{n})=O(b_{n})$ for a strictly positive sequence $\left\lbrace b_{n} \right\rbrace $ satisfying $b_{n} \rightarrow 0$  as $n \rightarrow \infty$.  Then, we have 
	$\| \mathbf{T}_{n}-\mathbf{\br}\|_{2} =O(\sqrt{b_{n}})$ as $n \rightarrow \infty$.
\end{lemma}
\begin{IEEEproof} See Appendix~\ref{proofpinsker}.
\end{IEEEproof}
\begin{remark}
Lemma~\ref{pinskercon} and \eqref{eq:tayldivA} imply that, if the sequence $\{\mathbf{T}_{n}\}$ of coordinates satisfies \mbox{$D_1(\mathbf{T}_n\| \mathbf{R}) = O(b_n)$} for some divergence $D_1$ and a strictly positive sequence $\{b_n\}$ that vanishes as $n\to\infty$, then the same sequence of coordinates also satisfies $D_2(\mathbf{T}_n\| \mathbf{R}) = O(b_n)$ for any other divergence $D_2$.
\end{remark}

For our choice of $r_n$ in \eqref{eq:rvalue}, we have that $r_{n}=O\left( \frac{1}{n}\right)$ as $\delta_{n} \rightarrow 0$.  By taking $\br=P$, Lemma~\ref{pinskercon} thus implies that, for every $T \in 	\mathcal{B}_{D,P}(r_{n})$,  
\begin{eqnarray}
	\|\mathbf{T}-\mathbf{P}\|_{2}  = O\left( \frac{1}{\sqrt{n}}\right). \label{eq:remaindern}
\end{eqnarray}
 Then, we obtain from \eqref{eq:tayldivA} that, for every $T \in \mathcal{B}_{D,P}(r_{n})$,
\begin{eqnarray}
	D(T \| P) 
	&= &(\mathbf{T}-\mathbf{P})^{\mathsf{T}} \bm{A}_{D,\mathbf{P}} ( \mathbf{T}-\mathbf{P})+O\left( \frac{1}{n^{3/2}}\right). \label{eq:tayldivAn}
\end{eqnarray}
This implies that  there exist $M>0$ and $N_{1} \geq  N_{0}$ such that, for all $ n \geq  N_{1} $ and $T \in \mathcal{B}_{D,P}(r_{n}) $,
\begin{eqnarray}
	|D(T \| P) - (\mathbf{T}-\mathbf{P})^{\mathsf{T}} \bm{A}_{D,\mathbf{P}} ( \mathbf{T}-\mathbf{P}) | \leq  \frac{ M }{n^{3/2}}. \label{eq:tylp2}
\end{eqnarray}
Next, we define for $ r'>0$,
\begin{IEEEeqnarray}{lCl}
	\mathcal{B}_{\bm{A}_{D,\mathbf{P}}}(r') &\triangleq & \left\lbrace T  \in \mathcal{P}(\mathcal{Z}) \colon (\mathbf{T}-\mathbf{P})^{\mathsf{T}} \bm{A}_{D,\mathbf{P}} ( \mathbf{T}-\mathbf{P}) <r' \right\rbrace\\
	\bar{\mathcal{B}}_{\bm{A}_{D,\mathbf{P}}}(r') &\triangleq & \left\lbrace T  \in \mathcal{P}(\mathcal{Z}) \colon (\mathbf{T}-\mathbf{P})^{\mathsf{T}} \bm{A}_{D,\mathbf{P}} ( \mathbf{T}-\mathbf{P}) \leq r' \right\rbrace.
\end{IEEEeqnarray}
We then have the following lemma.
\begin{lemma} \label{lemmaball}  Let $ r_{n}$ be as in \eqref{eq:rvalue} and $\mathcal{B}_{D,P}(r_{n})$ as in \eqref{eq:klball}. Then,
	\begin{equation}
	\mathcal{B}_{D,P}(r_{n}) \subseteq
		\bar{\mathcal{B}}_{\bm{A}_{D,\mathbf{P}}} \left( r_{n}+ \frac{M}{  n^{3/2}}\right), \quad   n \geq N_{1}. \label{eq:klchi}
	\end{equation}
\end{lemma}

\begin{IEEEproof} 
	It follows from \eqref{eq:tylp2} that, for  $T \in 	\mathcal{B}_{D,P}(r_{n})$  and $ n \geq N_{1}$, 
	\begin{equation}
		(\mathbf{T}-\mathbf{P})^{T} \bm{A}_{D,\mathbf{P}} ( \mathbf{T}-\mathbf{P})- \frac{ M }{n^{3/2}} 
		\leq D(T \| P)<r_{n}.
	\end{equation}	
	This implies that
	\begin{eqnarray*}
		(\mathbf{T}-\mathbf{P})^{\mathsf{T}} \bm{A}_{D,\mathbf{P}} ( \mathbf{T}-\mathbf{P}) < r_{n}+ \frac{M}{ n^{3/2}}.
	\end{eqnarray*}
	Thus $T \in \mathcal{B}_{\bm{A}_{D,\mathbf{P}}}\left( r_{n}+ \frac{M}{  n^{3/2}}\right) \subseteq \bar{\mathcal{B}}_{\bm{A}_{D,\mathbf{P}}}  \left( r_{n}+ \frac{M}{  n^{3/2}}\right)$.
\end{IEEEproof}
\begin{lemma} \label{ctaylor} For any two probability distributions  $Q=(Q_{1},\ldots,Q_{k})^{\mathsf{T}}$ and  $T=(T_{1},\ldots, T_{k})^{\mathsf{T}}$, the second-order Taylor approximation  of the function $T \mapsto D_{\textnormal{KL}}(T \|Q)$ around  the null hypothesis $T=P$ is given by
	\begin{equation}
		D_{\textnormal{KL}}(T \|Q)=   D_{\textnormal{KL}}(P \| Q) + \sum_{i=1}^{k} (T_{i} -P_{i} ) \ln    \left( \frac{P_{i}}{ Q_{i}} \right)+  \frac{1}{2} d_{\chi^{2}} (T, P) + O(\|\mathbf{T}-\mathbf{P}\|_{2}^{3})  \label{eq:tylq}
	\end{equation}
	as $\|\mathbf{T}-\mathbf{P}\|_{2} \rightarrow 0$.
\end{lemma}
\begin{IEEEproof} See Appendix~\ref{claimtaylor}.
\end{IEEEproof}

Using \eqref{eq:remaindern} and \eqref{eq:tylq}, we obtain that, for every $T \in \mathcal{B}_{D,P}(r_{n}) $ and $r_{n}$ in \eqref{eq:rvalue},
\begin{equation}
	D_{\text{KL}}(T \|Q)=  D_{\text{KL}}(P \| Q) + \sum_{i=1}^{k} (T_{i} -P_{i} ) \ln    \left( \frac{P_{i}}{ Q_{i}} \right)+  \frac{1}{2} d_{\chi^{2}} (T, P) + O\left( \frac{1}{n^{3/2}}\right). \label{eq:tylqb}
\end{equation}
Consequently, there exist $M_{2}>0$ and $\ddot{N}_{2} \in \mathbb{N}$ such that, for all $n \geq  \ddot{N}_{2} $ and  $T \in \mathcal{B}_{D,P}(r_{n}) $,
\begin{IEEEeqnarray}{lCl}
	D _{\text{KL}}(T\|  Q) & \geq &  D_{\text{KL}}(P \| Q) + \sum_{i=1}^{k} (T_{i} -P_{i} ) \ln    \left( \frac{P_{i}}{ Q_{i}} \right)+ \frac{1}{2} d_{\chi^{2}} (T, P)- \frac{ M_{2} }{n^{3/2}} \nonumber  \\
	& \geq &  D_{\textnormal{KL}}(P \| Q) + \sum_{i=1}^{k} (T_{i} -P_{i} ) \ln    \left( \frac{P_{i}}{ Q_{i}} \right)- \frac{ M_{2} }{n^{3/2}}  
	\label{eq:tylq4}
\end{IEEEeqnarray}
where the last inequality follows since the $\chi^2$-divergence $ d_{\chi^{2}} (T, P)$ is nonnegative.

To lower-bound the right-hand side of \eqref{eq:tylq4}, we bound the second term in the right-hand side of \eqref{eq:tylq4}. To this end, we introduce the function
\begin{equation}
	\ell(\Gamma)
	\triangleq \sum_{i=1}^{k} (\Gamma_{i} -P_{i} ) \ln    \left( \frac{P_{i}}{ Q_{i}} \right)  = \mathbf{c}^{\mathsf{T}} (\mathbf{\Gamma} -\mathbf{P}), \quad \Gamma=(\Gamma_{1}, \ldots, \Gamma_{k})^{\mathsf{T}} \in  \mathcal{P}(\mathcal{Z}) \label{eq:hfun1}
\end{equation}
where $ \mathbf{c}=(c_{1},\ldots,c_{k-1})^{\mathsf{T}}$ was defined in \eqref{eq:cidef}. By Lemma~\ref{lemmaball}, we have that, for $n \geq N_{1}$ and $\tilde{P} \in  \mathcal{P}_{n} \cap \mathcal{B}_{D,P}(r_{n})$,
\begin{equation}
	\ell(\tilde{P}) \geq \min_{\Gamma \in \mathcal{P}_{n} \cap \mathcal{B}_{D,P}(r_{n}) } \ell(\Gamma)  \geq \min_{\Gamma \in \bar{\mathcal{B}}_{\bm{A}_{D,\mathbf{P}}}  \left( r_{n}+ \frac{M}{  n^{3/2}}\right) } \ell(\Gamma).  \label{eq:min}
\end{equation}
We next compute the right-most minimum in \eqref{eq:min}. To this end, we define the following quantities:

The $i$-th component of the vector $\bm{A}_{D,\mathbf{P}}^{-1} \mathbf{c}$ is denoted as
\begin{eqnarray}
	b_{i} \triangleq (\bm{A}_{D,\mathbf{P}}^{-1} \mathbf{c})_{i},\quad  1, \ldots,k-1. \label{eq:ithcom}
\end{eqnarray}
We further define the index sets
\begin{IEEEeqnarray}{lCl}
	\mathcal{I} & \triangleq & \left\lbrace   i = 1, \ldots,k-1   \colon b_{i} \neq 0  \right\rbrace \label{eq:index1} \\
	\mathcal{I}_{+} & \triangleq & \left\lbrace   i \in	\mathcal{I}   \colon b_{i} > 0  \right\rbrace. \label{eq:index2} 
\end{IEEEeqnarray}  
Since $\bm{A}_{D,\mathbf{P}}$ is invertible and $\mathbf{c}$ is non-zero (since, by assumption,  $P \neq Q$), there exists an index $ i= 1, \ldots,k-1$ such that $b_i = (\bm{A}_{D,\mathbf{P}}^{-1} \mathbf{c})_{i} \neq 0$. Hence $\mathcal{I}$ is non-empty. We next define
\begin{IEEEeqnarray}{lCl}
	\psi  &\triangleq &
	\underset{ i =1, \ldots,k }{\min{P_{i}}}  \label{eq:index3} \\
	\tau &\triangleq & \max{ \left\lbrace \tau_{1},\tau_{2} \right\rbrace } \label{eq:index3a}
\end{IEEEeqnarray}  
where 
\begin{equation}
	\tau_{1}
	=
	\begin{cases}
	-\sum_{i=1}^{k-1} b_{i},   \quad & \textnormal{if $\sum_{i=1}^{k-1} b_{i} <0$}  \\
	0, \quad & \text{otherwise} \\
	\end{cases} 	\label{eq:index4} 
\end{equation}
and
\begin{equation}
	\tau_{2}
	=
	\begin{cases}
		\underset{i \in \mathcal{I}_{+} }{\max{b_{i}}}, \quad & \textnormal{if $\mathcal{I}_{+} \neq \emptyset$} \\
		0, \quad & \textnormal{if $\mathcal{I}_{+} = \emptyset$} \\
	\end{cases} \label{eq:index4a}
\end{equation}
where $\emptyset$ denotes the empty set. Note that $\tau >0$. To see this, first suppose that $ \mathcal{I}_{+}=\emptyset$. Since $\mathcal{I}$ is non-empty, this implies that $\sum_{i=1}^{k-1} b_{i} <0$. In this case, we have $\tau \geq \tau_{1}= 	-\sum_{i=1}^{k-1} b_{i}>0$.  Now suppose that
$ \mathcal{I}_{+}\neq \emptyset$. This implies that there exists an index $j= 1, \ldots,k-1$ such that $b_{j} >0$, so $\tau_{2}=\underset{i \in \mathcal{I}_{+} }{\max{b_{i}}}>0$. In this case, we have $\tau \geq \tau_{2}>0$.

 \comment{ Now suppose that $\sum_{i=1}^{k-1} b_{i} \geq 0$. As $\mathcal{I}$ is non-empty, this implies that there exists an index $i= 1, \ldots,k-1$ such that $b_{i} >0$, hence $ \mathcal{I}_{+} $ is non-empty. It follows that $\tau_{2}>0$, from which we obtain that $\tau \geq \tau_2>0$.}

We are now ready to characterize the right-most minimum in \eqref{eq:min} as well as the minimizing distribution $\Gamma $:

\begin{lemma} \label{lemmaoptm} Let the threshold $\tilde{r}>0$ satisfy $\sqrt{\tilde{r}} < \frac{\psi}{\tau}  \sqrt{\mathbf{c}^{T} \bm{A}_{D,\mathbf{P}}^{-1} \mathbf{c}}$, where $\psi$ defined in  \eqref{eq:index3} and $\tau$ defined in  \eqref{eq:index3a}. Then, 
\begin{equation} 
		\min_{\Gamma \in \bar{\mathcal{B}}_{\bm{A}_{D,\mathbf{P}}}  \left( \tilde{r} \right)} \ell(\Gamma)   =-\sqrt{\tilde{r}} \sqrt{\mathbf{c}^{\mathsf{T}} \bm{A}_{D,\mathbf{P}}^{-1} \mathbf{c}}.   \label{eq:min1}
	\end{equation} 
Moreover, the probability distribution $\Gamma^{*}=(\Gamma^{*}_{1},\ldots,\Gamma^{*}_{k} )^{\mathsf{T}}$ that minimizes $\ell(\Gamma)$  over $\bar{\mathcal{B}}_{\bm{A}_{D,\mathbf{P}}}  \left( \tilde{r} \right)$ is given by
	\begin{equation}
		\Gamma^{*}_{i} = P_{i} + x_{i}^{*}, \quad i=1,\ldots,k \label{eq:minproba}
	\end{equation} 
	where 
	\begin{IEEEeqnarray}{lCl}
		x_{i}^{*}&=& \frac{-\sqrt{\tilde{r}} b_{i}}{ \sqrt{\mathbf{c}^{\mathsf{T}} \bm{A}_{D,\mathbf{P}}^{-1} \mathbf{c}}},\quad i=1,\ldots,k-1\\ \label{eq:minprobc}
		x_{k}^{*} &=& -\sum_{i=1}^{k-1}x_{i}^{*}  \label{eq:minprobb}
	\end{IEEEeqnarray}	
 with $b_{i}$ as defined in \eqref{eq:ithcom}.
 \end{lemma}
\begin{IEEEproof}
	See Appendix~\ref{prooflemmaoptm}.
\end{IEEEproof}

Lemma~\ref{lemmaoptm} together with \eqref{eq:tylq4} allow us to lower-bound $D_{\textnormal{KL}}(T\|P)$ for every $T \in \mathcal{B}_{D,P}(r_{n}) $. Indeed, let
\begin{equation}
	r'_{n}=r_{n}+ \frac{M}{ n^{3/2}}. \label{eq:rvalueprime}
\end{equation}
Since $r'_{n}$ vanishes as $n\rightarrow \infty$, we can choose an $\ddot{N}$ such that, for $n \geq \ddot{N}$,  $r'_{n}$  satisfies $\sqrt{r'_n} < \frac{\psi}{\tau}  \sqrt{\mathbf{c}^{\mathsf{T}} \bm{A}_{D,\mathbf{P}}^{-1} \mathbf{c}}$. Then, \eqref{eq:min} and Lemma~\ref{lemmaoptm} yield that, for $\tilde{P} \in  \mathcal{P}_{n} \cap \mathcal{B}_{D,P} ( r_{n})$ and  $n \geq N_{2}\triangleq\max\{ N_{1}, \ddot{N}_{2}, \ddot{N}\}$, 
\begin{IEEEeqnarray}{lCl}
	\ell(\tilde{P}) &\geq & \min_{\bar{\mathcal{B}}_{\bm{A}_{D,\mathbf{P}}}  \left( r'_{n} \right) } \ell(\Gamma)  \nonumber  \\
	& =&-\sqrt{r'_{n}} \sqrt{\mathbf{c}^{T} \bm{A}_{D,\mathbf{P}}^{-1} \mathbf{c}}. \label{eq:minvalue}
\end{IEEEeqnarray}
It follows that, for every $\tilde{P}  \in  \mathcal{P}_{n} \cap \mathcal{B}_{D,P} ( r_{n})$,  \eqref{eq:tylq4} can be further lower-bounded as 
\begin{equation}
	D_{\textnormal{KL}}(\tilde{P} \| Q)
	 \geq   D_{\textnormal{KL}}(P \| Q) -\sqrt{r'_{n}} \sqrt{\mathbf{c}^{\mathsf{T}} \bm{A}_{D,\mathbf{P}}^{-1} \mathbf{c}}- \frac{ M_{2} }{n^{3/2}}, \quad n \geq N_{2}. 
	\label{eq:tylq4a}
\end{equation}
Applying \eqref{eq:tylq4a} to \eqref{eq:tyerr}, we obtain that the type-II error is upper-bounded by
\begin{IEEEeqnarray}{lCl}
	\beta_{n}(\mc_{n}^{D}(r_{n})) & \leq&   \sum_{\tilde{P} \in \mathcal{P}_{n} \cap \mathcal{B}_{D,P} (r_{n})}  \exp \{-n D_{\textnormal{KL}}(\tilde{P} \| Q) \} \nonumber  \\
	& \leq & \sum_{\tilde{P} \in \mathcal{P}_{n} \cap \mathcal{B}_{D,P} (r_{n})}  \exp \left\lbrace  -nD_{\textnormal{KL}}(P \| Q)+ n \sqrt{r'_{n}} \sqrt{\mathbf{c}^{\mathsf{T}} \bm{A}_{D,\mathbf{P}}^{-1} \mathbf{c}} +  \frac{M_{2}}{\sqrt{n}} \right\rbrace  \nonumber  \\
	& \leq & (n+1)^{\mid \mathcal{Z}\mid}  \exp \left\lbrace  -nD_{\textnormal{KL}}(P \| Q) + n \sqrt{r'_{n}} \sqrt{\mathbf{c}^{\mathsf{T}} \bm{A}_{D,\mathbf{P}}^{-1} \mathbf{c}}  +  \frac{M_{2}}{\sqrt{n}} \right\rbrace,\quad n \geq N_{2}  \label{eq:tyerror}
\end{IEEEeqnarray}
where the last inequality follows since $| \mathcal{P}_{n} | \leq (n+1)^{| \mathcal{Z}|}$ \cite[Th.~11.1.1]{ATCB}. By taking logarithms on both sides of \eqref{eq:tyerror}, we obtain
\begin{equation}
	\ln \beta_{n}(\mc_{n}^{D}(r_{n})) \leq | \mathcal{Z}| \ln (n+1) -nD_{\textnormal{KL}}(P \| Q) + n \sqrt{r'_{n}} \sqrt{\mathbf{c}^{\mathsf{T}} \bm{A}_{D,\mathbf{P}}^{-1} \mathbf{c}} +  \frac{M_{2}}{\sqrt{n}}, \quad  n \geq N_{2}. \label{eq:sub} 
\end{equation}

We next note that
\begin{equation}
\sqrt{r'_n} \leq \frac{1}{\sqrt{n}}\sqrt{\mathsf{Q}^{-1}_{\chi^{2}_{\bm{\lambda},k-1}}( \epsilon)} + O\left(\frac{\delta_n}{\sqrt{n}}\right) + O\left(\frac{1}{n}\right). \label{eq:sqrt_rn}
\end{equation}
Indeed, by \eqref{eq:rvalue} and \eqref{eq:rvalueprime}, 
\begin{equation}
	\sqrt{ r'_{n}} = \frac{1}{\sqrt{n}} \sqrt{\mathsf{Q}^{-1}_{\chi^{2}_{\bm{\lambda},k-1}} ( \epsilon-M_{0} \delta_{n}) + \frac{M}{\sqrt{n}}}.  \label{eq:sqr}
\end{equation}
A Taylor-series expansion of $\sqrt{x+ \delta}$ around $x>0$ yields then that
\begin{IEEEeqnarray}{lCl}
	\sqrt{ r'_{n}} &\leq& \frac{1}{\sqrt{n}} \left[  \sqrt{\mathsf{Q}^{-1}_{\chi^{2}_{\bm{\lambda},k-1}} ( \epsilon-M_{0} \delta_{n})} +  \frac{M }{ 2 \sqrt{n} \sqrt{\mathsf{Q}^{-1}_{\chi^{2}_{\bm{\lambda},k-1}}( \epsilon-M_{0} \delta_{n})}}   \right] \nonumber \\
	&\leq& \frac{1}{\sqrt{n}}   \sqrt{\mathsf{Q}^{-1}_{\chi^{2}_{\bm{\lambda},k-1}} ( \epsilon-M_{0} \delta_{n})}  +  \frac{M}{ 2 n \sqrt{\mathsf{Q}^{-1}_{\chi^{2}_{\bm{\lambda},k-1}} ( \epsilon) }}    \label{eq:sub1}
\end{IEEEeqnarray}
where the second inequality follows because $\mathsf{Q}^{-1}_{\chi^{2}_{\bm{\lambda},k-1}}(\cdot)$ is monotonically decreasing. Using a Taylor-series expansion of $\sqrt{\mathsf{Q}^{-1}_{\chi^{2}_{\bm{\lambda},k-1}}(\epsilon-M_0\delta_n)}$ around $\epsilon$, we obtain
\begin{equation}
	\sqrt{\mathsf{Q}^{-1}_{\chi^{2}_{\bm{\lambda},k-1}} ( \epsilon-M_{0} \delta_{n})}
	= \sqrt{\mathsf{Q}^{-1}_{\chi^{2}_{\bm{\lambda},k-1}}( \epsilon) } +O\left( \delta_{n} \right) \label{eq:ts}
\end{equation}
which together with \eqref{eq:sub1} yields \eqref{eq:sqrt_rn}.

Applying \eqref{eq:sqrt_rn} to \eqref{eq:sub}, we obtain
\begin{IEEEeqnarray}{lCl}
	\ln \beta_{n}(\mc_{n}^{D}(r_{n})) & \leq & 
	|\mathcal{Z}| \ln (n+1) -nD_{\textnormal{KL}}(P \| Q)  + \sqrt{n} \sqrt{\mathbf{c}^{\mathsf{T}} \bm{A}_{D,\mathbf{P}}^{-1} \mathbf{c}} \sqrt{\mathsf{Q}^{-1}_{\chi^{2}_{\bm{\lambda},k-1}} ( \epsilon) } +O( \delta_{n}\sqrt{n}) + O(1)  +\frac{M_{2}}{\sqrt{n}} \nonumber \\
	& = & 
	-nD_{\textnormal{KL}}(P \| Q) + \sqrt{n} \sqrt{\mathbf{c}^{\mathsf{T}} \bm{A}_{D,\mathbf{P}}^{-1} \mathbf{c}}\sqrt{\mathsf{Q}^{-1}_{\chi^{2}_{\bm{\lambda},k-1}} ( \epsilon) } +O( \max\{ \delta_{n} \sqrt{n}, \ln n \})
\end{IEEEeqnarray}
upon collecting terms that grow at most as fast as $\max\{\delta_n\sqrt{n},\ln n\}$ as $n\to\infty$. This proves Part~1 of Theorem~\ref{divergence}.  
\subsection{Proof of  Part 2 of  Theorem  \ref{divergence}}
Define, for $0 <\epsilon <1$ and $n \in \mathbb{N}$, 
\begin{eqnarray}
	\mathcal{R}_{n}^{\epsilon} \triangleq \left\lbrace  r > 0 \colon P^{n} \left( D(\tp_{Z^{n}} \| P) \geq r \right)  \leq \epsilon \right\rbrace \label{eq:rnset}
\end{eqnarray} 
and
\begin{eqnarray}
	r^{\epsilon}_{n} \triangleq \inf \mathcal{R}_{n}^{\epsilon}. \label{eq:rn}
\end{eqnarray} 
By definition, if $r < r^{\epsilon}_{n}$, then the type-I error  $\alpha_{n}(\mc^{D}(r))$ exceeds $\epsilon$. Hence such a threshold violates
\eqref{eq:typeepsilon}. We thus assume without loss of optimality that $r\geq r^{\epsilon}_{n}$. In this case, 
\begin{equation}
	\beta_{n}(\mc_{n}^{D}(r)) \triangleq	Q^{n} \left(D(\tp_{Z^{n}} \| P)   < r \right)  \geq Q^{n} \left( D(\tp_{Z^{n}} \| P)   < r^{\epsilon}_{n} \right)  . \label{eq:rgeq}
\end{equation}
Consequently,  the type-II error of the divergence test $\mc_{n}^{D}(r)$ can be lower-bounded as
\begin{eqnarray}
	\beta_{n}(\mc_{n}^{D}(r))  &\geq&   Q^{n}(D(\tp_{Z^{n}} \| P)   < r^{\epsilon}_{n})  \nonumber \\
	&= &  \sum_{ z^{n}:\; \tp_{z^{n}} \in \mathcal{B}_{D,P}(r^{\epsilon}_{n})} Q^{n}(z^{n}) \nonumber \\
	&=&\sum_{\tilde{P} \in \mathcal{P}_{n} \cap \mathcal{B}_{D,P}(r^{\epsilon}_{n})}  Q^{n}(\tc(\tilde{P})). \label{eq:sec3ca}
\end{eqnarray}
The following lemma characterizes the asymptotic behavior of $r_n^{\epsilon}$.

\begin{lemma} \label{spectrum} For $\delta_{n}$ given in \eqref{eq:ratekl1}, we have
	\begin{eqnarray}
		r^{\epsilon}_{n} = \frac{1}{n} \mathsf{Q}^{-1}_{\chi^{2}_{\bm{\lambda},k-1}} ( \epsilon)+  O \left( \frac{\delta_{n}}{n}  \right). \label{eq:sec6}
	\end{eqnarray}
\end{lemma}
\begin{proof}
	By \eqref{eq:ratekl1}, we have that, for $n \geq N_{0}$ and every $r>0$,
	\begin{equation}
	\mathsf{Q}_{\chi^{2}_{\bm{\lambda},k-1}}( nr )- M_{0} \delta_{n} \leq P^{n}(D(\tp_{Z^{n}} \| P)   \geq r) \leq 	\mathsf{Q}_{\chi^{2}_{\bm{\lambda},k-1}}( nr )+ M_{0} \delta_{n}.  \label{eq:spec1}
	\end{equation}
	Let
	\begin{equation}
		\dot{r}_{n}= \frac{1}{n}  \mathsf{Q}^{-1}_{\chi^{2}_{\bm{\lambda},k-1}} ( \epsilon- M_{0} \delta_{n}).
	\end{equation} 
	Substituting the value $r=\dot{r}_{n}$ in \eqref{eq:spec1}, we get for $n \geq N_{0}$,
	\begin{equation}
		P^{n}(D(\tp_{Z^{n}} \| P)   \geq \dot{r}_{n})  \leq  \epsilon. \label{eq:spec1rw1}
	\end{equation}
	This implies that $r_{n}^{\epsilon}$ exists and is bounded by
	\begin{IEEEeqnarray}{lCl}
		r^{\epsilon}_{n}  &\leq&  \dot{r}_{n} \nonumber \\
		& =& \frac{1}{n}  \mathsf{Q}^{-1}_{\chi^{2}_{\bm{\lambda},k-1}}( \epsilon- M_{0} \delta_{n}), \quad n \geq N_{0}. \label{eq:rhs}
	\end{IEEEeqnarray}
	Indeed, the set  $\mathcal{R}_{n}^{\epsilon}$ is lower-bounded by $0$. Then, it follows from \eqref{eq:spec1rw1} that, for $n \geq N_{0}$, we have $\dot{r}_{n} \in \mathcal{R}_{n}^{\epsilon}$.  This implies that  $\mathcal{R}_{n}^{\epsilon}$  is a non-empty subset of the real numbers $\mathbb{R}$ that has a lower bound.  Hence, by the completeness property of $\mathbb{R}$,  $r_{n}^{\epsilon}$ exists.  
	
	We next derive a lower bound on $r_n^{\epsilon}$. Let
	\begin{equation}
		\ddot{r}_{n}= \frac{1}{n}  \mathsf{Q}^{-1}_{\chi^{2}_{\bm{\lambda},k-1}}( \epsilon + 2M_{0} \delta_{n}). \label{eq:rtilde}
	\end{equation} 
	Substituting the value of $	\ddot{r}_{n}$ in \eqref{eq:spec1}, we get for $n \geq N_{0}$,
	\begin{equation}
		P^{n}(D(\tp_{Z^{n}} \| P)   \geq 	\ddot{r}_{n})   \geq   \epsilon + M_{0} \delta_{n} 
		>  \epsilon.  \label{eq:spec1rw1a}
	\end{equation}
	Thus, $\ddot{r}_{n}$ does not belong to the set $\mathcal{R}_{n}^{\epsilon}$. It follows from the property of the infimum that 
	\begin{equation}
		r^{\epsilon}_{n} \geq   \frac{1}{n}  \mathsf{Q}^{-1}_{\chi^{2}_{\bm{\lambda},k-1}}( \epsilon + 2M_{0} \delta_{n}), \quad n \geq N_{0}. \label{eq:rhsa}
	\end{equation}
	Thus, we obtain from \eqref{eq:rhs}  and \eqref{eq:rhsa} that $r_n^{\epsilon}$ is bounded by
	\begin{equation}
	\frac{1}{n}  \mathsf{Q}^{-1}_{\chi^{2}_{\bm{\lambda},k-1}}( \epsilon + 2M_{0} \delta_{n})\leq r_{n}^{\epsilon} \leq  \frac{1}{n}  \mathsf{Q}^{-1}_{\chi^{2}_{\bm{\lambda},k-1}} ( \epsilon- M_{0} \delta_{n}), \quad n \geq N_{0}. \label{eq:rnepsilonin}
	\end{equation}
	The lemma follows then from Taylor-series expansions of $ \mathsf{Q}^{-1}_{\chi^{2}_{\bm{\lambda},k-1}}( \epsilon + 2M_{0} \delta_{n}) $ and $  \mathsf{Q}^{-1}_{\chi^{2}_{\bm{\lambda},k-1}}( \epsilon- M_{0} \delta_{n})$ around $\epsilon$.
\end{proof}

Lemma~\ref{spectrum} implies that $r^{\epsilon}_{n} = \Theta(\frac{1}{n})$. By taking $\br =P$, it thus follows from Lemma~\ref{pinskercon}  that, for every \mbox{$T \in \mathcal{B}_{D,P}(r^{\epsilon}_{n})$}, 
\begin{equation}
	\|\mathbf{T}-\mathbf{P}\|_{2}= O\left( \frac{1}{\sqrt{n}}\right).\label{eq:orderd}
\end{equation}
Consequently, for every $T \in \mathcal{B}_{D,P}(r^{\epsilon}_{n})$, \eqref{eq:tayldivA} can be written as 
\begin{equation}
	D(T \| P) = (\mathbf{T}-\mathbf{P})^{\mathsf{T}} \bm{A}_{D,\mathbf{P}} ( \mathbf{T}-\mathbf{P})+O\left( \frac{1}{n^{3/2}}\right) \label{eq:tayldivAa}
\end{equation}
hence there exist $\bar{M}'_{0}>0$ and $\bar{N}'_{0} \in \mathbb{N}$ such that, for all $ n \geq  \bar{N}'_{0} $ and $T \in \mathcal{B}_{D,P}(r^{\epsilon}_{n})$,
\begin{equation}
	|D(T \| P) - (\mathbf{T}-\mathbf{P})^{\mathsf{T}} \bm{A}_{D,\mathbf{P}} ( \mathbf{T}-\mathbf{P}) | < \frac{ \bar{M}'_{0} }{n^{3/2}}. \label{eq:tylp2w}
\end{equation}

We next note that, since $\bm{A}_{D,\mathbf{P}} \succ 0$, the Rayleigh-Ritz theorem \cite[Th.~4.2.2]{RJ90} yields for every $\mathbf{T} \in \mathbb{R}^{k-1}$,
\begin{equation}
	\tilde{\lambda}_{\textnormal{min}} 	\|\mathbf{T}-\mathbf{P}\|^{2}_{2} \leq  (\mathbf{T}-\mathbf{P})^{T} \bm{A}_{D,\mathbf{P}} ( \mathbf{T}-\mathbf{P}) \label{eq:eignbound}
\end{equation}
where  $\tilde{\lambda}_{\textnormal{min}} >0$ is the minimum eigenvalue of the positive-definite matrix $\bm{A}_{D,\mathbf{P}}$. Furthermore, for  \mbox{$T \in	\bar{\mathcal{B}}_{\bm{A}_{D,\mathbf{P}}}\left( r^{\epsilon}_{n}-\frac{\bar{M}'_{0}}{ n^{3/2}}\right)$}, we have
\begin{equation}
    (\mathbf{T}-\mathbf{P})^{\mathsf{T}} \bm{A}_{D,\mathbf{P}} ( \mathbf{T}-\mathbf{P})\leq r^{\epsilon}_{n}-\frac{\bar{M}'_{0}}{n^{3/2}}. \label{eq:ADP_bound}
\end{equation}
Thus, combining \eqref{eq:eignbound} and \eqref{eq:ADP_bound}, and again using that, by Lemma~\ref{spectrum}, $r^{\epsilon}_{n} = \Theta(\frac{1}{n})$, we obtain that, for every $T \in\bar{\mathcal{B}}_{\bm{A}_{D,\mathbf{P}}}\left( r^{\epsilon}_{n}-\frac{\bar{M}'_{0}}{ n^{3/2}}\right)$,
\begin{equation}
    \|\mathbf{T}-\mathbf{P}\|_{2}= O\left( \frac{1}{\sqrt{n}}\right).
\end{equation}
It follows  from \eqref{eq:tayldivA} that  there exist $\bar{M}'_{1}>0$ and $\bar{N}'_{1} \in \mathbb{N}$ such that, for all $ n \geq  \bar{N}'_{1}$ and   $T \in	\bar{\mathcal{B}}_{\bm{A}_{D,\mathbf{P}}}\left( r^{\epsilon}_{n}-\frac{\bar{M}'_{0}}{ n^{3/2}}\right)$, 
\begin{equation}
	|D(T \| P) - (\mathbf{T}-\mathbf{P})^{\mathsf{T}} \bm{A}_{D,\mathbf{P}} ( \mathbf{T}-\mathbf{P}) | < \frac{ \bar{M}'_{1} }{n^{3/2}}. \label{eq:tylp2bw}
\end{equation}
This allows us to obtain the following  lemma.
\begin{lemma} \label{lemmaballw} For $ n \geq N'_{1}\triangleq\max \{ \bar{N}'_{0}, \bar{N}'_{1}\}$, we have
	\begin{equation}
	\bar{\mathcal{B}}_{\bm{A}_{D,\mathbf{P}}}\left( r^{\epsilon}_{n}-\frac{M'_{1}}{ n^{3/2}}\right)\subseteq \mathcal{B}_{D,P}(r^{\epsilon}_{n})\label{eq:klchiw}
	\end{equation}
	where $M'_{1}\triangleq\max \{ \bar{M}'_{0}, \bar{M}'_{1}\}$ and  $r^{\epsilon}_{n}$ is as in \eqref{eq:sec6}.
\end{lemma}
\begin{IEEEproof} 
	Since $M'_{1}=\max \{ \bar{M}'_{0}, \bar{M}'_{1}\}$,  we have $\bar{M}'_{0}\leq M'_{1}$. This implies that 
	\begin{equation}
	\bar{\mathcal{B}}_{\bm{A}_{D,\mathbf{P}}}\left( r^{\epsilon}_{n}-\frac{M'_{1}}{n^{3/2}}\right)  \subseteq \bar{\mathcal{B}}_{\bm{A}_{D,\mathbf{P} }} \left( r^{\epsilon}_{n}-\frac{\bar{M}'_{0}}{ n^{3/2}} \right).
	\end{equation}
	Then, using  \eqref{eq:tylp2bw}, we have for $T \in \bar{\mathcal{B}}_{\bm{A}_{D,\mathbf{P}}} \left( r^{\epsilon}_{n}-\frac{M'_{1}}{ n^{3/2}}\right)$ and $n \geq N'_{1}=\max \{ \bar{N}'_{0}, \bar{N}'_{1}\}$ that
	\begin{IEEEeqnarray}{lCl}
		D(T \| P)  &< &   (\mathbf{T}-\mathbf{P})^{\mathsf{T}} \bm{A}_{D,\mathbf{P}} ( \mathbf{T}-\mathbf{P})+ \frac{ \bar{M}'_{1} }{n^{3/2}} \notag\\
		&\leq &  \left(r^{\epsilon}_{n}-\frac{M'_{1}}{ n^{3/2}}\right) +\frac{ M'_{1} }{n^{3/2}}\notag\\
		&\leq& r^{\epsilon}_{n}
	\end{IEEEeqnarray}
	hence $T \in \mathcal{B}_{D,P}(r^{\epsilon}_{n})$.
\end{IEEEproof}

Applying Lemma~\ref{lemmaballw} to \eqref{eq:sec3ca}, we obtain that, when $r \geq r^{\epsilon}_{n}$ and $n \geq N_{1}'$, the type-II error of the divergence test $\mc_{n}^{D}(r)$ can be lower-bounded as
\begin{IEEEeqnarray}{lCl}
	\beta_{n}(\mc_{n}^{D}(r))  &\geq& \sum_{\tilde{P} \in \mathcal{P}_{n} \cap \mathcal{B}_{D,P}(r^{\epsilon}_{n})}  Q^{n}(\tc(\tilde{P})) \notag \\
	&\geq & \sum_{\tilde{P} \in \mathcal{P}_{n} \cap 	\bar{\mathcal{B}}_{\bm{A}_{D,\mathbf{P}}} \left(  \bar{r}_{n}\right) }  Q^{n}(\tc(\tilde{P})) \label{eq:sec3c}
\end{IEEEeqnarray}
where  
\begin{equation}
	\bar{r}_{n}=r^{\epsilon}_{n}-\frac{M'_{1}}{ n^{3/2}}. \label{eq:rbar}
\end{equation}

Similar to the proof of Part~1 of Theorem~\ref{divergence}, the right-hand side of \eqref{eq:sec3c} can be bounded by $\exp\{-n D_{\textnormal{KL}}(P_n^*\|Q)\}$ for some type distribution $P_n^* \in \mathcal{P}_{n} \cap \bar{\mathcal{B}}_{\bm{A}_{D,\mathbf{P}}}(\bar{r}_{n})$. We then note that, by Lemma~\ref{ctaylor} and \eqref{eq:orderd}, there exist $M_{2}'>0$ and $N_{2}' \in \mathbb{N}$ such that, for $T \in \mathcal{B}_{D,P}(r^{\epsilon}_{n})$, 
\begin{equation}
	\left|D_{\textnormal{KL}} (T \|  Q) -  D_{\textnormal{KL}}(P \| Q) - \sum_{i=1}^{k} (T_{i} -P_{i} ) \ln    \left( \frac{P_{i}}{ Q_{i}} \right)-  \frac{1}{2} d_{\chi^{2}} (T, P) \right| \leq \frac{ M_{2}' }{n^{3/2}},\quad n \geq  N_{2}'.  \label{eq:tylq2c}
\end{equation}
To upper-bound $D_{\textnormal{KL}}(P_n^*\|Q)$, we thus need to evaluate
\begin{equation}
\ell(T) = \sum_{i=1}^{k} (T_{i} -P_{i} ) \ln    \left( \frac{P_{i}}{ Q_{i}} \right)
\end{equation}
for $T=P_n^*$. We shall do so in Lemma~\ref{type} below. The probability distribution $\Gamma^*$ in Lemma~\ref{type} is well-defined if
\begin{equation}
 \sqrt{\bar{r}_{n}} < \frac{\psi}{\tau}  \sqrt{\mathbf{c}^{\mathsf{T}} \bm{A}_{D,\mathbf{P}}^{-1} \mathbf{c}} \label{eq:rnbarn}
\end{equation}
where $\psi$ and $\tau$ were defined in \eqref{eq:index3} and \eqref{eq:index3a}, respectively. Fortunately, we can find a sufficiently large $\tilde{N}'\in\mathbb{N}$ such that \eqref{eq:rnbarn} holds for $n \geq \tilde{N}'$, since, by Lemma~\ref{spectrum} and \eqref{eq:rbar}, $\bar{r}_{n}$ vanishes as $n \to \infty$.

\begin{lemma}\label{type}  Consider the probability distribution $\Gamma^{*}=(\Gamma^{*}_{1},\ldots,\Gamma^{*}_{k} )^{\mathsf{T}}$ that minimizes $\ell(\Gamma)$  over $\bar{\mathcal{B}}_{\bm{A}_{D,\mathbf{P}}}  \left( \bar{r}_{n} \right)$, defined in \eqref{eq:minproba}--\eqref{eq:minprobc}, namely
	\begin{equation}
		\Gamma^{*}_{i} = P_{i} + x_{i}^{*} \quad i=1,\ldots,k \label{eq:z1}
	\end{equation}
	where 
	\begin{IEEEeqnarray}{lCl}
		x_{i}^{*} &=& \frac{-\sqrt{\bar{r}_{n}  } b_{i}}{ \sqrt{\mathbf{c}^{\mathsf{T}} \bm{A}_{D,\mathbf{P}}^{-1} \mathbf{c}}},\quad i=1,\ldots,k-1. \label{eq:minprb2} \\
		x_{k}^{*} &=&-\sum_{i=1}^{k-1}x_{i}^{*} \label{eq:minprb3} 
	\end{IEEEeqnarray}
	with $b_{i}$ given  in \eqref{eq:ithcom}. Then, there exist $\tilde{N} \geq \tilde{N}'$ and a type distribution $\tp^{*}_{n}\in \mathcal{P}_{n} \cap \bar{\mathcal{B}}_{\bm{A}_{D,\mathbf{P}}}  \left( \bar{r}_{n} \right)$ such that
	\begin{equation}
		| n \ell(	\Gamma^{*}) -n\ell(\tp^{*}_{n})| \leq \kappa, \quad n \geq \tilde{N} \label{eq:hfunbound}
	\end{equation}
	for some constant $\kappa >0$.
\end{lemma}
\begin{IEEEproof} See Appendix~\ref{lemmatype}.
\end{IEEEproof}

We next lower-bound \eqref{eq:sec3c} using \eqref{eq:tylq2c} and Lemma~\ref{type}. Indeed, we have that
\begin{IEEEeqnarray}{lCl}
	\beta_{n}(\mc_{n}^{D}(r))  &\geq&  \sum_{\tilde{P} \in \mathcal{P}_{n} \cap \bar{\mathcal{B}}_{\bm{A}_{D,\mathbf{P}}} \left(  \bar{r}_{n}\right) }  Q^{n}(\tc(\tilde{P})) \nonumber \\
	& \geq &  Q^{n}(\tc(\tp^{*}_{n})) \nonumber \\
	& \geq&  \frac{1}{(n+1)^{|\mathcal{Z}|}} \exp \{-n D_{\textnormal{KL}}(\tp^{*}_{n} \| Q) \} \label{eq:prtyclc} 
\end{IEEEeqnarray}
for $n \geq \tilde{N}_{0} \triangleq \max \{N_{1}', N_{2}',\tilde{N}\}$ and some type distribution
$\tp^{*}_{n}  \in \mathcal{P}_{n} \cap  \bar{\mathcal{B}}_{\bm{A}_{D,\mathbf{P}}} \left(  \bar{r}_{n}\right)$ satisfying \eqref{eq:hfunbound}. The last inequality in 
\eqref{eq:prtyclc} follows from \cite[Th.~11.1.4]{ATCB}.

To further lower-bound $\beta_{n}(\mc_{n}^{D}(r))$, we use \eqref{eq:tylq2c} and \eqref{eq:hfunbound}, and we recall that
\begin{equation*}
    \ell(\Gamma^{*}) =- \sqrt{\bar{r}_{n}} \sqrt{\mathbf{c}^{\mathsf{T}} \bm{A}_{D,\mathbf{P}}^{-1} \mathbf{c}}    
\end{equation*}
to upper-bound $n D _{\textnormal{KL}}(\tp^{*}_{n} \| Q)$ as
\begin{IEEEeqnarray}{lCl}
	n D _{\textnormal{KL}}(\tp^{*}_{n} \|  Q)  & \leq &  nD_{\textnormal{KL}}(P \| Q) + n\ell(\tp^{*}_{n}) + \frac{1}{2} n d_{\chi^{2}} (\tp^{*}_{n}, P) + \frac{M'_{2}}{\sqrt{n}} \nonumber \\
	& \leq &  nD_{\textnormal{KL}}(P \| Q) + n \ell(\Gamma^{*}) +\kappa  + \frac{1}{2} n d_{\chi^{2}} (\tp^{*}_{n}, P) + \frac{M'_{2}}{\sqrt{n}}  \nonumber \\
	& \leq & nD_{\textnormal{KL}}(P \| Q) - n \sqrt{\bar{r}_{n}} \sqrt{\mathbf{c}^{\mathsf{T}} \bm{A}_{D,\mathbf{P}}^{-1} \mathbf{c}	}  + \kappa + \frac{1}{2} n d_{\chi^{2}} (\tp^{*}_{n}, P) + \frac{M'_{2}}{\sqrt{n}}, \quad n \geq \tilde{N}_{0}. \label{eq:sec2}
\end{IEEEeqnarray}
By the Rayleigh-Ritz theorem \cite[Th.~4.2.2]{RJ90}, the $\chi^2$-divergence, defined in \eqref{eq:chisddef}, can be upper-bounded as 
\begin{equation}
 d_{\chi^{2}} (T, P)  \leq \lambda'_{\textnormal{max}} 	\|\mathbf{T}-\mathbf{P}\|^{2}_{2} \label{eq:eignbounds}
\end{equation}
where $\lambda'_{\textnormal{max}} >0$ is the maximum eigenvalue of the positive-definite matrix $\bm{\Sigma}_{\mathbf{P}}$. Since $\tp_{n}^{*} \in \bar{\mathcal{B}}_{\bm{A}_{D,\mathbf{P}}} \left(  \bar{r}_{n}\right)$, we have from \eqref{eq:orderd} that
$	\|\mathbf{\tp}_{n}^{*} -\mathbf{P}\|_{2} =O(1/\sqrt{n})$.  It follows that $n d_{\chi^{2}} (\tp^{*}_{n}, P)=O(1)$, i.e., there exist $\check{M}>0$ and $\check{N} \in \mathbb{N}$ such that  
\begin{equation}
n d_{\chi^{2}} (\tp^{*}_{n}, P) \leq \check{M},\quad n \geq \check{N}. \label{eq:chsbounds}
\end{equation}
Together with \eqref{eq:sec2}, this yields
\begin{equation}
	n D _{\textnormal{KL}}(\tp^{*}_{n} \|  Q)  \leq 
nD_{\textnormal{KL}}(P \| Q) - n \sqrt{\bar{r}_{n}} \sqrt{\mathbf{c}^{\mathsf{T}} \bm{A}_{D,\mathbf{P}}^{-1} \mathbf{c}}  + \kappa + \frac{\check{M}}{2}  + \frac{M'_{2}}{\sqrt{n}} \label{eq:secaa}
\end{equation}
for $n \geq \tilde{N}_{1} \triangleq \max \left\lbrace \check{N}, \tilde{N}_{0} \right\rbrace$. Applying \eqref{eq:secaa} to \eqref{eq:prtyclc}, and taking logarithms on both sides of the inequality, we obtain that
\begin{equation}
	\ln \beta_{n}(\mc_{n}^{D}(r))  \geq -|\mathcal{Z}| \ln (n+1)  -nD_{\textnormal{KL}}(P \| Q) + n \sqrt{\bar{r}_{n}} \sqrt{\mathbf{c}^{\mathsf{T}} \bm{A}_{D,\mathbf{P}}^{-1} \mathbf{c}}- \kappa -  \frac{\check{M}}{2} - \frac{M'_{2}}{\sqrt{n}}.  \label{eq:sec5}
\end{equation}
Note that, by \eqref{eq:sec6} and \eqref{eq:rbar},  
\begin{equation}
	\bar{r}_{n}
	= \frac{1}{n}  \mathsf{Q}^{-1}_{\chi^{2}_{\bm{\lambda},k-1}}(\epsilon) +  O \left( \frac{\delta_{n}}{n}  \right) -\frac{M'_{1}}{ n^{3/2}}.  \label{eq:sec8c}
\end{equation}
A Taylor-series expansion of $\sqrt{\bar{r}_n}$ around $\frac{1}{n}  \mathsf{Q}^{-1}_{\chi^{2}_{\bm{\lambda},k-1}}(\epsilon)$ then gives
\begin{equation}
	\sqrt{\bar{r}_{n}} = \frac{1}{\sqrt{n}}  \sqrt{\mathsf{Q}^{-1}_{\chi^{2}_{\bm{\lambda},k-1}}(\epsilon)} + O\left(\frac{\delta_n}{\sqrt{n}}\right) +  O \left( \frac{1}{n} \right).   \label{eq:sec8}
\end{equation}
Applying \eqref{eq:sec8} to \eqref{eq:sec5}, and collecting terms that grow at most as fast as $\max\{ \delta_{n} \sqrt{n}, \ln n\}$, we thus obtain 
\begin{equation}
    \ln \beta_{n}(\mc_{n}^{D}(r)) \geq  
	-nD_{\textnormal{KL}}(P \| Q) + \sqrt{n}\sqrt{\mathbf{c}^{\mathsf{T}} \bm{A}_{D,\mathbf{P}}^{-1} \mathbf{c}}  \sqrt{   \mathsf{Q}^{-1}_{\chi^{2}_{\bm{\lambda},k-1}}(\epsilon)}  +O( \max\{ \delta_{n} \sqrt{n}, \ln n\}).
\end{equation}
This proves Part 2 of Theorem~\ref{divergence}.

\section{Summary and Conclusions}
\label{sec:conclusion}

We considered a binary hypothesis testing problem where the hypothesis test has access to the null hypothesis $P$ but not to the alternative hypothesis $Q$. A suitable test for this case is  the Hoeffding test, which accepts  $P$ if $D_{\textnormal{KL}}(\tp_{Z^{n}} \|P) < r$,  for some $r>0$, and otherwise accepts $Q$. We characterized the second-order asymptotic behavior of the type-II error $\beta_n$ of this test in Stein's regime, where the type-I error $\alpha_n$ is required to be bounded by some $\epsilon$. We further proposed a generalization of the Hoeffding test, termed divergence test, for which the KL divergence is replaced by some other (arbitrary) divergence, and characterized its second-order asymptotic behavior of the type-II error $\beta_n$ in Stein's regime. 

Our results show that, irrespective of the choice of divergence, the divergence test achieves the same first-order term as the Neyman-Pearson test. Thus, like the Hoeffding test, every divergence test is first-order optimal. However, the second-order term of the divergence test, and hence also of the Hoeffding test, is strictly smaller than the second-order term of the Neyman-Pearson test and scales unfavorably with the cardinality of $P$ and $Q$. Consequently, lack of knowledge of the alternative hypothesis $Q$ negatively impacts the second-order term.

We further demonstrated that the divergence test achieves the same second-order term as the Hoeffding test whenever the divergence belongs to the class of invariant divergences. The class of invariant divergences is large and includes the Rényi divergence and the $f$-divergences (and, hence, also the KL divergence). In contrast, when the divergence is non-invariant, the second-order term achieved by the divergence test differs from the second-order term achieved by the Hoeffding test and may be strictly larger for some alternative hypotheses $Q$. Potentially, this behavior could be exploited by a composite hypothesis test with partial knowledge of the alternative hypothesis $Q$ by tailoring the divergence of the divergence test to the set of possible alternative hypotheses $\mathcal{Q}$. 

\appendices

\section{Proof of Lemma~\ref{convgdiv}} \label{proofconvgdiv}
For an i.i.d.\ sequence $Z^{n}$ with distribution $ P \in \mathcal{P}(\mathcal{Z})$ taking value in the set $\mathcal{Z}=\{a_1,\ldots,a_k\}$, consider the $(k-1)$-dimensional random vector $\mathbf{\tp}_{Z^{n}}=(\tp_{Z^{n}}(a_{1}), \ldots, \tp_{Z^{n}}(a_{k-1}) )^{\mathsf{T}}$, and let
\begin{equation}
	\mathbf{X}_i=  \left( \mathbf{1}\{Z_i= a_{1} \}, \ldots, \mathbf{1}\{Z_i= a_{k-1} \} \right)^{\mathsf{T}}. \label{eq:rviid}
\end{equation}
The expectation vector and covariance matrix of  $\mathbf{X}_i$ are given by 
\begin{IEEEeqnarray}{rCl}
	E[\mathbf{X}_i] &=& \mathbf{P} \\
	\textnormal{Cov}(\mathbf{X}_i) &=&\bm{\Sigma}^{-1}_{\mathbf{P}}
\end{IEEEeqnarray}  
where $\bm{\Sigma}^{-1}_{\mathbf{P}}$ is the inverse of the matrix $\bm{\Sigma}_{\mathbf{P}}$ and has entries
\begin{equation}
	\bm{\Sigma}^{-1}_{\mathbf{P}}(i,j) =
	\begin{cases}
		P_{i}(1-P_{i}), \quad &  \text{for $i=j$} \\
		-P_{i}P_{j}, \quad &\text{for $i\neq j$.}
	\end{cases} \label{eq:covsigma}
\end{equation}
Since the sequence $\mathbf{X}_{1}, \ldots, \mathbf{X}_{n} $ is i.i.d., it follows from the central limit theorem \cite[Sec.~VIII.4]{feller71} that 
\begin{equation}
	\sqrt{n}\left(  \frac{1}{n} \sum_{i=1}^{n}{\mathbf{X}_{i}} -\mathbf{P} \right) \xrightarrow{d}\mathbf{V} \sim \mathcal{N}(\bm{0}, \bm{\Sigma}^{-1}_{\mathbf{P}}), \quad \textnormal{as $n\to\infty$.}
\end{equation}
Noting that $\frac{1}{n}\sum_{i=1}^{n}{\mathbf{X}_{i}}=\mathbf{\tp}_{Z^{n}}$, we  thus get that
\begin{equation}
	\mathbf{V}_{n}=	\sqrt{n}\left(  \mathbf{\tp}_{Z^{n}} -\mathbf{P} \right) \xrightarrow{d} \mathcal{N}(\bm{0}, \bm{\Sigma}^{-1}_{\mathbf{P}}) \quad \textnormal{as $n\to\infty$.} \label{eq:typecd}
\end{equation}
Together with  \cite[Lemma~5.3]{AAN209}, this implies that $\mathbf{V}_{n}=O_{P}(1)$, so
\begin{equation}
	\mathbf{\tp}_{Z^{n}}= \mathbf{P} + O_{P} \left( \frac{1}{\sqrt{n}}\right). 	\label{eq:ortype1} 
\end{equation}

To characterize the convergence of $D(\tp_{Z^{n}}\| P)$, we need the following proposition.
\begin{proposition}\label{rvprop}
Let $\{\mathbf{X}_{n}\}$ be a sequence of random vectors in $\mathbb{R}^{k-1}$, and let $\mathbf{a} \in \mathbb{R}^{k-1}$ be a constant vector, such that  
\begin{equation}
	\mathbf{X}_{n}=\mathbf{a}+ O_{P}(r_{n})
\end{equation}
 for some positive sequence $\{r_n\}$ satisfying $r_n\to 0$ as $n\to\infty$. Further let $g \colon\mathbb{R}^{k-1} \to \mathbb{R}$ be a twice continuously-differentiable function in some neighborhood $\Omega$ of $\mathbf{a}$. Then, we have
	\begin{equation}
		g(	\mathbf{X}_{n} ) = g(\mathbf{a} )+ \sum_{\ell=1}^{k-1}\frac{\partial g}{\partial x_{\ell}} (\mathbf{a})\bigl(\mathbf{X}_{n}(\ell) -a_{\ell}\bigr) +  \frac{1}{2}\sum_{\ell,m=1}^{k-1} \frac{\partial^{2} g}{\partial x_{\ell} x_{m}} (\mathbf{a}) \bigl(\mathbf{X}_{n}(\ell) -a_{\ell}\bigr)\bigl(\mathbf{X}_{n}(m) -a_{m}\bigr)+o_{P}(r_{n}^{2})
		\label{eq:ts4a}
	\end{equation}
	where $\mathbf{X}_{n}(\ell)$ is the $\ell$-th component of $\mathbf{X}_{n}$ and $a_{\ell}$ is the $\ell$-th component of $\mathbf{a}$.
\end{proposition}
\begin{IEEEproof}
 The proof follows along the same line as that of \cite[Prop.~6.1.1]{BDB91} by replacing the linear Taylor approximation of the function $g$ at $\mathbf{a}$ by the quadratic approximation.
 \end{IEEEproof}

Since $\bm{T}\mapsto D(\bm{T} \| \bm{P})$ is a smooth function of $\bm{T}$, and since \eqref{eq:ortype1} holds, it follows from Proposition~\ref{rvprop} that
\begin{IEEEeqnarray}{lCl}
 D( \tp_{Z^{n}}\| P)   &=& \left( \mathbf{\tp}_{Z^{n}}-\mathbf{P}\right) ^{\mathsf{T}}  \bm{A}_{D,\mathbf{P}} \left( \mathbf{\tp}_{Z^{n}}-\mathbf{P} \right) +o_{P}\left( \frac{1}{n}\right) \notag\\
	&= & \frac{1}{n} \mathbf{V}_{n}^{\mathsf{T}}  \bm{A}_{D,\mathbf{P}} \mathbf{V}_{n}+o_{P}\left( \frac{1}{n}\right) \label{eq:divstat}
\end{IEEEeqnarray}
where the last equality follows from \eqref{eq:typecd}. 

Define the quadratic form 
\begin{eqnarray}
	q_{ \bm{A}_{D,\mathbf{P}}}(\mathbf{v})\overset{\Delta}{=}  \mathbf{v}^{\mathsf{T}}  \bm{A}_{D,\mathbf{P}} \mathbf{v}, \quad \mathbf{v} \in \mathbb{R}^{k-1}.
\end{eqnarray}
Note that $\mathbf{v} \mapsto q_{ \bm{A}_{D,\mathbf{P}}}(\mathbf{v})$ is a continuous function of $\mathbf{v}$. It thus follows from \eqref{eq:typecd} and \cite[Prop.~6.3.4]{BDB91} that 
\begin{equation}
	q_{ \bm{A}_{D,\mathbf{P}}}(\mathbf{V}_{n})  \xrightarrow{d}  q_{ \bm{A}_{D,\mathbf{P}}}(	\mathbf{V}). \label{eq:comt}
\end{equation}

We next demonstrate that $q_{ \bm{A}_{D,\mathbf{P}}}(\mathbf{V})$ has a generalized chi-square distribution. Indeed, let the unique positive square-root of $\bm{\Sigma}^{-1}_{\mathbf{P}}$ be denoted  by $\bm{\Sigma}_{\mathbf{P}}^{-1/2}$, and note that, since $\bm{\Sigma}_{\mathbf{P}}^{-1/2}  \bm{A}_{D,\mathbf{P}} \bm{\Sigma}_{\mathbf{P}}^{-1/2} $ is positive-definite, there exists an orthogonal matrix $S$ satisfying 
\begin{equation}
	\bm{\Sigma}_{\mathbf{P}}^{-1/2}  \bm{A}_{D,\mathbf{P}} \bm{\Sigma}_{\mathbf{P}}^{-1/2}=S  \Lambda S^{\mathsf{T}}
\end{equation}
where $\Lambda$ is the diagonal matrix consisting of the eigenvalues of $\bm{\Sigma}_{\mathbf{P}}^{-1/2}  \bm{A}_{D,\mathbf{P}} \bm{\Sigma}_{\mathbf{P}}^{-1/2}$.  We can therefore transform the Gaussian vector $\mathbf{V}$ into a standard Gaussian vector $\mathbf{U} =S^{\mathsf{T}} \bm{\Sigma}_{\mathbf{P}}^{1/2} \mathbf{V}$ and express the quadratic form  $q_{ \bm{A}_{D,\mathbf{P}}}(	\mathbf{V})$ as
\begin{equation}
q_{ \bm{A}_{D,\mathbf{P}}}(	\mathbf{V})= \mathbf{U}^{\mathsf{T}} \Lambda \mathbf{U}.  \label{eq:eigenquad}
\end{equation}
Denoting the eigenvalues of $\bm{\Sigma}_{\mathbf{P}}^{-1/2}  \bm{A}_{D,\mathbf{P}} \bm{\Sigma}_{\mathbf{P}}^{-1/2}$  by  $\lambda_{i}$, $i=1, \ldots,k-1$, this, in turn, can be written as
\begin{equation}
	q_{ \bm{A}_{D,\mathbf{P}}}(	\mathbf{V})= \sum_{i=1}^{k-1} \lambda_{i} U_{i}^{2}
\end{equation}
where $U_{1}, \ldots, U_{k-1}$ are independent standard normal random variables. Thus, taking $\Upsilon_{i}=U_{i}^{2}$, $i=1,\ldots,k-1$, it follows from Definition~\ref{genchisrv} that
\begin{equation}
    q_{ \bm{A}_{D,\mathbf{P}}}(	\mathbf{V})= \sum_{i=1}^{k-1} \lambda_{i} \Upsilon_{i}
\end{equation}
has a generalized chi-square distribution with vector parameter $\bm{\lambda}=(\lambda_{1},\ldots,\lambda_{k-1})^{\mathsf{T}}$ and degrees of freedom $k-1$. We therefore obtain from \eqref{eq:divstat}, \eqref{eq:comt}, and  \cite[Prop.~6.3.3]{BDB91} that
\begin{equation}
	n D( \tp_{Z^{n}}\| P)
	=  q_{ \bm{A}_{D,\mathbf{P}}}(\mathbf{V}_{n}) +o_{P}\left( 1 \right)
\end{equation}
converges in distribution to the generalized chi-square random variable $q_{ \bm{A}_{D,\mathbf{P}}}(\mathbf{V})$. 

 Let $\mathsf{F}_{n}(\cdot)$ denote the cumulative distribution function (CDF) of $n D( \tp_{Z^{n}}\| P)$, and let $ \mathsf{F}_{\chi^{2}_{\bm{\lambda},k-1}}(\cdot)$ denote the CDF of the generalized chi-square distribution with vector parameter $\bm{\lambda}=(\lambda_{1},\ldots,\lambda_{k-1})^{\mathsf{T}}$ and degrees of freedom $k-1$. Since the CDF $ \mathsf{F}_{\chi^{2}_{\bm{\lambda},k-1}}(\cdot)$ is continuous, convergence in distribution of the statistic $n D( \tp_{Z^{n}}\| P)$ to the generalized chi-square random variable $\chi^2_{\bm{\lambda},k-1}$ implies that $\mathsf{F}_{n}(\cdot)$ converges to $ \mathsf{F}_{\chi^{2}_{\bm{\lambda},k-1}}(\cdot)$ pointwise. It then follows by \cite[Th.~7.6.2]{RC199} that this convergence is, in fact, uniform. \comment{Thus, for every $\varepsilon>0$, there exists an $N_{\epsilon} \in \mathbb{N}$ such that, for all $n \geq N_{\varepsilon}$,
	\begin{equation}
		\sup_{c > 0} 	\left| 	\mathsf{F}_{n}(c)- 	 \mathsf{F}_{\chi^{2}_{\bm{\lambda},k-1}}(c) \right| \leq \varepsilon. \label{eq:ratediv2d2}
	\end{equation}}
	This implies that the tail probability of the statistic $ n	D(\tp_{Z^{n}} \| P)$ can be approximated as 
	\begin{equation}
		P^{n}(n	D (\tp_{Z^{n}} \| P) \geq c)=  \mathsf{Q}_{\chi^{2}_{\bm{\lambda},k-1}}(c) + O(\delta_{n}) \label{eq:ratediv2d}
	\end{equation}
	for all $c >0$ and some positive sequence  $\delta_{n}$ that is independent of $c$ and satisfies  $\lim_{n \rightarrow \infty} \delta_{n}=0$. This proves Lemma~\ref{convgdiv}.

\section{Proof of Lemma~\ref{pinskercon} } \label{proofpinsker}
Recall that $f_R(\mathbf{T})\triangleq D(T\|R)$. We shall first show that, for every sequence $\{\mathbf{T}_{n}\}$ in $\bd$ such that $f_{\br}(\mathbf{T}_{n}) \rightarrow 0$  as $n \rightarrow \infty$, we have $\mathbf{T}_{n}\rightarrow \mathbf{\br}$  as $n \rightarrow \infty$.
Note that, by the definition in \eqref{eq:domaindef}, the set $\bd$  is a bounded set in $\mathbb{R}^{k-1}$. Consequently, the closure of $\bd$, denoted by  $\overline{\bd}$,  is a compact set in $\mathbb{R}^{k-1}$. It follows from the Bolzano-Weierstrass theorem that $\{\mathbf{T}_{n}\}$ has a converging subsequence $\{\mathbf{T}_{n_i}\}$, and we denote the limit of this subsequence by $\mathbf{Q}' \in \overline{\bd}$.  If $\mathbf{Q}'$ lies on the boundary $\partial \bd$ then, by \eqref{eq:divdefexcon},
\begin{equation}
\liminf_{i  \rightarrow \infty}	f_{\br}(\mathbf{T}_{n_i})>0.  \label{eq:divdefexcon1}
\end{equation}
 Similarly, if $\mathbf{Q}' \in  \bd $, then
\begin{equation}
\lim_{i  \rightarrow \infty}	f_{\br}(\mathbf{T}_{n_i}) = f_{\br}(\mathbf{Q}') \geq 0  \label{eq:divdefexcon2}
\end{equation}
with equality if, and only if, $\mathbf{Q}' = \mathbf{\br}$. Here, the limit exists and is equal to $f_{\br}(\mathbf{Q}')$ because $f_{\br}(\cdot)$ is continuous on $\bd$. Since  $f_{\br}(\mathbf{T}_{n}) \to 0$ as $n \to \infty$ by assumption, \eqref{eq:divdefexcon1}--\eqref{eq:divdefexcon2} imply that that every convergent subsequence of $\{\mathbf{T}_{n}\}$ must converge to $\mathbf{\br}$. 
However, if every converging subsequence of $\{\mathbf{T}_{n}\}$ converges to $\mathbf{\br}$ as $n \to \infty$, then $\{\mathbf{T}_{n}\}$ also converges to $\mathbf{\br}$ \cite[Th.~3.4.9]{BS200}. Hence, $ D(T_n \| \br) \rightarrow 0$ implies that $\| \mathbf{T}_n-\mathbf{\br} \|_{2} \rightarrow 0$.

We next show that, if  $f_{\br}(\mathbf{T}_n)=O(b_{n})$  for a strictly positive sequence $\{b_n\}$ satisfying $\lim_{n \rightarrow \infty} b_{n}=0$, then $\|\mathbf{T}_n-\mathbf{\br}\|_{2} = O(\sqrt{b_n})$. By the above argument, if $f_{\br}(\mathbf{T}_n)=O(b_{n})$ for $b_n \to 0$, then $\|\mathbf{T}_n-\mathbf{\br}\|_{2} \rightarrow 0$. It follows from \eqref{eq:tayldivA} that there exist $\delta>0$ and  $M>0$ such that  
\begin{equation}
	|f_{\br}(\mathbf{T}_n) -(\mathbf{T}_n-\mathbf{\br})^{\mathsf{T}} \bm{A}_{D,\mathbf{\br}} ( \mathbf{T}_n-\mathbf{\br}) |  \leq M \| \mathbf{T}_n-\mathbf{\br} \|_{2}^{3}  \label{eq:tayldivA1}
\end{equation}
whenever $0<\| \mathbf{T}_n-\mathbf{\br} \|_{2} < \delta$. Recall that, by the Rayleigh-Ritz theorem \cite[Th.~4.2.2]{RJ90}, we have
\begin{equation}
 (\mathbf{T}_n-\mathbf{\br})^{\mathsf{T}} \bm{A}_{D,\mathbf{\br}} ( \mathbf{T}_n-\mathbf{\br}) \geq \tilde{\lambda}_{\textnormal{min}} 	\|\mathbf{T}_n-\mathbf{\br}\|^{2}_{2} \label{eq:rrt}
\end{equation}
where $\tilde{\lambda}_{\textnormal{min}} >0$ is the minimum eigenvalue of  $\bm{A}_{D,\mathbf{\br}} $. Combining \eqref{eq:tayldivA1} and \eqref{eq:rrt} then yields 
\begin{equation}
 f_{\br}(\mathbf{T}_n) \geq \tilde{\lambda}_{\textnormal{min}} \|\mathbf{T}_n-\mathbf{\br}\|^{2}_{2} - M \| \mathbf{T}_n-\mathbf{\br} \|_{2}^{3} \geq \|\mathbf{T}_n-\mathbf{\br}\|^{2}_{2} \left(\tilde{\lambda}_{\textnormal{min}} - M\delta\right) \label{eq:rrtc}
\end{equation}
for $0<\| \mathbf{T}_n-\mathbf{\br} \|_{2} < \delta$. Without loss of generality, we can assume that $\delta\leq \tilde{\lambda}_{\textnormal{min}}/(2M)$, since a smaller $\delta$ in \eqref{eq:tayldivA1} is less restrictive. Consequently, $(\tilde{\lambda}_{\textnormal{min}} - M\delta)>0$ and we conclude from \eqref{eq:rrtc} that $f_{\br}(\mathbf{T}_n)=O(b_n)$ implies that $\|\mathbf{T}_n-\mathbf{\br} \|_{2} =O(\sqrt{b_{n}})$. This proves Lemma~\ref{pinskercon}.
\section{Proof of Lemma~\ref{ctaylor}} \label{claimtaylor}
For the probability distributions $T=(T_{1},\ldots, T_{k})^{\mathsf{T}}$ and $Q \in \mathcal{P}(\mathcal{Z})$,  the function
\begin{equation}
	f_{Q}(\mathbf{T}) \triangleq D_{\textnormal{KL}}(T \| Q)
	= \sum_{i=1}^{k-1} T_{i}   \ln    \left( \frac{T_{i}}{ Q_{i}} \right) + \left( 1-\sum_{i=1}^{k-1} T_{i}\right)   \ln    \left( \frac{(1-\sum_{i=1}^{k-1} T_{i})}{ Q_{k}}\right) \label{eq:ts0}
\end{equation}
is a smooth function   of $\mathbf{T}=(T_{1}, \ldots, T_{k-1})^{\mathsf{T}}$. The  partial derivatives of  $f_{Q}(\cdot)$ with respect to $T_{i}, i=1,\ldots, k-1$ up to third orders  are given by
\begin{IEEEeqnarray}{lCl}
	\frac{\partial f_{Q}(\mathbf{T})}{\partial T_{i}}  &=& \ln \left( \frac{T_{i}}{ Q_{i}} \right) -\ln \left( \frac{T_{k}}{ Q_{k}} \right)  \label{eq:ts1}\\
	\frac{\partial^{2} f_{Q}(\mathbf{T})}{\partial T_{i} T_{j}}  &=& 
	\begin{cases}
		\frac{1}{T_{i}} +\frac{1}{T_{k}},  \quad	& \text{if}  \quad i=j\\
		\frac{1}{T_{k}},  \quad 	&   \text{if}  \quad  i \neq j
	\end{cases}
	\label{eq:ts2}\\
	\frac{\partial^{3} f_{Q}(\mathbf{T})}{\partial T_{i} T_{j} T_{\ell}}  &=& 
	\begin{cases}
		\frac{-1}{T_{i}^{2}} +\frac{1}{T_{k}^{2}},  \quad	& \text{if}  \quad i=j=\ell\\
		\frac{1}{T_{k}^{2}},  \quad 	&     \text{otherwise.}
	\end{cases}
	\label{eq:ts3}
\end{IEEEeqnarray}
Consequently,  the second-order Taylor approximation of  $f_{Q}(\mathbf{T})$  around $\mathbf{P}$ is given by
\begin{equation}
	f_{Q}(\mathbf{T}) =  f_{Q}(\mathbf{P})+ \sum_{i=1}^{k-1}\frac{\partial f_{Q}(\mathbf{P})}{\partial T_{i}} (T_{i} -P_{i}) +  \frac{1}{2}\sum_{i,j=1}^{k-1} \frac{\partial^{2} f_{Q}(\mathbf{P})}{\partial T_{i} T_{j}} (T_{i} -P_{i})(T_{j} -P_{j}) +R_{3}(\mathbf{T})
	\label{eq:ts4}
\end{equation}
where the remainder term $R_{3}(\mathbf{T})$  in  Lagrange's form is given  by
\begin{equation}
	R_{3}(\mathbf{T})= \frac{1}{3!} \sum_{i,j,\ell=1}^{k-1} \frac{\partial^{3} f_{Q}(\bm{\Pi})}{\partial T_{i} T_{j} T_{\ell}}  (T_{i} -P_{i})(T_{j} -P_{j})(T_{\ell} -P_{\ell})
	\label{eq:ts5}
\end{equation}
for $\mathbf{\Pi}=\mathbf{P}+t(\mathbf{T} - \mathbf{P})$ and some $t \in (0,1)$. Using the triangle inequality, the absolute value of this term can be upper-bounded as
\begin{equation}
	|R_{3}(\mathbf{T})| \leq  \frac{1}{3!} \sum_{i,j,\ell=1}^{k-1} \left|  \frac{\partial^{3} f_{Q}(\mathbf{\Pi})}{\partial T_{i} T_{j} T_{\ell}}  \right|  \left| (T_{i} -P_{i})(T_{j} -P_{j})(T_{\ell} -P_{\ell}) \right|.
	\label{eq:ts5w2}
\end{equation}

When $\|\mathbf{T} -\mathbf{P} \|_{2} \rightarrow 0$, we have that $P_{i}-\delta<T_{i}$ for an arbitrary $\delta>0$ and every $i=1,\ldots, k-1$. This implies that  $P_{k}-(k-1) \delta <T_{k}$. Choosing
\begin{equation}
    \delta\triangleq\frac{1}{4(k-1)} \min_{1\leq i \leq k} P_{i}
\end{equation}
it follows that $P_{i}-\delta>0, i=1,\ldots, k-1$ and $P_{k}-(k-1) \delta >0$. Consequently, 
\begin{equation}
	\frac{1}{T_{i}^{2}} < \frac{1}{(P_{i}-\delta)^{2}}, \quad  i=1,\ldots, k-1  \label{eq:lbty}
\end{equation}
and 
\begin{equation}
	\frac{1}{T_{k}^{2}} < \frac{1}{(P_{k}-(k-1) \delta)^{2}}. \label{eq:lbtya}
\end{equation}
We then obtain that, for every $\mathbf{T}$ satisfying $ 0 < \|\mathbf{T} -\mathbf{P}\|_{2} < \delta$,
\begin{equation}
	\left| 	\frac{\partial^{3} f_{Q}(\mathbf{T})}{\partial T_{i} T_{j} T_{l}} \right| \leq  C_{1}
\end{equation}
where
\begin{equation}
	C_{1}\triangleq \max_{i}  \frac{1}{(P_{i}-\delta)^{2}} +\frac{1}{(P_{k}-(k-1) \delta)^{2}}.
\end{equation}
Hence, \eqref{eq:ts5w2} can be further upper-bounded as
\begin{IEEEeqnarray}{lCl}
	|R_{3}(\mathbf{T})| &\leq & \frac{1}{3!} C_{1} \sum_{i,j,\ell=1}^{k-1}  \left| (T_{i} -P_{i})(T_{j} -P_{j})(T_{\ell} -P_{\ell}) \right|  \notag \\
	&\leq &	C' \| \mathbf{T} -\mathbf{P}\|_{2}^{3}
	\label{eq:ts5w1e}
\end{IEEEeqnarray}
where $C'\triangleq\frac{1}{3!} C_{1}(k-1)^{3}$. The second inequality in \eqref{eq:ts5w1e} follows since $|T_{i} -P_{i}| \leq \|\mathbf{T} -\mathbf{P}\|_{2},i=1,\ldots, k-1$. 

We next note that, since $T_{k}-P_{k}=\sum_{i=1}^{k-1}(P_{i}-T_{i})$, we obtain from \eqref{eq:ts1} and \eqref{eq:ts2} that
\begin{IEEEeqnarray}{lCl}
	\sum_{i=1}^{k-1}\frac{\partial f_{Q}(\mathbf{P})}{\partial T_{i}} (T_{i} -P_{i}) &=& 
	\sum_{i=1}^{k-1}  \left[ \ln \left( \frac{P_{i}}{ Q_{i}} \right) -\ln \left( \frac{P_{k}}{ Q_{k}} \right) \right]  (T_{i} -P_{i})  \notag \\
	&=& \sum_{i=1}^{k} (T_{i} -P_{i}) \ln \left( \frac{P_{i}}{ Q_{i}} \right)   
	\label{eq:ts8}
\end{IEEEeqnarray}
and 
\begin{IEEEeqnarray}{lCl}
	\sum_{i,j=1}^{k-1} \frac{\partial^{2} f_{Q}(\mathbf{P})}{\partial T_{i} T_{j}} (T_{i} -P_{i})(T_{j} -P_{j}) &=&
	\sum_{i=1}^{k-1} \left( \frac{1}{P_{i}} +\frac{1}{P_{k}}\right)  (T_{i} -P_{i})^{2}+  \sum_{\substack{i,j=1,\ldots,k-1,\\i\neq j}}  \frac{(T_{i} -P_{i})(T_{j} -P_{j})}{P_{k}}  \notag   \\
	&=&
	\sum_{i=1}^{k-1}  \frac{(T_{i} -P_{i})^{2}}{P_{i}} +\sum_{i,j=1 }^{k-1}  \frac{(T_{i} -P_{i})(T_{j} -P_{j})}{P_{k}}  \notag \\
	&=& \sum_{i}^{k}  \frac{(T_{i} -P_{i})^{2}}{P_{i}} \notag \\
	& =& d_{\chi^{2}} (T, P).
	\label{eq:ts9}
\end{IEEEeqnarray}
Combining \eqref{eq:ts5w1e}--\eqref{eq:ts9} with \eqref{eq:ts4}, we thus obtain that 
\begin{equation}
D_{\textnormal{KL}}(T \| Q) =  D(P \| Q) + \sum_{i=1}^{k} (T_{i} -P_{i} ) \ln    \left( \frac{P_{i}}{ Q_{i}} \right)+  \frac{1}{2} d_{\chi^{2}} (T, P) + O(\| \mathbf{T} -\mathbf{P}\|_{2}^{3})  
\end{equation}
as $\|\mathbf{T} -\mathbf{P}\|_{2} \rightarrow 0$. This proves Lemma~\ref{ctaylor}.
\section{Proof of Lemma~\ref{lemmaoptm}} \label{prooflemmaoptm}
Let $\Gamma=(\Gamma_{1}, \ldots, \Gamma_{k})^{\mathsf{T}} \in \mathcal{P}(\mathcal{Z})$ be a probability distribution with coordinate $\mathbf{\Gamma}=(\Gamma_{1}, \ldots, \Gamma_{k-1})^{\mathsf{T}}$.   
Lemma~\ref{lemmaoptm} is tantamount to solving the minimization problem
\begin{equation}
	\left.\begin{aligned}
		&\underset{\mathbf{x}}{\text{minimize}} \; \; \tilde{\ell}(\mathbf{x}) \triangleq \mathbf{c}^{\mathsf{T}} \mathbf{x}\\
		&\text{subject to:}\\[5pt]
		&g_{0}(\mathbf{x})  \triangleq  \mathbf{x}^{\mathsf{T}} \bm{A}_{D,\mathbf{P}} \mathbf{x} - \tilde{r} \leq 0 
	\end{aligned}\,\,\right\} \label{eq:mp2}
\end{equation}
where  $\mathbf{c}=(c_{1},\ldots,c_{k-1})^{\mathsf{T}}$ was defined in \eqref{eq:cidef}, $\bm{A}_{D,\mathbf{P}}$ was defined in \eqref{eq:matrixa}, and $\mathbf{x}=(x_{1},\ldots,x_{k-1})^{\mathsf{T}}$ is a $(k-1)$-dimensional vector with entries $x_{i}=\Gamma_{i}-P_{i}$, $i=1,\ldots,k-1$.

To find the solution of this constrained optimization problem, we consider the Karush-Kuhn-Tucker (KKT) conditions. Since the objective function of the minimization problem  $\tilde{\ell}(\mathbf{x})$ and  the constraint function $g_{0}(\mathbf{x})$  are convex and continuously differentiable in $\mathbb{R}^{k-1}$, the KKT conditions are necessary and sufficient for optimality. Let us first consider the Lagrangian function
\begin{equation}
	\mathcal{L}(\mathbf{x}, \mu_{0}) = \tilde{\ell}(\mathbf{x}) + \mu_{0} g_0(\mathbf{x})
\end{equation}
where  $\mu_{0}$ is a KKT multiplier. Evaluating the KKT conditions, we obtain that
\begin{IEEEeqnarray}{rCl}
	c_{i}+ \mu_{0} \left(  2 \sum_{j=1}^{k-1} \bm{A}_{D,\mathbf{P}}(i,j) x_{j} \right) & = & 0,\quad i=1,\ldots,k-1 \label{eq:kkt1a}\\
	\sum_{i=1}^{k-1} \sum_{j=1}^{k-1} \bm{A}_{D,\mathbf{P}}(i,j) x_{i}x_{j}- \tilde{r} & \leq & 0 \label{eq:kkt2a} \\
	\mu_{0}  &\geq & 0 \label{eq:kkt5a}  \\
	\mu_{0} \left( \sum_{i=1}^{k-1} \sum_{j=1}^{k-1} \bm{A}_{D,\mathbf{P}}(i,j) x_{i}x_{j}- \tilde{r}\right)   & = & 0 \label{eq:kkt8a}
\end{IEEEeqnarray}
where $\bm{A}_{D,\mathbf{P}}(i,j)$ denotes the row-$i$, column-$j$ entry of $\bm{A}_{D,\mathbf{P}}$. We next argue that $\mu_{0} =0$ is not a valid solution. Indeed, suppose that $\mu_{0} =0$.
Then,  \eqref{eq:kkt1a}  yields
$c_{i}=0$ for $i=1,\ldots,k-1$, which in turn only holds if $P=Q$. However, we have $P \neq Q $ by the assumption of  Theorem~\ref{divergence}, so $\mu_{0} =0$ is not valid. In the following, we therefore assume that $\mu_{0} > 0$. From \eqref{eq:kkt1a}, we then obtain that 
\begin{equation}
	\mathbf{x} = \frac{-1}{2 \mu_{0}}  \bm{A}_{D,\mathbf{P}}^{-1} \mathbf{c} \label{eq:kkts4}
\end{equation}
where $\bm{A}_{D,\mathbf{P}}^{-1}$ exists since $\bm{A}_{D,\mathbf{P}} \succ 0$. Furthermore, \eqref{eq:kkt8a} implies that
\begin{equation}
	\sum_{i=1}^{k-1} \sum_{j=1}^{k-1} \bm{A}_{D,\mathbf{P}}(i,j) x_{i}x_{j}- \tilde{r} =0. \label{eq:kkts5}
\end{equation}
Substituting \eqref{eq:kkts4} in  \eqref{eq:kkts5}, and using \eqref{eq:kkt5a}, we get
\begin{equation}
	\mu_{0}= \frac{1}{2 \sqrt{\tilde{r}}}  \sqrt{\mathbf{c}^{\mathsf{T}} \bm{A}_{D,\mathbf{P}}^{-1} \mathbf{c}} \label{eq:kkts7}
\end{equation}
which combined with \eqref{eq:kkts4} yields the optimal solution
\begin{equation}
	\mathbf{x}^{*}= \frac{-\sqrt{\tilde{r}} \bm{A}_{D,\mathbf{P}}^{-1} \mathbf{c}}{ \sqrt{\mathbf{c}^{\mathsf{T}} \bm{A}_{D,\mathbf{P}}^{-1} \mathbf{c}}}.    \label{eq:kkts8}
\end{equation}
It follows that the  probability distribution  $\Gamma^{*}=(\Gamma_{1}^{*}, \ldots, \Gamma_{k}^{*})^{\mathsf{T}}$ that minimizes the function $\ell(\Gamma)$  is given by
\begin{equation}
	\Gamma^{*}_{i} = P_{i} + x_{i}^{*}, \quad i=1,\ldots,k \label{eq:minprob1}
\end{equation}
where 
\begin{IEEEeqnarray}{lCl}
	x_{i}^{*} &=& \frac{-\sqrt{\tilde{r}} b_{i}}{ \sqrt{\mathbf{c}^{\mathsf{T}} \bm{A}_{D,\mathbf{P}}^{-1} \mathbf{c}}}, \quad i=1,\ldots,k-1 \label{eq:kkts9} \\
	x_{k}^{*} &=& -\sum_{i=1}^{k-1}x_{i}^{*}  \label{eq:minprob2}
\end{IEEEeqnarray}
and $b_{i} = (\bm{A}_{D,\mathbf{P}}^{-1} \mathbf{c})_{i}$.

We next argue that $\Gamma^{*}$ is indeed a valid probability distribution. Indeed, since $P$ is a probability distribution and, by \eqref{eq:minprob2}, $\sum_{i=1}^k x_i^*=0$, we have that $\sum_{i=1}^{k}\Gamma_{i}^{*}=1$. It thus remains to show that $\Gamma_{i}^{*}>0$ for $ i=1,\ldots,k$.

First, let us consider the indices $i = 1,\ldots,k-1$. If $b_{i}\leq 0$, then,  by \eqref{eq:kkts9}, $x_{i}^{*}\geq 0$ and we directly obtain that $\Gamma^{*}_{j} = P_{i} + x_i^* >0$. If $b_{i}>0$, then $i \in \mathcal{I}_{+}$.  Since $\sqrt{\tilde{r}} < \frac{\psi}{\tau}  \sqrt{\mathbf{c}^{\mathsf{T}} \bm{A}_{D,\mathbf{P}}^{-1} \mathbf{c}}$ by assumption, we obtain
\begin{IEEEeqnarray}{lCl}
	x_{i}^{*} &=& \frac{- \sqrt{\tilde{r}}b_{i}}{ \sqrt{\mathbf{c}^{\mathsf{T}} \bm{A}_{D,\mathbf{P}}^{-1} \mathbf{c}}} \notag \\
	&	>& \frac{- b_{i}}{ \sqrt{\mathbf{c}^{\mathsf{T}} \bm{A}_{D,\mathbf{P}}^{-1} \mathbf{c}}}
	\frac{\psi}{\tau}  \sqrt{\mathbf{c}^{\mathsf{T}} \bm{A}_{D,\mathbf{P}}^{-1} \mathbf{c}} \notag\\
	&\geq& 	- P_i 
\end{IEEEeqnarray}
where the last step follows because  $\frac{ \psi}{\tau} \leq \frac{ P_{i}}{b_{i}}$ since $\psi \leq P_{i}$ by  \eqref{eq:index3} and $\tau \geq b_{i}$, $i \in \mathcal{I}_{+}$ by \eqref{eq:index3a} and \eqref{eq:index4a}. It follows that $\Gamma^{*}_{j} =  P_{j} + x_{j}^{*} >0$.

Next, we consider $\Gamma^{*}_{k}$.  By \eqref{eq:kkts9} and \eqref{eq:minprob2}, we have 
\begin{equation}
	\Gamma^{*}_{k} = P_{k} + \frac{\sqrt{\tilde{r}} \sum_{i=1}^{k-1}b_{i}}{ \sqrt{\mathbf{c}^{\mathsf{T}} \bm{A}_{D,\mathbf{P}}^{-1} \mathbf{c}}}.  \label{eq:minprob3}
\end{equation} 
If $\sum_{i=1}^{k-1}b_{i} \geq 0$, then $x^*\geq 0$ and we directly obtain that $\Gamma_k^* = P_k + x_k^* > 0$. If $\sum_{i=1}^{k-1}b_{i}<0$, then the assumption  $\sqrt{\tilde{r}} < \frac{\psi}{\tau}  \sqrt{\mathbf{c}^{\mathsf{T}} \bm{A}_{D,\mathbf{P}}^{-1} \mathbf{c}}$ yields that 
\begin{IEEEeqnarray}{lCl}
	\sqrt{\tilde{r}} &<&  \frac{\psi}{\tau}  \sqrt{\mathbf{c}^{\mathsf{T}} \bm{A}_{D,\mathbf{P}}^{-1} \mathbf{c}} \notag\\
	& \leq &  \frac{P_{k}}{-\sum_{i=1}^{k-1}b_{i}}  \sqrt{\mathbf{c}^{\mathsf{T}} \bm{A}_{D,\mathbf{P}}^{-1} \mathbf{c}} 
	\label{eq:minprob4}
\end{IEEEeqnarray} 
where the second inequality follows because  $\tau \geq \tau_{1}=-\sum_{i=1}^{k-1}b_{i}$ by \eqref{eq:index3a} and \eqref{eq:index4}, and because $\psi \leq P_{k}$ by  \eqref{eq:index3}. It follows that
\begin{equation}
	x_{k}^{*} = \frac{\sqrt{\tilde{r}} \sum_{i=1}^{k-1}b_{i}}{ \sqrt{\mathbf{c}^{\mathsf{T}} \bm{A}_{D,\mathbf{P}}^{-1} \mathbf{c}}} >-P_{k}
	\label{eq:minprob5}
\end{equation} 
which implies that $	\Gamma^{*}_{k}  =  P_{k} +x_{k}^{*} >0$.

We conclude the proof of Lemma~\ref{lemmaoptm} by noting that the value of $\ell(\Gamma^{*})$ is given by 
\begin{equation}
\ell(\Gamma^{*})= \tilde{\ell}(\mathbf{x}^*) = \mathbf{c}^{\mathsf{T}} \mathbf{x}^{*}.   
\end{equation}
By \eqref{eq:kkts8}, this evaluates to
\begin{equation}
\ell(\Gamma^{*}) =-\sqrt{\tilde{r}} \sqrt{\mathbf{c}^{\mathsf{T}} \bm{A}_{D,\mathbf{P}}^{-1} \mathbf{c}}
\end{equation}
which is \eqref{eq:min1}.
\section{Proof of Lemma~\ref{type}} \label{lemmatype}
 Consider the probability distribution $\Gamma^{*}=(\Gamma^{*}_{1},\ldots,\Gamma^{*}_{k} )^{\mathsf{T}}$ that minimizes $\ell(\Gamma)$  over $\bar{\mathcal{B}}_{\bm{A}_{D,\mathbf{P}}}  \left( \bar{r}_{n} \right)$, namely,
\begin{equation}
	\Gamma^{*}_{i} = P_{i} + x_{i}^{*}, \quad i=1,\ldots,k \label{eq:minprob1a}
\end{equation} 
where 
\begin{IEEEeqnarray}{lCl}
	x_{i}^{*} &=& \frac{-\sqrt{\bar{r}_{n}} b_{i}}{ \sqrt{\mathbf{c}^{\mathsf{T}} \bm{A}_{D,\mathbf{P}}^{-1} \mathbf{c}}},\quad i=1,\ldots,k-1\label{eq:minprob2b} \\
	x_{k}^{*} &=&-\sum_{i=1}^{k-1}x_{i}^{*}  \label{eq:minprob2a} 
\end{IEEEeqnarray}
and $b_{i} = (\bm{A}_{D,\mathbf{P}}^{-1} \mathbf{c})_{i}$. We shall assume that $n\geq\tilde{N}'$ so that, by \eqref{eq:rnbarn}, $\sqrt{\bar{r}_{n}} < \frac{\psi}{\tau}  \sqrt{\mathbf{c}^{\mathsf{T}} \bm{A}_{D,\mathbf{P}}^{-1} \mathbf{c}}$. Consequently, we have $\Gamma^{*}_{i} >0$, $i=1,\ldots,k$.

To prove the lemma, we need to find a type distribution $\tp^{*}_{n}=(\tp^{*}_{n}(a_{1}), \ldots, \tp^{*}_{n}(a_{k}) )^{\mathsf{T}}$ such that for a sufficiently large $\tilde{N}\geq \tilde{N}'$ and all $n\geq\tilde{N}$ the following is true:
\begin{enumerate}
\item $\tp^{*}_{n} \in 	\bar{\mathcal{B}}_{\bm{A}_{D,\mathbf{P}}}\left( \bar{r}_{n}\right)$, i.e.,	
\begin{equation}
	(\mathbf{\tp}^{*}_{n}-\mathbf{P})^{\mathsf{T}} \bm{A}_{D,\mathbf{P}} ( \mathbf{\tp}^{*}_{n}-\mathbf{P}) \leq \bar{r}_{n} \label{eq:z9}
\end{equation}
where $\mathbf{\tp}^{*}_{n}=(\tp^{*}_{n}(a_{1}), \ldots, \tp^{*}_{n}(a_{k-1}) )^{\mathsf{T}}$ is the coordinate of  $\tp^{*}_{n}$.
\item For every $ i=1, \ldots,k$,  $n \tp^{*}_{n}(a_{i}) $ is a positive  integer and
\begin{equation}
	\sum_{i=1}^{k}n \tp^{*}_{n}(a_{i}) =n. \label{eq:z8}  
\end{equation}
\item $\tp^{*}_{n}$ satisfies 
\begin{equation}
	|n \ell(\Gamma^{*}) -n\ell(\tp^{*}_{n})| \leq \kappa, \quad \text{ for some } \kappa >0. \label{eq:zz}
\end{equation}
\end{enumerate} 
To this end,  let us write the coordinates   $\mathbf{\Gamma}^{*}=(\Gamma^{*}_{1},\ldots,\Gamma^{*}_{k-1} )^{\mathsf{T}}$ of $\Gamma^{*}$ as
\begin{equation}
\mathbf{\Gamma}^{*} = \mathbf{P} +\mathbf{x}^{*}
\label{eq:gamma*}
\end{equation} 
where $\mathbf{x}^{*}=(x_{1}^{*},\ldots,x_{k-1}^{*})^{\mathsf{T}}$ with $x_{i}^{*}, i=1, \ldots,k-1$ defined in \eqref{eq:minprob2b}. Next define, for some $0<\bar{\alpha}<1$ to be specified later,
\begin{equation}
	\mathbf{\bar{\Gamma}}^{*} \triangleq \mathbf{P} +(1-\bar{\alpha})\mathbf{x}^{*}.
	\label{eq:gammabar}
\end{equation} 
We then choose  $\mathbf{\tp}^{*}_{n}=(\tp^{*}_{n}(a_{1}), \ldots, \tp^{*}_{n}(a_{k-1}) )^{\mathsf{T}}$ as follows: 
\begin{equation}
	n \tp^{*}_{n}(a_{i})  =
	\begin{cases}
			\lfloor  n \bar{\Gamma}^{*}_{i} \rfloor, \quad	& \text{if} \;  \langle \bm{A}_{D} \mathbf{x}^{*}, \bm{e}_{i} \rangle >0 \\
		\lceil n \bar{\Gamma}^{*}_{i} \rceil, \quad 	&   \text{if}  \;  \langle \bm{A}_{D} \mathbf{x}^{*}, \bm{e}_{i} \rangle \leq 0
	\end{cases} 
\label{eq:tydef1}
\end{equation}
 where $\lfloor  \cdot \rfloor $ is the floor function; $\lceil \cdot \rceil$ is the ceiling function; $\bm{e}_i=(0,\ldots,0,1,0,\ldots,0)^{\mathsf{T}}$ denotes the standard basis vector in $\mathbb{R}^{k-1}$ whose components are all zero except  at position $i$, where it is one; and $\langle  \cdot, \cdot \rangle $ denotes the dot product in 
$\mathbb{R}^{k-1}$.

We next show that $ \mathbf{\tp}^{*}_{n}$ satisfies the conditions~\eqref{eq:z9}--\eqref{eq:zz}. For ease of exposition, we define the vector
\mbox{$\bm{{\delta}}=(\delta_{1}, \ldots, 	\delta_{k-1})^{\mathsf{T}}$} as
\begin{equation}
\bm{\delta} \triangleq	n \mathbf{\tp}^{*}_{n}-n	\mathbf{\bar{\Gamma}}^{*}. \label{eq:deltadef}
\end{equation} 
It follows immediately from~\eqref{eq:tydef1} that 
\begin{eqnarray}
|\delta_{i} | <1, \quad i=1,\ldots,k-1. \label{eq:deltab}
\end{eqnarray} 
Furthermore,
\begin{IEEEeqnarray}{lCll}
\delta_{i} & \leq & 0, \quad	& \text{if $\langle \bm{A}_{D,\mathbf{P}} \mathbf{x}^{*}, \bm{e}_{i} \rangle  >0$} \label{eq:delta1}\\
	\delta_{i} & \geq &0, \quad 	&  \text{if $\langle  \bm{A}_{D,\mathbf{P}} \mathbf{x}^{*}, \bm{e}_{i} \rangle  \leq 0.$}\label{eq:delta2}
\end{IEEEeqnarray}
By \eqref{eq:gammabar} and \eqref{eq:deltadef}, we can express  $\mathbf{\tp}^{*}_{n}$ as
\begin{IEEEeqnarray}{lCl}
 \mathbf{\tp}^{*}_{n}	&=& \mathbf{P} +(1-\bar{\alpha})\mathbf{x}^{*} +\frac{\bm{\delta}}{n}  \notag\\
&	=& \mathbf{P} +\mathbf{\bar{x}} \label{eq:btypdef2} 
\end{IEEEeqnarray} 
where 
\begin{equation}
\mathbf{\bar{x}} 	\triangleq (1-\bar{\alpha})\mathbf{x}^{*} +\frac{\bm{\delta}}{n}. \label{eq:barxdef1} 
\end{equation} 

\subsection{Proof of  \eqref{eq:z9}}
Consider
 	\begin{eqnarray}
	(\mathbf{\tp}^{*}_{n}-\mathbf{P})^{\mathsf{T}} \bm{A}_{D,\mathbf{P}} ( \mathbf{\tp}^{*}_{n}-\mathbf{P})= 	\mathbf{\bar{x}}^{\mathsf{T}} \bm{A}_{D,\mathbf{P}} \mathbf{\bar{x}}. 
\end{eqnarray}
For convenience, we define
\begin{equation}
\mathbf{\bar{y}} \triangleq (1-\bar{\alpha})\mathbf{x}^{*}\quad \text{and} \quad  \bm{\bar{\delta}} \triangleq \frac{\bm{\delta}}{n}.  \label{eq:barxdef1a} 
\end{equation} 
Then, we have
	\begin{IEEEeqnarray}{lCl}
	\mathbf{\bar{x}}^{\mathsf{T}} \bm{A}_{D,\mathbf{P}} \mathbf{\bar{x}} &=& (\mathbf{\bar{y}} + \bm{\bar{\delta}})^{\mathsf{T}}
	 \bm{A}_{D,\mathbf{P}} (\mathbf{\bar{y}} + \bm{\bar{\delta}}) \notag \\
	 &=&\mathbf{\bar{y}}^{\mathsf{T}}\bm{A}_{D,\mathbf{P}} \mathbf{\bar{y}}+ 2 \bm{\bar{\delta}}^{\mathsf{T}}\bm{A}_{D,\mathbf{P}}\mathbf{\bar{y}}+  \bm{\bar{\delta}}^{\mathsf{T}}\bm{A}_{D,\mathbf{P}} \bm{\bar{\delta}} \label{eq:inp}
	 \end{IEEEeqnarray}
 where we have used that $\bm{A}_{D,\mathbf{P}}$ is a symmetric matrix. After some algebraic manipulations, it can be shown that $(\mathbf{x}^{*})^{\mathsf{T}} \bm{A}_{D,\mathbf{P}} \mathbf{x}^{*}=\bar{r}_{n}$. Consequently,
\begin{equation}
	\mathbf{\bar{y}}^{\mathsf{T}}\bm{A}_{D,\mathbf{P}} \mathbf{\bar{y}} 
	=(1-\bar{\alpha})^{2} \bar{r}_{n}. \label{eq:inp1}
\end{equation}
Furthermore, using the Rayleigh-Ritz theorem \cite[Th.~4.2.2]{RJ90}, we can upper-bound the third term  on the right-hand side of~\eqref{eq:inp} by
\begin{equation}
 \bm{\bar{\delta}}^{\mathsf{T}}\bm{A}_{D,\mathbf{P}} \bm{\bar{\delta}} \leq \tilde{\lambda}_{\textnormal{max}} \|  \bm{\bar{\delta}}\|_{2}^{2} \label{eq:inp2}
\end{equation}
where $\tilde{\lambda}_{\textnormal{max}}$ is the maximum eigenvalue of $\bm{A}_{D,\mathbf{P}}$.
Finally, writing $\bm{{\delta}}=\sum_{i=1}^{k-1}\delta_{i} \bm{e}_{i}$, and using \eqref{eq:barxdef1a}, we upper-bound the second term on the right-hand side of~\eqref{eq:inp} by
\begin{IEEEeqnarray}{lCl}
2\bm{\bar{\delta}}^{\mathsf{T}}\bm{A}_{D,\mathbf{P}}\mathbf{\bar{y}} &=& 2 \langle  \bm{A}_{D,\mathbf{P}} \mathbf{\bar{y}}, \bm{\bar{\delta}} \rangle \notag\\
&=& 2\left\langle (1-\bar{\alpha})\bm{A}_{D,\mathbf{P}} \mathbf{x}^{*} , \frac{1}{n}\sum_{i=1}^{k-1}\delta_{i} \bm{e}_{i} \right\rangle  \notag \\
&=& \frac{2 (1-\bar{\alpha})}{n} \sum_{i=1}^{k-1} \delta_{i} \langle  \bm{A}_{D,\mathbf{P}} \mathbf{x}^{*} , \bm{e}_{i} \rangle  \notag\\
&\leq & 0 \label{eq:inp3}
\end{IEEEeqnarray}
where the last inequality follows from \eqref{eq:delta1}--\eqref{eq:delta2} and because we assumed that $0<\bar{\alpha}<1$.

Combining \eqref{eq:inp1}--\eqref{eq:inp3} with \eqref{eq:inp}, we obtain
\begin{equation}
\mathbf{\bar{x}}^{\mathsf{T}} \bm{A}_{D,\mathbf{P}} \mathbf{\bar{x}} 
 \leq (1-\bar{\alpha})^{2} \bar{r}_{n} + \tilde{\lambda}_{\textnormal{max}} \|  \bm{\bar{\delta}}\|_{2}^{2}. \label{eq:inp5}
\end{equation}
Furthermore, \eqref{eq:deltab} and \eqref{eq:barxdef1a} yield that
\begin{equation}
\|  \bm{\bar{\delta}}\|_{2}^{2} = \frac{\|  \bm{\delta}\|_{2}^{2} }{n^{2}}
< \frac{(k-1)}{n^{2}}. \label{eq:inp6}
\end{equation}
Consequently,
\begin{IEEEeqnarray}{lCl}
\mathbf{\bar{x}}^{\mathsf{T}} \bm{A}_{D,\mathbf{P}} \mathbf{\bar{x}}  &<&(1-\bar{\alpha})^{2} \bar{r}_{n} +  \frac{\tilde{\lambda}_{\textnormal{max}} (k-1)}{n^{2}} \notag \\
&=& \bar{r}_{n}+(\bar{\alpha}^{2}\bar{r}_{n} -\bar{\alpha} \bar{r}_{n}) +\left(  \frac{\tilde{\lambda}_{\textnormal{max}} (k-1)}{n^{2}}  -\bar{\alpha} \bar{r}_{n}\right).\label{eq:inp7}
\end{IEEEeqnarray}

From \eqref{eq:sec6} and \eqref{eq:rbar}, we have that $\bar{r}_{n}=\Theta(\frac{1}{n})$, so  there exist $0<\check{M}_{1} \leq \check{M}_{2} < \infty$ and $\check{N}_{1} \in \mathbb{N}$ such that
\begin{equation}
\frac{\check{M}_{1}}{ n} \leq \bar{r}_{n} \leq  \frac{\check{M}_{2}}{ n}, \quad n \geq \check{N}_{1}. \label{eq:rno1}
\end{equation}
We choose
\begin{equation}
\bar{\alpha}= \frac{2\tilde{\lambda}_{\textnormal{max}} (k-1)}{\check{M}_{1} n} \label{eq:rno2}
	\end{equation}
which is positive and vanishes as $n$ tends to infinity, hence it satisfies $0< \bar{\alpha} <1$ for all $n \geq  \check{N}_{2} $ and a sufficiently large $\check{N}_{2} \geq \check{N}_{1}$. For this choice of $\bar{\alpha}$, we have that
\begin{equation}
	\left( 	\frac{\tilde{\lambda}_{\textnormal{max}} (k-1)}{n^{2}}  -\bar{\alpha} \bar{r}_{n}\right)  \leq -\frac{\tilde{\lambda}_{\textnormal{max}} (k-1)}{n^{2}} <0, \quad n \geq \check{N}_{2}. \label{eq:rno4}
\end{equation}
Furthermore, 
\begin{equation}
	\bar{\alpha}^{2} \bar{r}_{n}= \Theta \left( \frac{1}{n^{3}}\right) \quad \text{and} \quad  \bar{\alpha}\bar{r}_{n}= \Theta \left( \frac{1}{n^{2}}\right).
\end{equation}
This implies that there  exists an $\check{N}_{3} \in \mathbb{N}$ such that
\begin{equation}
	\bar{\alpha}^{2} \bar{r}_{n}-\bar{\alpha}\bar{r}_{n} <0, \quad n \geq \check{N}_{3}. \label{eq:rno5}
\end{equation}
Applying \eqref{eq:rno4} and  \eqref{eq:rno5} to \eqref{eq:inp7}, we obtain that, for $n \geq \max\left\lbrace  \check{N}_{2},  \check{N}_{3} \right\rbrace $,  
\begin{IEEEeqnarray}{lCl}
	\mathbf{\bar{x}}^{\mathsf{T}} \bm{A}_{D,\mathbf{P}} \mathbf{\bar{x}}  &< &\bar{r}_{n}+(\bar{\alpha}^{2}\bar{r}_{n} -\bar{\alpha} \bar{r}_{n}) +\left(  \frac{\tilde{\lambda}_{\textnormal{max}} (k-1)}{n^{2}}  -\bar{\alpha} \bar{r}_{n}\right) \notag \\
	&\leq& \bar{r}_{n}.\label{eq:inp8}
\end{IEEEeqnarray}
This proves that $\mathbf{\tp}^{*}_{n}$ satisfies \eqref{eq:z9} for sufficiently large $n$.
	
\subsection{Proof of \eqref{eq:z8}}
For $i=1,\ldots,k-1$, $n \tp^{*}_{n}(a_{i}) $  is clearly an integer, since it is obtained by applying the floor or ceiling function to $n\bar{\Gamma}^{*}_{i}$; cf.~\eqref{eq:tydef1}. We next show that $n \tp^{*}_{n}(a_{i}) $ is also positive. To this end, we demonstrate that $\bar{\Gamma}^{*}_{i}>1/n$ for sufficiently large $n$, from which we directly obtain that $n \tp^{*}_{n}(a_{i}) \geq \lfloor n\bar{\Gamma}^{*}_{i}\rfloor \geq 1$.

Indeed, by \eqref{eq:gammabar}, we have that
\begin{equation}
	\bar{\Gamma}^{*}_{i} = P_{i} + x^{*}_{i}-\bar{\alpha} x^{*}_{i}, \quad i=1,\ldots,k-1. \label{eq:gammabar1}
\end{equation} 
If $x^{*}_{i}=0$, then $\bar{\Gamma}^{*}_{i}  = P_{i}>0$. It follows that there exists an $\bar{N}_1$ such that $\bar{\Gamma}^{*}_{i}  > 1/n$ for all $n\geq \bar{N}_1$. If $x^{*}_{i}<0$, then $-\bar{\alpha} x^{*}_{i}>0$. Since $\Gamma^{*}_{i} >0$ for $n \geq \tilde{N}_1$, we then obtain that
\begin{IEEEeqnarray}{lCl}
	\bar{\Gamma}^{*}_{i} & = & P_{i} + x^{*}_{i}-\bar{\alpha}x^{*}_{i} \notag\\
	&=&\Gamma^{*}_{i} -\bar{\alpha} x^{*}_{i} \notag\\
	& >&\frac{1}{n}\label{eq:gammabar2}
\end{IEEEeqnarray}
for all $n\geq \bar{N}_2$ and some sufficiently large $\bar{N}_2$. If $x^{*}_{i}>0$, then $\bar{\alpha} x^{*}_{i}>0$ and, by \eqref{eq:minprob2b},
\begin{equation}
\frac{b_{i}}{ \sqrt{\mathbf{c}^{\mathsf{T}} \bm{A}_{D,\mathbf{P}}^{-1} \mathbf{c}}}<0.
\end{equation}
Since $\bar{r}_{n}=\Theta(\frac{1}{n})$, this implies that $x^{*}_{i}=\Theta(1/\sqrt{n})$.  Similarly, $\bar{\alpha}$ in \eqref{eq:rno2} satisfies $\bar{\alpha}=\Theta(1/n)$. It follows that  $x^{*}_{i}-\bar{\alpha}x^{*}_{i}=\Theta(1/\sqrt{n}) -\Theta(1/n^{3/2})$, so there exists an $\bar{N}_3$ such that $x^{*}_{i}-\bar{\alpha} x^{*}_{i}>1/n$  for $n \geq \bar{N}_3$. Consequently,
\begin{equation}
	\bar{\Gamma}^{*}_{i}  =  P_{i} + (x^{*}_{i}-\bar{\alpha} x^{*}_{i}) >\frac{1}{n}, \quad \quad n \geq \bar{N}_3.
\end{equation} 
We conclude that $\bar{\Gamma}^{*}_{i} > 1/n$ for all $ n \geq \max\lbrace \bar{N}_1, \bar{N}_2,\bar{N}_2 \rbrace $, hence $n \tp^{*}_{n}(a_{i})$, $i=1,\ldots,k-1$ is a positive integer.

We next show that
\begin{equation}
\sum_{i=1}^{k-1} n \tp^{*}_{n}(a_{i}) <n. \label{eq:not_n}
\end{equation}
In this case
\begin{equation}
	n \tp^{*}_{n}(a_{k}) =n- \sum_{i=1}^{k-1} n \tp^{*}_{n}(a_{i})  \label{eq:typek}
\end{equation}
is a positive integer and $\tp^{*}_{n}=(\tp^{*}_{n}(a_{1}), \ldots, \tp^{*}_{n}(a_{k}) )^{\mathsf{T}}$ satisfies \eqref{eq:z8} as desired. To prove \eqref{eq:not_n}, we use \eqref{eq:minprob2a} and \eqref{eq:btypdef2} to express the left-hand side of \eqref{eq:not_n} as
\begin{IEEEeqnarray}{lCl}
	\sum_{i=1}^{k-1} n \tp^{*}_{n}(a_{i}) 
	&=&\sum_{i=1}^{k-1} nP_{i} + \sum_{i=1}^{k-1}n(1-\bar{\alpha}) x_{i}^{*} + \sum_{i=1}^{k-1}\delta_{i} \notag \\
	&=&n-nP_{k}  -n x_{k}^{*} + n \bar{\alpha} x_{k}^{*} + \sum_{i=1}^{k-1}\delta_{i}. \label{eq:sumn1}
\end{IEEEeqnarray} 
 Using \eqref{eq:deltab}, this can be upper-bounded as
 \begin{equation}
 	\sum_{i=1}^{k-1} n \tp^{*}_{n}(a_{i}) 	< n-nP_{k}  -n x_{k}^{*} + n \bar{\alpha}_{n} x_{k}^{*} + (k-1). \label{eq:sumn4}
 \end{equation} 
Recall that $\bar{\alpha}=\Theta(1/n)$. Further note that \eqref{eq:minprob2b} and \eqref{eq:minprob2a} imply that
 \begin{equation}
x^{*}_{k}	=O\left(\frac{1}{\sqrt{n}}\right)\label{eq:sumn4a}
\end{equation} 
since $ \bar{r}_{n}=\Theta(1/n)$. Consequently, $ -n x_{k}^{*} + n \bar{\alpha} x_{k}^{*} + (k-1) =O(\sqrt{n})$. Since $nP_{k}=\Theta(n)$, we can therefore find an $\bar{N}_4$ such that
\begin{equation}
-nP_{k}  -n x_{k}^{*} + n \bar{\alpha} x_{k}^{*} + (k-1)<0, \quad n \geq \bar{N}_4.
\end{equation}
Thus, \eqref{eq:not_n} follows, and the type distribution $\tp^{*}_{n}=(\tp^{*}_{n}(a_{1}), \ldots, \tp^{*}_{n}(a_{k}) )^{\mathsf{T}}$ with $n \tp^{*}_{n}(a_{k})$ as indicated in \eqref{eq:typek} satisfies \eqref{eq:z8}.

\subsection{Proof of \eqref{eq:zz}}
It follows from \eqref{eq:hfun1} and  \eqref{eq:barxdef1} that 
\begin{IEEEeqnarray}{lCl}
	|n \ell(\Gamma^{*}) -n\ell(\tp^{*}_{n})| &=& \left| \sum_{i=1}^{k-1} nc_{i}x_{i}^{*}-\sum_{i=1}^{k-1}nc_{i} \bar{x}_{i} \right|  \notag \\
	&=& \left| \sum_{i=1}^{k-1}  c_{i} (n\bar{\alpha} x^{*}_{i} -\delta_{i})  \right| \notag  \\
	& \leq & \sum_{i=1}^{k-1}\bigl(|c_{i}| | n\bar{\alpha} x^{*}_{i}|+ | \delta_{i} |  |c_{i}|\bigr).  \label{eq:inp9}
\end{IEEEeqnarray}
Note that
\begin{equation}
	n \bar{\alpha} x^{*}_{i}=O\left( \frac{1}{\sqrt{n}}\right), \quad i=1,\ldots,k-1
\end{equation}
since $\bar{\alpha}=\Theta(1/n)$, and, by applying $\bar{r}_{n}=\Theta(1/n)$ to \eqref{eq:minprob2b}, $x^*_i = O(1/\sqrt{n})$. This implies that there exist $\widehat{M} >0$ and $\widehat{N} \in \mathbb{N}$ such that
\begin{equation}
	|n\bar{\alpha}x^{*}_{i}| \leq \widehat{M}, \quad n \geq \widehat{N}. \label{eq:inp10}
\end{equation}
Together with \eqref{eq:deltab}, this implies that \eqref{eq:inp9} can be further upper-bounded as
\begin{IEEEeqnarray}{lCl}
	|n \ell(\Gamma^{*}) -n\ell(\tp^{*}_{n})| &\leq &( \widehat{M}+1)\sum_{i=1}^{k-1}|c_{i}| \notag \\
	&\triangleq& \kappa, \quad n\geq \widehat{N}
\end{IEEEeqnarray}
where $\kappa$ is a positive constant. This proves \eqref{eq:zz}.

Having demonstrated that the type distribution $\tp^{*}_{n}$, as defined in \eqref{eq:tydef1}, satisfies \eqref{eq:z9}--\eqref{eq:zz}, we obtain Lemma \ref{type}.

\section*{Acknowledgment}
\addcontentsline{toc}{section}{Acknowledgment}
The authors would like to thank Shun Watanabe for his helpful comments during the initial phase of this work and for pointing out relevant literature.

\bibliography{Bibliography.bib}
\bibliographystyle{IEEEtran}

\end{document}